%% file: main.tex
\documentclass[11pt,letterpaper,onecolumn]{article}
\usepackage[
  tocappendix=false,
  tocsubsubsections=false,
  tocreferences=true
]{fullversion}
\input{macros}

\title{Reachability in Five-Dimensional Vector Addition Systems Is PSPACE-Complete}

\author{Weijun Chen}
\email{cwj2018@sjtu.edu.cn}
\orcid{0009-0007-2611-0338}
\affiliation{%
  \institution{BASICS, \\ Shanghai Jiao Tong University}
  \city{Shanghai}
  \country{China}
}

\author{Bo Fu}
\email{fubo0970@sjtu.edu.cn}
\orcid{0009-0002-6513-8935}
\affiliation{%
  \institution{BASICS, \\ Shanghai Jiao Tong University}
  \city{Shanghai}
  \country{China}
}

\author{Yuxi Fu\thanks{Corresponding author}}
\email{fu-yx@cs.sjtu.edu.cn}
\orcid{0000-0001-6429-7550}
\affiliation{%
  \institution{BASICS, \\ Shanghai Jiao Tong University}
  \city{Shanghai}
  \country{China}
}


\begin{document}
\makefrontmatter{\input{contents/0-Abstract}}

\input{contents/1-Introduction}
\input{contents/2-Overview}
\input{contents/3-Preliminary}
\input{contents/4-Quantitative}

\input{contents/5-Flatness}
\input{contents/6-Coverable}
\input{contents/7-Semicoverable}
\input{contents/8-Noncoverable}

\paragraph{AI Methodology.}
\addcontentsline{toc}{section}{AI Methodology}
The main mathematical ideas and proof strategies were developed by the authors.
OpenAI's ChatGPT was used to assist in checking and refining detailed technical arguments, particularly in~\cref{sec:removing_coverability} and~\cref{app:arbitrary_target_proofs}.
It was also used to prepare, audit, and revise the manuscript.
The authors have verified all content, including references, and take full responsibility for the manuscript.

\clearpage
\appendix
\crefalias{section}{appendix}
\crefalias{subsection}{appendix}
\crefalias{subsubsection}{appendix}
\input{appendix/Appendix-A-Overview}
\input{appendix/Appendix-B-Transformation-Area}
\input{appendix/Appendix-C-Factorization}
\input{appendix/Appendix-D-Basic}
\input{appendix/Appendix-E-Induction}
\input{appendix/Appendix-F-general}

\bibliography{contents/9-Reference}
\end{document}

%% file: macros.tex
\newcommand{\defproblem}[3]{
    \;\par\vspace{0.5em}
    \noindent\fbox{
        \begin{minipage}{0.96\columnwidth}
            \underline{#1}
            \begin{description}
                \item[Input:] {#2}
                \item[Question:] {#3}
            \end{description}
        \end{minipage}
    }\vspace{0.5em}
    \par\;
}

\newcommand{\OpName}[1]{\mathop{\operatorname{\textup{#1}}}\nolimits}
\newcommand{\OpNameSC}[1]{\mathop{\operatorname{\textup{\textsc{#1}}}}\nolimits}

\renewcommand{\Vec}[1]{\bm{#1}}
\newcommand{\norm}[1]{\left\Vert{#1}\right\Vert}

\newcommand{\Supp}{\OpName{supp}}
\newcommand{\CycleSpace}{\OpNameSC{cyc}}
\newcommand{\Len}{\OpNameSC{len}}
\newcommand{\Span}{\OpName{span}}
\newcommand{\Drop}{\OpNameSC{drop}}
\newcommand{\Cone}{\OpName{cone}}
\newcommand{\Rank}{\operatorname{rank}\nolimits}
\newcommand{\PolyF}[1]{{\normalfont\textsf{poly}({#1})}}

\newcommand{\dimcyc}{\dim_{\textup{cyc}}}
\newcommand{\dimcom}{\dim_{\textup{com}}}
\newcommand{\Reach}{\OpName{Reach}}
\newcommand{\Size}{\OpNameSC{size}}

\newcommand{\PSPACE}{\ensuremath{\mathsf{PSPACE}}\xspace}
\newcommand{\NPSPACE}{\ensuremath{\mathsf{NPSPACE}}\xspace}
\newcommand{\NP}{\ensuremath{\mathsf{NP}}\xspace}
\newcommand{\NL}{\ensuremath{\mathsf{NL}}\xspace}
\newcommand{\EXPSPACE}{\ensuremath{\mathsf{EXPSPACE}}\xspace}
\newcommand{\NEXPT}{\ensuremath{\mathsf{NEXPTIME}}\xspace}

\newcommand{\Inner}[1]{\langle {#1} \rangle}
\newcommand{\NN}{\ensuremath{\mathbb{N}}}
\newcommand{\ZZ}{\ensuremath{\mathbb{Z}}}
\newcommand{\QQ}{\ensuremath{\mathbb{Q}}}
\newcommand{\XX}{\ensuremath{\mathbb{X}}}
\newcommand{\Eff}{\ensuremath{\OpNameSC{eff}}}
\newcommand{\Parikh}[1]{\boldsymbol{\psi}^{#1}}
\newcommand{\MinLen}{\ell}
\newcommand{\TA}{\mathcal T}
\newcommand{\bigO}{\mathcal{O}}
\newcommand{\restr}[2]{#1\!\upharpoonright_{#2}}
\newcommand{\enc}[1]{[\![#1]\!]}
\newcommand{\GTA}{\mathcal G}
\newcommand{\hnorm}{\OpNameSC{norm}}
\newcommand{\rev}{\OpNameSC{rev}}
\newcommand{\BDCounters}{\OpNameSC{bd}}

%% file: contents/0-Abstract.tex
Vector addition systems (VAS) are a fundamental model of concurrent computation in which transitions update nonnegative integer counters by adding integer vectors, provided all counters remain nonnegative.
Closely related to Petri nets, VAS are the single-state counterpart of vector addition systems with states (VASS), which use finitely many states to restrict the available transitions.
A classical simulation of $d$-VASS by $(d+3)$-VAS connects their reachability problems, but the complexity-theoretic picture of fixed-dimensional VAS remains largely unclear.
While reachability in $2$-VASS is \PSPACE{}-complete under binary encoding, the best known upper bound for $5$-VAS was only elementary.
We prove that every reachable $5$-VAS instance has a reachability witness of length polynomial in its unary size.
Consequently, $5$-VAS reachability is \NL{}-complete under unary encoding and \PSPACE{}-complete under binary encoding.
We also establish a polynomial factorization of the reachability relation of every $d$-VAS with $d \ge 3$.
For $5$-VAS, this strengthens Hopcroft and Pansiot's effective semilinearity of reachability sets to the full reachability relation.
For every fixed $d\ge3$, it also yields an $\mathsf F_d$ upper bound for reachability in $(d+3)$-VAS.

%% file: contents/1-Introduction.tex
\section{Introduction}
\label{sec:introduction}

Vector addition systems (VAS), introduced by Karp and Miller in the 1960s~\cite{DBLP:journals/jcss/KarpM69}, are a fundamental model of concurrent computation.
A $d$-dimensional VAS models $d$ types of resources, each represented by a nonnegative integer counter, and a finite set of rules called transitions.
A configuration is a nonnegative integer vector recording the values of these counters.
Each transition updates a configuration by adding a fixed integer vector, and can be applied only if the resulting counters remain nonnegative.
VAS are essentially equivalent in expressive power to Petri nets, widely used to model protocols, workflows, and chemical and biological processes~\cite{DBLP:conf/ac/Esparza96,DBLP:journals/pieee/Murata89,DBLP:journals/tie/ZurawskiZ94}.

The central decision problem for VAS is \emph{reachability}: given a VAS $V$ and two configurations $\Vec{s}$ and $\Vec{t}$, determine whether a \emph{run}, i.e., a finite sequence of transitions executable from $\Vec{s}$, reaches $\Vec{t}$.
Its decidability was established by Mayr in 1981~\cite{DBLP:conf/stoc/Mayr81} with subsequent refinements by Kosaraju~\cite{DBLP:conf/stoc/Kosaraju82} and Lambert~\cite{DBLP:journals/tcs/Lambert92}, leading to the KLM algorithm.
Its precise computational complexity remained open for decades.
Leroux and Schmitz established an Ackermannian upper bound in 2019~\cite{DBLP:conf/lics/LerouxS19}, while Czerwiński and Orlikowski~\cite{DBLP:conf/focs/CzerwinskiO21} and, independently, Leroux~\cite{DBLP:conf/focs/Leroux21} established matching lower bounds in 2021, settling reachability as Ackermann-complete.

Much of the complexity theory of reachability has been developed for a closely related model called \emph{vector addition systems with states} (VASS), which extend VAS with finitely many states restricting the available transitions.
While every VAS is a VASS with a single state, Hopcroft and Pansiot showed that every $d$-VASS can be simulated by a $(d+3)$-VAS~\cite{DBLP:journals/tcs/HopcroftP79}.
Thus, the two models are equivalent when dimension is unfixed, but the three additional counters make the fixed-dimensional complexity comparison nontrivial.
For VASS, the exact complexity of reachability under binary encoding is known only in dimensions at most two: \NP{}-complete for $1$-VASS~\cite{DBLP:conf/concur/HaaseKOW09} and \PSPACE{}-complete for $2$-VASS~\cite{DBLP:journals/jacm/BlondinEFGHLMT21}.
Only recently was reachability in $3$-VASS shown to be elementary~\cite{DBLP:conf/icalp/CzerwinskiJ0O25}, with its upper bound later improved to \EXPSPACE{}~\cite{DBLP:journals/corr/abs-2607-14983}.

These advances for low-dimensional VASS make substantial use of the structure of two-dimensional reachability sets and relations.
In particular, for a fixed source configuration, its \emph{reachability set} consists of all reachable targets, whereas the \emph{reachability relation} consists of all reachable source-target pairs.
A set or relation is \emph{semilinear} if it is a finite union of sets obtained from a base vector by adding nonnegative integer combinations of finitely many period vectors.
Hopcroft and Pansiot established that the reachability sets of $2$-VASS are effectively semilinear~\cite{DBLP:journals/tcs/HopcroftP79}, while Leroux and Sutre proved that even the full reachability relation is effectively semilinear~\cite{DBLP:conf/concur/LerouxS04}.
Remarkably, Hopcroft and Pansiot also proved that the reachability sets of $5$-VAS are effectively semilinear~\cite{DBLP:journals/tcs/HopcroftP79}.
This parallel raises two natural questions: firstly, is the reachability relation of $5$-VAS semilinear, and secondly, is reachability in $5$-VAS \PSPACE{}-complete under binary encoding, as in $2$-VASS?
The latter problem is already \PSPACE{}-hard by the three-counter simulation, whereas only an elementary upper bound was established recently~\cite{DBLP:conf/concur/ChenFZ26}.

The relationship between VAS and VASS also motivates improvements in higher dimensions.
For $d$-VASS, reachability is known to be in $\mathsf{F}_d$~\cite{DBLP:conf/icalp/FuYZ24}, the $d$-th level of the fast-growing complexity hierarchy.
In particular, $\mathsf{F}_3$ corresponds to tower-of-exponentials complexity, while $\mathsf{F}_\omega$ corresponds to Ackermannian complexity~\cite{DBLP:journals/toct/Schmitz16}.
Our current knowledge about lower bounds is that $(2k+3)$-VASS is $\mathsf{F}_k$-hard~\cite{DBLP:conf/fsttcs/CzerwinskiJ0LO23}.
Huge gaps remain between the known upper and lower bounds.
For a long time, $d$-VAS inherited the same upper bound, until Chen, Fu, and Zheng improved it to $\mathsf{F}_{d-2}$~\cite{DBLP:conf/concur/ChenFZ26} by exploiting the absence of states.
In view of the three-counter simulation, one may ask whether reachability in $(d+3)$-VAS admits the same $\mathsf{F}_d$ upper bound as reachability in $d$-VASS.

\subsection{Our Results}
\label{subsec:our_results}

\paragraph{Reachability Complexity.}
\label{par:intro_reachability_complexity}
Our main result establishes the exact complexity of reachability in $5$-VAS under both binary and unary encoding, where encoding schemes refer to how the integers specifying the transitions and the configurations are represented.

\begin{restatable}{theorem}{mainComplexity}
    \label{thm:main_complexity}
	The reachability problem for $5$-VAS is $\PSPACE$-complete under binary encoding and $\NL$-complete under unary encoding.
\end{restatable}

These results match those for $2$-VASS, whose reachability is also \PSPACE{}-complete under binary encoding and \NL{}-complete under unary encoding~\cite{DBLP:journals/jacm/BlondinEFGHLMT21}.
Recently, Kami\'nski and Lasota established \PSPACE{}-hardness~\cite{DBLP:conf/rp/KaminskiL26} for $3$-VAS, which, together with the known upper bounds~\cite{DBLP:conf/concur/ChenFZ26}, settled the binary complexity of reachability in dimensions three and four.
Combined with our result, this establishes the same \PSPACE{} (binary) and \NL{} (unary) completeness for every $d\in\{3,4,5\}$.

To prove the upper bounds in \cref{thm:main_complexity}, we bound the length of shortest runs.
The \emph{reachability distance} $\MinLen_V(\Vec{s},\Vec{t})$ is the length of a shortest run from $\Vec{s}$ to $\Vec{t}$ in $V$ if such a run exists.
We write $\Size(V,\Vec{s},\Vec{t})$ for the \emph{unary size} of the instance, which accounts for the number of transitions, the norms of their effects, and the values of the source and target configurations.

\begin{restatable}[Polynomial Reachability Distance]{theorem}{lenFiveVAS}
	\label{thm:polynomial_witness}
	There exists a constant $c$ such that for every reachable $5$-VAS instance $(V,\Vec s,\Vec t)$, one has $\MinLen_V(\Vec s,\Vec t)\le\Size(V,\Vec s,\Vec t)^c$.
\end{restatable}

A nondeterministic algorithm can guess a run of length at most $M^c$ using $\bigO(\log M)$ space to store the current configuration and a step counter, where $M=\Size(V,\Vec{s},\Vec{t})$.
This yields \NL{} membership under unary encoding and \PSPACE{} membership under binary encoding.

\paragraph{Factorization of Reachability Relations.}
\label{par:intro_factorization}

\cref{thm:polynomial_witness} bounds the length of a witness for each reachable source-target pair, but does not describe the structure of the whole reachability relation.
We develop a polynomial factorization of the reachability relation of every $d$-VAS with $d\ge 3$.
A \emph{linear path scheme} (LPS) $\Lambda=\alpha_{0}\beta_{1}^{*}\alpha_{1}\cdots\beta_{k}^{*}\alpha_{k}$ consists of fixed paths $\alpha_0,\ldots, \alpha_k$ interleaved with repeatable cycles $\beta_1,\ldots,\beta_k$, providing a standard representation of reachability in low-dimensional VASS~\cite{DBLP:journals/jacm/BlondinEFGHLMT21}.
We extend them to \emph{mixed linear path schemes} (mixed LPS), which allow different components to be interpreted under either ordinary nonnegative semantics or integer semantics.

We call a representation of the reachability relation of a VAS $V$ of the form
\begin{equation*}
    \Reach_V = \bigcup_{(R_1,\ldots,R_m)\in\mathcal F}R_1\circ\cdots\circ R_m
\end{equation*}
a \emph{polynomial factorization} if $\mathcal F$ is finite and each $R_i$ is either represented by a mixed LPS or is the reachability relation of $V$ with three counters bounded by some $B$.
The latter is essentially captured by a $(d-3)$-VASS, obtained by encoding these bounded counters into states.
The length $m$ of each factorization, the bounds $B$, and all mixed-LPS parameters must be bounded by a polynomial in $\Size(V)$.
Moreover, every witness for a pair in $R_i$ lifts to a run in $V$ of no greater length.

\begin{theorem}
    \label{thm:intro_factorization}
    For every $d\ge3$, the reachability relation of every $d$-VAS admits a polynomial factorization.
\end{theorem}

For $5$-VAS, the factors with three bounded counters can be reduced to $2$-VASS by encoding these counters into states.
Since the reachability relation of $2$-VASS can be captured by LPS~\cite{DBLP:journals/jacm/BlondinEFGHLMT21}, these factors can also be represented by LPS with polynomially bounded parameters.

\begin{restatable}{corollary}{flatnessFiveVAS}
    \label{cor:flatness_of_5vas}
    For every $5$-VAS $V$, its reachability relation can be represented by a finite set of mixed LPS, with all parameters polynomially bounded in $\Size(V)$.
\end{restatable}

Reachability through a mixed LPS can be described by linear inequalities over integers, with variables describing cycle repetitions and intermediate configurations.
Thus, mixed LPS define semilinear relations, while retaining a path-and-cycle structure that is more convenient for VASS constructions.
Our corollary therefore establishes semilinearity of the reachability relation of $5$-VAS, strengthening the classical result on reachability sets with a fixed source~\cite{DBLP:journals/tcs/HopcroftP79}.

\paragraph{Higher-Dimensional Complexity.}
\label{par:intro_higher_dimensional_complexity}
\cref{thm:intro_factorization} also yields an improved upper bound for reachability in higher-dimensional VAS.
For a $(d+3)$-VAS, the factors can be simulated by VASS whose strongly connected components have at most $d$ independent cycle effects, even though these VASS may have more than $d$ counters.
Applying the component-dimension reachability bound of~\cite{DBLP:conf/concur/ChenFZ26}, which builds on the improved KLM algorithm~\cite{DBLP:conf/icalp/FuYZ24}, yields the following result.

\begin{restatable}{theorem}{dVASReachInFdmThree}
\label{thm:higher_dimensional_complexity}
Reachability in $(d+3)$-VAS is in $\mathsf{F}_{d}$ for every fixed $d\ge 3$, even under binary encoding.
\end{restatable}

\subsection{Future Work}
\label{subsec:future_work}

Taking semilinearity as a dividing line, the complexity of low-dimensional VAS reachability is not yet fully understood.
In particular, the exact complexity of $2$-VAS reachability remains open.
It is therefore natural to ask whether $2$-VAS shares the complexity of $3$-, $4$-, and $5$-VAS: \PSPACE{}-complete under binary encoding and \NL{}-complete under unary encoding.
A positive answer would be surprising, establishing the same complexity across four consecutive dimensions.
More broadly, it would suggest that reachability complexity may increase in steps rather than with every additional counter, motivating the search for substantially better upper bounds in higher dimensions, perhaps even $\mathsf{F}_{\lceil d/2\rceil}$ for $d$-VASS instead of the currently known $\mathsf{F}_d$.

Beyond the semilinearity boundary, a similar comparison arises between $3$-VASS and $6$-VAS.
Both models can have non-semilinear reachability sets~\cite{DBLP:journals/tcs/HopcroftP79}.
However, their complexity analysis is quite different.
Reachability in $3$-VASS has recently been shown to be in \EXPSPACE~\cite{DBLP:journals/corr/abs-2607-14983}, whereas our result places $6$-VAS reachability in $\mathsf{F}_3$, leaving its elementary membership open.

Another direction is to investigate whether polynomial factorization can also help analyze VASS reachability.
Every single-state SCC is itself a VAS, so~\cref{thm:intro_factorization} applies to its reachability relation, replacing it with compositions of mixed LPS and factors with three bounded counters.
It would be interesting to investigate whether these local factorizations can be incorporated into existing methods, such as the KLM algorithm, to improve reachability bounds for VASS.

\subsection{Organization}
\label{subsec:organization}
\cref{sec:tech_overview} provides a technical overview of our results, followed by preliminary definitions and tools in~\cref{sec:preliminary}.
We then establish the generalized transformation area theorem in~\cref{sec:generalized_transformation_areas} and prove the polynomial factorization theorem in~\cref{sec:reachability_factorization}.
\cref{sec:basic_reachability_bounds,sec:semi_coverable,sec:removing_coverability} are devoted to the proof of the polynomial reachability-distance bound for $5$-VAS.

%% file: contents/2-Overview.tex
\section{Technical Overview}
\label{sec:tech_overview}

We now outline the proofs of our two main technical results: the polynomial reachability-distance bound in~\cref{thm:polynomial_witness} and the polynomial factorization theorem in~\cref{thm:intro_factorization}.
The complexity classification in~\cref{thm:main_complexity} follows from the former, while the mixed LPS characterization for $5$-VAS and higher-dimensional complexity results in~\cref{cor:flatness_of_5vas,thm:higher_dimensional_complexity} follow from the latter.
We defer the proofs of these consequences to~\cref{app:overview}.
The two main theorems share a common tool: a generalized transformation-area theorem for VAS with two bounded coordinates encoded into states.
We introduce this tool in~\cref{subsec:overview_ta}, explain how it leads to a polynomial factorization of reachability in~\cref{subsec:overview_factorability}, and focus on the proof of~\cref{thm:polynomial_witness} in~\cref{subsec:overview_witness}.

\subsection{\texorpdfstring{$B$}{B}-Encoding and Transformation Areas}
\label{subsec:overview_ta}

A standard technique for studying VAS is to extract bounded coordinates and encode them into states, thereby reducing the system to a lower-dimensional VASS.
Although the resulting VASS has multiple states, it preserves some freedom to rearrange transitions inherited from the single-state structure of VAS.
We formalize this construction as a $B$-encoding and then extend the classical transformation-area technique of van Leeuwen~\cite{DBLP:conf/stoc/Leeuwen74} and Hopcroft-Pansiot~\cite{DBLP:journals/tcs/HopcroftP79} to this setting.

\paragraph{\texorpdfstring{$B$}{B}-Encoding.}
\label{par:overview_b_encoding}
Consider a run of a $5$-VAS $V$ in which the first two coordinates remain bounded by $B$.
We can store their values in states, while keeping the remaining three coordinates as counters.
For example, a configuration $(u_1,u_2,x_1,x_2,x_3)$ is represented by the state $(u_1,u_2)$ with counter vector $(x_1,x_2,x_3)$.
To define the transitions, write each transition of $V$ as $\Vec a=(\Vec b,\Vec c)\in\ZZ^2\times\ZZ^3$.
For every state $\Vec u\in[0,B]^2$, the encoding contains the transition
\[
    \Vec u\xrightarrow{\Vec c}\Vec u+\Vec b \qquad\text{whenever}\qquad \Vec u+\Vec b\in[0,B]^2.
\]
Thus, $\Vec b$ updates the state, while $\Vec c$ updates the remaining counters, subject to their nonnegativity.
In this way, the resulting $3$-VASS represents exactly the runs of $V$ that respect these bounds.

The same construction applies to any $d$-VAS $V$ and subset $I\subseteq[d]$.
We write $\enc{V}_{I,B}$ for the resulting $(d-|I|)$-VASS and call it a \emph{$B$-encoding}.
For the rest of the discussion, we mainly consider $|I|=2$.
By permuting coordinates, we assume $I=\{1,2\}$, i.e., the first two coordinates, and abbreviate $\enc{V}_{I,B}$ as $\enc{V}_B$, when there is no ambiguity.
The \emph{Parikh image} of a path in $\enc{V}_B$ counts the occurrences of transitions of $V$, treating their copies at different states as the same transition.

\paragraph{Transformation Area.}
\label{par:overview_transformation_area}
Reordering the transitions of a VAS preserves their total effect, but may cause intermediate configurations to become negative.
A \emph{$\ZZ$-run} is a transition sequence interpreted without requiring counter values to remain nonnegative.
Van Leeuwen~\cite{DBLP:conf/stoc/Leeuwen74} introduced \emph{transformation areas}, where every $\ZZ$-run between two configurations can be rearranged into a valid run with the same Parikh image.
He showed that such a region can be obtained by requiring all but one coordinate of the endpoints to be sufficiently large.
Hopcroft and Pansiot~\cite{DBLP:journals/tcs/HopcroftP79} extended this technique to $\ZZ$-runs that are already valid on two bounded coordinates, requiring only the remaining coordinates to be sufficiently large at the endpoints.
We extend this technique to $B$-encodings.

Let $V$ be a $(d+2)$-VAS, and identify each configuration $\Vec p(\Vec x)$ of $\enc{V}_B$ with $(\Vec p,\Vec x) \in \NN^{d+2}$ in $V$.
A set $\mathcal A \subseteq [0, B]^2 \times \NN^d$ is a \emph{generalized transformation area} for $\enc{V}_{B}$ if for every $\Vec{s}, \Vec{t}\in \mathcal A$ and every $\ZZ$-run $\Vec{s}\xrightarrow{\pi}_{\ZZ}\Vec{t}$ in $\enc{V}_{B}$, there is a run $\Vec{s}\xrightarrow{\rho}\Vec{t}$ in $V$ with the same Parikh image as $\pi$.
We focus on generalized transformation areas of the form
\[
    \GTA_C=[0,B]^2\times\{\Vec x\in\NN^d:\Vec x\ge C\cdot\Vec 1\}.
\]
We call $C$ a \emph{transformation-area threshold} for $\enc{V}_B$ if $\GTA_C$ is a generalized transformation area.
We establish in~\cref{thm:generalized_transformation_area} an explicit polynomial bound on this threshold: every $C\ge(d+2)(B+1)^2\norm{T}$ is a transformation-area threshold for $\enc{V}_B$, where $T$ is the transition set of $V$.

\phantomsection
\label{par:overview_peak_states}
The proof builds on the \emph{cycle-relocation} technique in~\cite{DBLP:journals/tcs/HopcroftP79}.
We call a state visited by a run a \emph{peak state} if its first encoded coordinate is maximal among all states visited by the run.
We choose an occurrence of a peak state $\Vec q$ along the given $\ZZ$-run and extract simple cycles from the prefix and suffix around it.
For each extracted cycle, choose a state $\Vec w$ minimizing its second encoded coordinate and rotate the cycle to start there.
We then relocate the cycle to $\Vec q$ by translating each of its encoded states $\Vec z$ to $\Vec q+\Vec z-\Vec w$.
Moreover, all translated states lie in $[0,2B]^2$: the first coordinate remains nonnegative by the choice of $\Vec q$, and the second by the choice of $\Vec w$.
Steinitz's lemma~\cite{Steinitz+1913+128+176} then provides an ordering that keeps the remaining counters nonnegative when the endpoints are sufficiently large.
Hence, our construction guarantees a run in $\enc{V}_{2B}$, which lifts to a run in $V$.

\subsection{Polynomial Factorization of VAS}
\label{subsec:overview_factorability}

We then outline the proof of~\cref{thm:intro_factorization}.
Our proof repeatedly applies a factorization procedure for $B$-encodings.
Each application produces mixed-LPS factors and factors with one additional bounded original coordinate, and only the latter require further processing.
By introducing dummy coordinates with zero transition effects when necessary, three rounds suffice to obtain the desired factorization of an arbitrary $d$-VAS.
We now explain the procedure for a single $B$-encoding.

\subsubsection{Mixed LPS Between Configurations in \texorpdfstring{$\GTA_C$}{the Generalized Transformation Area}}
\label{subsubsec:overview_mixed_lps}
We first explain the easy case: how mixed-LPS factors arise from the generalized transformation area.
Consider a $\ZZ$-run $\Vec p(\Vec x)\xrightarrow{\pi}_{\ZZ}\Vec q(\Vec y)$ in a $B$-encoding $\enc{V}_B$, with both endpoints in $\GTA_C$.
It can be decomposed by removing simple cycles, leaving a short path together with the removed cycles.
This allows us to represent the given $\ZZ$-run by an LPS $\Lambda$ of polynomially bounded size under $\ZZ$-semantics.
To ensure that both endpoints lie in $\GTA_C$, we add a checking path at each end, consisting of two transitions with effects $-C\cdot\Vec{1}$ and $+C\cdot\Vec{1}$ under $\NN$-semantics.
Such a path is enabled precisely when all coordinates are at least $C$.
Combining these checks with $\Lambda$ gives a mixed LPS, whose reachability relation is contained in that of the original VAS $V$.

However, $\GTA_C$ alone is insufficient to factor arbitrary runs.
To see the difficulty, consider a $5$-VAS $V$ and its $B$-encoding $\enc{V}_B$, which is a $3$-VASS.
At each state of the encoding, the transformation area corresponds to the region $C\cdot\Vec{1}+\NN^3$.
Its complement consists of three overlapping regions near the coordinate planes, each defined by
\[
    \mathcal A_i:=\{\Vec{x}\in\NN^3:\Vec{x}(i)<C\},
    \qquad i\in[3].
\]
A run that stays entirely within one such region can be reduced to a $2$-VASS by encoding the corresponding bounded coordinate into states.
However, a run may alternate between these regions arbitrarily many times.
Simply splitting a run according to which coordinate remains bounded does not yield a bound on the number of factors.

\subsubsection{Factorization of Arbitrary Runs via High Positions}
\label{subsubsec:overview_high_positions}

To bound the number of factors, we develop an iterative factorization procedure in~\cref{sec:reachability_factorization}.

\paragraph{Rackoff's Extraction and High Positions.}
\label{par:overview_rackoff_extraction}
The first ingredient is Rackoff's extraction lemma~\cite{DBLP:conf/lics/LerouxS19,DBLP:conf/concur/ChenFZ26}.
For a fixed-dimensional VASS $U$ and a threshold $C$, there exists $H=(C+\Size(U))^{\bigO(1)}$ such that if every counter reaches at least $H$ somewhere along a run $\pi$, then deleting cycles from $\pi$ yields a run of length at most $H$ that ends in the same state with all counters at least $C$.

We choose $C$ to be a transformation-area threshold for $\enc{V}_{2B}$ and let $H$ be given by Rackoff's extraction.
Consider a run $\pi$ in $\enc{V}_B$.
We call a position of $\pi$ \emph{high} if every counter reaches at least $H$ somewhere in both the prefix ending at that position and the suffix starting there.
At such a position with state $\Vec r$, applying Rackoff's extraction to the prefix and the reversed suffix yields two short runs $\Vec p(\Vec x)\xrightarrow{\alpha}\Vec r(\Vec u)$ and $\Vec r(\Vec v)\xrightarrow{\beta}\Vec q(\Vec y)$ where $\Vec u,\Vec v\ge C\cdot\Vec 1$.
Thus, both intermediate endpoints lie in $\GTA_C$, at the same state $\Vec r$.
If $\Vec r$ is also a peak state, the deleted cycles can be relocated to $\Vec r$ to obtain a $\ZZ$-run $\Vec r(\Vec{u}) \xrightarrow{*}_\ZZ \Vec{r} (\Vec{v})$ in $\enc{V}_{2B}$.
Then the mixed-LPS construction described above allows us to represent $\pi$ by a mixed LPS with polynomially bounded parameters.

\paragraph{High-Interval Recursion.}
\label{par:overview_high_interval_recursion}
The high positions of a run form a contiguous interval, which determines a possibly empty middle sub-run $\pi_{\mathrm{mid}}$.
The prefix and suffix outside $\pi_{\mathrm{mid}}$ each have a coordinate bounded by $H+\Size(V)$, where the additional $\Size(V)$ accounts for the boundary transitions.
It suffices to proceed with the sub-run $\pi_{\mathrm{mid}}$, which contains no peak state of $\pi$.
Thus, the first encoded coordinate of a peak state of $\pi_{\mathrm{mid}}$ is strictly smaller than that of $\pi$.
Since this value lies in $[0,B]$, the procedure terminates after at most $B+1$ steps, producing at most $2B+2$ factors.

\subsection{Polynomial Reachability Distance in 5-VAS}
\label{subsec:overview_witness}

We now turn to~\cref{thm:polynomial_witness}, whose proof may exploit the given source $\Vec{s}$ and target $\Vec{t}$, unlike the preceding factorization argument.
Our starting point is the characterization of Chen, Fu, and Zheng~\cite{DBLP:conf/concur/ChenFZ26}, which reduces the polynomial reachability-distance problem for $5$-VAS to two classes of instances.
The characterization relies on the \emph{cycle-dimension} of a VASS $V$, defined as the dimension of the space $\CycleSpace(V)$ spanned by the effects of its simple cycles.

\paragraph{The CFZ Characterization.}
\label{par:overview_cfz_characterization}
We recall the characterization specialized to dimension five.
For every reachable $5$-VAS instance $(V,\Vec s,\Vec t)$, their characterization provides a reachable VASS instance $(V',s',t')$ of polynomial size such that $\MinLen_V(\Vec s,\Vec t)\le\MinLen_{V'}(s',t')$.
The new instance has one of two forms.
\begin{caselist}
    \item
    $V'$ consists of two $5$-VASS of cycle-dimension at most $2$, connected by a zero-effect transition.
    \item
    $V'$ is a single $B$-encoding $\enc{V}_B$, where $B\le\Size(V,\Vec s,\Vec t)^{\bigO(1)}$.
\end{caselist}

It therefore suffices to establish polynomial reachability-distance bounds for these two classes.
The first case can be handled using existing hybrid-set results for VASS of cycle-dimension at most two.
The second case presents the main challenge, as its SCCs may have full three-dimensional cycle spaces.
We focus on this case after explaining how to handle the first.

\paragraph{Case I: Hybrid-Set Representations.}
\label{par:overview_hybrid_representations}

VASS of cycle-dimension at most two have been studied by Zheng~\cite{DBLP:conf/concur/Zheng25,DBLP:conf/focs/0001CMOSW24,DBLP:conf/icalp/CzerwinskiJ0O25}.
These techniques were further developed in~\cite{DBLP:journals/corr/abs-2607-14983}, yielding efficient hybrid-set representations and short reachability-distance results for geometrically $2$-dimensional $3$-VASS.
A \emph{hybrid set} has the form 
\[
    H=\Vec b+P^{\langle\Vec c,\cdot\rangle\le K} := \{\Vec{b} + P\Vec{\lambda}: \Vec\lambda\in\NN^m \text{ and }\Inner{\Vec{c}, \Vec{\lambda}} \le K\},
\]
where $\Vec b$ is its base, $P=\{\Vec p_1,\ldots,\Vec p_m\}$ is a set of nonnegative periods, and $\Vec c \in \NN^m$ assigns a cost to each period.
The parameter $K$ bounds the total cost of the selected period occurrences.
A period is \emph{free} if its cost is zero and \emph{restricted} otherwise.
We measure the size of a hybrid-set representation by its \emph{representation norm} $\hnorm(H):=\max\{\norm{\Vec b},\norm P,\norm{\Vec c}\}$, which excludes $K$.

Let $V_1$ and $V_2$ denote the two components in Case I, which may encode different coordinates.
Existing hybrid-set results~\cite{DBLP:journals/corr/abs-2607-14983} describe the configurations reachable at the bridge through $V_1$.
To connect with $V_2$, we lift a hybrid set containing the original bridge configuration to $\NN^5$, fix the two coordinates encoded by $V_2$, and project onto the remaining three coordinates.
This requires absorbing some period occurrences into the base, which increases the representation norm by at most a polynomial factor.
Applying the reachability-distance theorem for geometrically $2$-dimensional $3$-VASS~\cite{DBLP:journals/corr/abs-2607-14983} to $V_2$ and $V_1$ yields a polynomial-length run through the bridge.

\paragraph{Case II: Main Challenge in \texorpdfstring{$B$}{B}-Encodings.}
\label{par:overview_encoding_case}
We now consider a reachable instance in the $3$-VASS $W=\enc{V}_B$.
Fixing the SCCs and connecting transitions (called \emph{bridges}) visited by a witnessing run $\pi$ gives a \emph{sequential sub-VASS}
\[
    U=(U_1)u_1(U_2)\cdots u_{k-1}(U_k).
\]
Note that these SCCs may have cycle-dimension three, and their number $k \le (B+1)^2$ is not bounded by a constant.
Throughout the analysis of Case II, we write $s=\Vec p(\Vec x)$ and $t=\Vec q(\Vec y)$ for the source and target configurations in $U$, where $\Vec p,\Vec q\in[0,B]^2$ are states of the encoding.
Their corresponding configurations in the original VAS $V$ are $\Vec s=(\Vec p,\Vec x)$ and $\Vec t=(\Vec q,\Vec y)$.

We begin with the basic cases, namely non-wide and coverable instances, for which polynomial-length witnesses can be obtained directly.
The remaining cases are handled by \emph{additive growth induction} over the number of SCCs $k$ in~\cref{subsubsec:overview_additive_growth}.

\subsubsection{Basic Cases: Wideness and Coverability}
\label{subsubsec:overview_basic_cases}

Recall that the cycle space $\CycleSpace(U)$ is spanned by the effects of simple cycles in $U$.
However, since cycles can only be used with nonnegative multiplicities, the cycle space does not capture which directions can be generated by their effects.
We therefore consider the \emph{cone} $\Cone(U)$, generated by the effects of simple cycles with nonnegative rational coefficients.
We call $U$ \emph{wide} if $\Cone(U)=\QQ^3$ and \emph{non-wide} otherwise.

\paragraph{Non-Wide Decomposition.}
\label{par:overview_nonwide_decomposition}
Suppose that $\Cone(U)\ne\QQ^3$.
Geometrically, the cycle cone must lie in a closed half-space through the origin.
By a separating-hyperplane argument as in~\cite{DBLP:conf/icalp/CzerwinskiJ0O25,DBLP:conf/concur/ChenFZ26}, there exists a nonzero integer vector $\Vec h$ of polynomial norm such that $\Inner{\Vec h,\Eff(\theta)}\ge0$ for every simple cycle $\theta$ in $U$.
Along any run from $s$ to $t$, deleting simple cycles from the prefix and suffix around an intermediate configuration $r(\Vec z)$ shows that $\Inner{\Vec h,\Vec z}$ is polynomially bounded in $\Size(U,s,t)$.
All intermediate configurations of such runs are therefore confined to polynomially many parallel affine planes of the form $\Inner{\Vec h,\Vec z}=m$.

By recording this value in states, we obtain a polynomial-size $3$-VASS whose cycle effects all lie in the two-dimensional space $\Vec h^\perp$.
The resulting VASS has cycle-dimension at most two, so the hybrid-set techniques from Case I yield a polynomial-length witness.

\paragraph{Coverability in the Wide Case.}
\label{par:overview_wide_coverability}
Suppose that $U$ is wide, namely $\Cone(U)=\QQ^3$.
Then every cycle effect can be canceled by inserting other cycles.
If the source and target can be moved into $\GTA_C$, respectively, such cancellation yields a $\ZZ$-run between the new endpoints.
Shortening this $\ZZ$-run and applying the transformation-area theorem then gives a polynomial-length run.

Let $C$ be a transformation-area threshold for $\enc{V}_B$.
We say that the source $s=\Vec p(\Vec x)$ is \emph{forward $C$-coverable} if $s\xrightarrow{*}\Vec p(\Vec u)$ for some $\Vec u\ge C\cdot\Vec 1$.
Symmetrically, the target $t=\Vec q(\Vec y)$ is \emph{backward $C$-coverable} if $\Vec q(\Vec v)\xrightarrow{*}t$ for some $\Vec v\ge C\cdot\Vec 1$.
We call the instance \emph{$C$-coverable} if both conditions hold.

When $U$ is wide and $C$-coverable, we first extract short covering runs at both endpoints by Rackoff's covering bound~\cite{DBLP:journals/tcs/Rackoff78}.
By canceling their effects, we obtain a $\ZZ$-run connecting two configurations in $\GTA_C$.
Since $\ZZ$-reachability can be characterized using linear equations, Pottier's theorem~\cite{DBLP:conf/rta/Pottier91} yields a polynomial-length $\ZZ$-run.
The generalized transformation-area theorem then converts it into a valid run in the original VAS without changing its length.
Together with the covering runs, this establishes the polynomial reachability-distance bound in~\cref{thm:polynomial_witness}.

\subsubsection{Bounded Prefixes and Additive-Growth Induction}
\label{subsubsec:overview_additive_growth}

We remove the remaining $C$-coverability assumptions in two stages while preserving a uniform polynomial bound.
We first consider instances whose source is forward $C$-coverable or whose target is backward $C$-coverable, which we call \emph{semi-$C$-coverable}.
Since transformation-area thresholds are preserved under reversal, we may assume that the target is backward $C$-coverable.
We then remove this final assumption using the same additive-growth induction.
The two stages are proved in~\cref{sec:semi_coverable,sec:removing_coverability}, respectively.

\paragraph{Main Obstacle.}
\label{par:overview_main_obstacle}

Assume that $\Vec{p}(\Vec{x}) \xrightarrow{\pi} \Vec{q}(\Vec{y})$ is a run in $U=(Q,T)$.
The wide coverable cases have already been handled.
Rackoff's extraction implies that, if $\Vec{p}(\Vec{x})$ is not forward $C$-coverable, some coordinate $i\in[3]$ stays below a polynomial threshold $D$ along the segment of $\pi$ in the first SCC.
Encoding this coordinate into states yields a $2$-VASS, whose reachable configurations can be described by hybrid sets as in Case I.
Following the approach of~\cite{DBLP:conf/icalp/CzerwinskiJ0O25,DBLP:journals/corr/abs-2607-14983}, one could adjoin the periods of a selected hybrid set to the next component and continue with fewer SCCs.

However, polynomial bounds for individual reduction steps do not suffice when they are repeatedly composed.
Even a recurrence as simple as $M_{j+1}\le M_j^c$ yields $M_k \le M_0^{c^k}$.
Since the number $k$ of SCCs is part of the input, the resulting polynomial degree is not fixed.
Thus, we need to avoid repeated polynomial composition altogether.

\paragraph{Bounded Coordinate Extraction.}
\label{par:overview_bounded_coordinate_extraction}
Starting from the bounded coordinate found in the first SCC, we extend the prefix across successive SCCs while keeping the same coordinate $i$ below $D$.
The prefix either reaches the target, in which case the entire run can be handled by a $2$-VASS, or stops at the entry of the first later SCC $U_{\iota+1}$ where coordinate $i$ becomes large along the witnessing run.
In the latter case, we obtain
\begin{equation}
	\label{eq:bounded_cut}
	s = \Vec{p}(\Vec{x}) \xrightarrow{\pi_0} \Vec{r}(\Vec{z}) \xrightarrow{\pi_1} \Vec{q}(\Vec{y}) = t.
\end{equation}
Since coordinate $i$ remains below $D$ throughout $\pi_0$, the entire prefix is captured by a single $2$-VASS.
Moreover, the stopping condition ensures that every $\Vec w\in\NN^3$ satisfying
\begin{equation}\label{eq:prefix_coverable_condition}
	\Vec w(i)\ge\Vec z(i),
	\qquad
	\Vec w(j)\ge D
	\quad\text{for all }j\ne i
\end{equation}
gives a forward $C$-coverable configuration $\Vec r(\Vec w)$ in the next SCC $U_{\iota+1}$. 
This suggests a way to reach the transformation area.
If we can modify the bounded prefix $\pi_0$ to raise the other two coordinates of its endpoint above $D$, while keeping coordinate $i$ at least $\Vec z(i)$, then the next SCC admits a run into $\GTA_C$.
The challenge is to achieve this modification while preserving a connection to the target $t$.

\paragraph{Additive-Growth Induction.}
\label{par:overview_additive_growth}
We induct on the number $k$ of SCCs using the partition in~\eqref{eq:bounded_cut}.
Fix a threshold $R=D^c\ge D$, where the constant $c$ will be chosen below.
We call the prefix $\pi_0$ \emph{$R$-small} if $\norm{\Vec z}\le\norm{\Vec x}+R$, and \emph{$R$-large} otherwise.
Let $\Vec u=(\Vec r,\Vec z)$ be the configuration of $V$ corresponding to the endpoint of the prefix.
Each step has one of the following outcomes:
\begin{description}
    \item[Terminal Case.]
    We obtain a polynomial-length run to the target directly.
    This is already possible when $U$ is non-wide, the source is forward $C$-coverable, or the bounded prefix reaches the target.
    We will show that every $R$-large prefix also leads to a terminal case.
    \item[Inductive Case.]
    The prefix $\pi_0$ is $R$-small, and hence $\norm{\Vec u}\le\norm{\Vec s}+R$.
	The new instance consists of the suffix $(U_{\iota+1})u_{\iota+1}\cdots u_{k-1}(U_k)$, with source $\Vec u$ and target $\Vec t$.
\end{description}

Fix $V,B,C,\Vec t$, put $M:=D+\norm{\Vec t}$, and let $L_k(N)$ bound the reachability distance for reachable sequential subinstances with at most $k$ SCCs, source norm at most $N$, and backward $C$-coverable target $\Vec t$.
A single nondecreasing polynomial $P$ bounds both the terminal cases and the cost of replacing the bounded prefix, giving
\begin{equation*}
    L_1(N)\le P(M+N),
    \qquad
    L_{k+1}(N)\le P(M+N+R)+L_k(N+R).
\end{equation*}
Thus $L_k(N)\le kP(M+N+kR)$.
Since $k\le(B+1)^2$ and $R=D^c$, where $c$ and the degree of $P$ are independent of $k$, this yields the uniform polynomial bound in~\cref{thm:polynomial_witness}.
The full induction is given in~\cref{sec:semi_coverable}.
It remains to show that $R$ can be chosen appropriately so that every $R$-large prefix falls into a terminal case.

\subsubsection{Translated Families and \texorpdfstring{$R$}{R}-Large Dichotomy}
\label{subsubsec:translated_family_overview}

Suppose that the third coordinate remains below $D$ along $\pi_0$, while $\Vec z(2)>\norm{\Vec x}+R$.
By~\eqref{eq:prefix_coverable_condition}, $\Vec r(\Vec z)$ is forward $C$-coverable whenever $\Vec z(1)\ge D$.
We would like to increase the deficient first coordinate to obtain $C$-coverability, but changing the midpoint arbitrarily may destroy its connection to the target.
We exploit the cycles responsible for the large growth of the second coordinate.
Deleting different numbers of repetitions of a mixed-sign cycle produces a family of translated midpoints, providing the freedom needed to recover the connection to the target.

\begin{example}
	\label{ex:translated_family}
	Suppose a prefix $\Vec p(\Vec x)\xrightarrow{\pi_0}\Vec r(\Vec z)$ is $R$-large, with $\Vec z(3)<D$ and $\Vec z(2)>\norm{\Vec x}+R$.
	Assume that $\pi_0$ contains $\kappa$ repetitions of a cycle $\theta$ with $\Eff(\theta)=(-2,3,0)$.
	Deleting $r$ copies of $\theta$ moves the midpoint to
	\begin{equation*}
		\Vec{z}^{(r)}=\Vec{z}+r\cdot(2,-3,0).
	\end{equation*}

	To raise both coordinates to $D$, we need $r\ge r_{\min} := \lceil(D-\Vec z(1))/2\rceil$ and $r\le r_{\max} :=\lfloor(\Vec z(2)-D)/3\rfloor$.
    In particular, to obtain at least $\lambda$ translations by deleting $\theta$, it suffices that $\min\{r_{\max},\kappa\}-\max\{r_{\min},0\}\ge\lambda-1$.
	Since $r_{\min}\le\lceil D/2\rceil$ and $\Vec z(2)\ge R$, this is guaranteed, for example, by
	\begin{equation*}
		\kappa\ge\left\lceil\frac D2\right\rceil+\lambda-1,
		\qquad
		R\ge D+3\kappa.
	\end{equation*}
	Thus, sufficiently large $\kappa$ and $R$ provide at least $\lambda$ candidate translated midpoints with both coordinates at least $D$, giving the freedom needed for the subsequent argument.
\end{example}

This example illustrates how deleting cycle repetitions changes the endpoint; the formal proof will show that the modified paths are valid runs.
More importantly, all candidate midpoints differ by translations along the same direction.
We extend this idea to hybrid sets by considering a family
\begin{equation*}
	H_r=\Vec b+r\cdot \Vec \delta+P^{\Inner{\Vec c,\cdot}\le K},
\end{equation*}
where the periods $P$, their costs $\Vec c$, and $K$ remain unchanged, while the base moves along $\Vec\delta$.
We call such a family $\{H_r\}_{r=0}^\kappa$ a \emph{translated family}.

\paragraph{Terminal Case: Translated Family.}
\label{par:overview_translated_terminal}
Consider a translated family $\{H_r\}_{r=0}^\kappa$, where $H_r=\Vec b+r\cdot\Vec\delta+P^{\langle\Vec c,\cdot\rangle\le K}$ and $\Vec\delta=(u,-v,0)$ with $u,v\ge1$ and $\Vec a(3)=0$ for every $\Vec a\in P$.
Suppose that $\Vec z\in H_0$ reaches the backward $C$-coverable target through the remaining suffix.
Assume also the stopping condition from~\eqref{eq:prefix_coverable_condition}, with coordinate $3$ bounded.

The key is to choose a translation $r\cdot\Vec \delta$ that both makes the translated source forward $C$-coverable and preserves its $\ZZ$-reachability to the target.
For sufficiently large $\kappa$, some $r\le D$ raises the first coordinate of $\Vec b+r\cdot\Vec\delta$ above $D$, while the nonnegativity of $\Vec b+\kappa\cdot\Vec\delta$ ensures that its second coordinate remains at least $(\kappa-r)v\ge D$.
The stopping condition therefore gives forward $C$-coverability.

To preserve the $\ZZ$-reachability, consider the VASS $U[P]$ augmented with the periods $P$ as self-loops at the source state.
If $U[P]$ is wide, one can choose three linearly independent cycle effects and let $\nu$ be the absolute value of their determinant.
Rounding this index upward to a multiple $r$ of $\nu$ makes $r\cdot\Vec\delta$ an integer combination of these three cycle effects, while the remaining translation range preserves $C$-coverability.
Cycle cancellation then connects the translated base to the target.
The non-wide case follows from the non-wide decomposition.
The resulting short run may use added periods whose total cost exceeds $K$.
In that case, $K\le \norm{\Vec c}\cdot \ell$, where $\ell$ is the length of the short run, so we can absorb the restricted periods into the base and then repeat the argument.
Consequently, for $\kappa\ge\kappa_0=(D+\norm P)^{\bigO(1)}$, some $r\in[0,\kappa]$ and $\Vec h\in H_r$ satisfy
\begin{equation*}
    \norm{\Vec h}, \MinLen_V\bigl((\Vec r,\Vec h),\Vec t\bigr)
    \le (D+\norm{\Vec t}+\hnorm(H_0)+\norm{\Vec\delta})^{\bigO(1)}.
\end{equation*}
The full joining argument is given in~\cref{subsec:joining_translated}.

\paragraph{The \texorpdfstring{$R$}{R}-Large Dichotomy.}
We have shown that a translated family $\{H_r\}_{r=0}^\kappa$ with sufficiently large $\kappa$ yields a terminal case.
We now explain how an $R$-large prefix produces either such a family or a hybrid set with a nonzero free period.
In the bounded-coordinate $2$-VASS, large growth must arise from repeated cycles in an LPS.
If a used cycle $\beta$ has nonzero nonnegative effect, repeating it provides a free period.
Otherwise, large growth in coordinate $k$ forces many repetitions of a mixed-sign cycle $\beta$ with effect $(-u,v)$, where $u,v\ge1$.
Deleting its copies shifts the remaining suffix along $\Vec\delta=(u,-v)$.
These two operations are illustrated in~\cref{fig:r_large_dichotomy}.

\begin{figure}[h]
	\centering
	\input{figures/r_large_dichotomy.tikz}
	\begin{subfigure}[b]{.475\linewidth}
		\centering
		\resizebox{\linewidth}{!}{\input{figures/r_large_free_period.tikz}}
		\caption{Free-period outcome.}
		\label{fig:r_large_free_period}
	\end{subfigure}
	\hfill
	\begin{subfigure}[b]{.475\linewidth}
		\centering
		\resizebox{\linewidth}{!}{\input{figures/r_large_translated_family.tikz}}
		\caption{Translated-family outcome.}
		\label{fig:r_large_translated_family}
	\end{subfigure}
	\caption{The $R$-large dichotomy, projected onto coordinates $j,k$.
	(a) Repeating a cycle $\beta$ with nonzero nonnegative effect $\Vec a$ produces a free period.
	(b) Deleting copies of a mixed-sign cycle $\beta$ shifts the reachable endpoints along $\Vec\delta=-\Eff(\beta)$, producing a translated family.
	The shapes representing $H_0,H_1,\ldots$ illustrate translations of a common hybrid set.}
	\label{fig:r_large_dichotomy}
\end{figure}
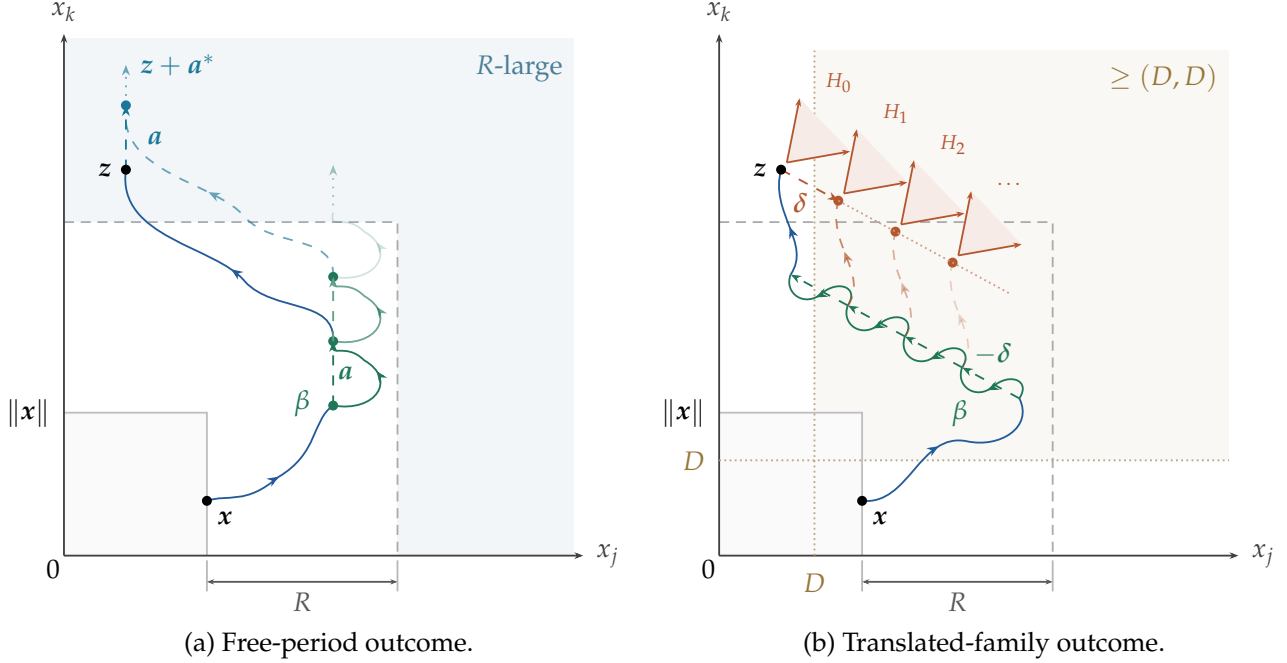

This leads to the \emph{$R$-large dichotomy}: for any $\kappa\ge1$, we can choose $R\le(D+\kappa)^{\bigO(1)}$ such that the prefix yields either a hybrid set $H$ with a free period and $\Vec z\in H$, or a translated family $\{H_r\}_{r=0}^{\kappa}$ with $\Vec z\in H_0 = \Vec b+P^{\Inner{\Vec c,\cdot}\le K}$.
Crucially, $\kappa$ can be increased arbitrarily without increasing the polynomial bounds on the periods, costs, or $\norm{\Vec\delta}$; only the required $R$ and $\norm{\Vec b}$ may depend on $\kappa$.
Every member $\Vec h$ of either outcome admits a prefix witness from $\Vec p(\Vec x)$ of length polynomial in $D+\norm{\Vec x}+\norm{\Vec h}$.

The remaining difficulty is to ensure that deleting repetitions of $\beta$ preserves nonnegativity.
We choose the last cycle increasing the second coordinate that is repeated at least $\kappa$ times.
Deleting its repetitions increases the first coordinate, while the large final value of the second coordinate provides enough reserve to keep the remaining suffix valid.
Indeed, all later cycles increasing the second coordinate occur fewer than $\kappa$ times, so their total positive contribution is bounded in terms of $\kappa$ and the LPS parameters.
Using hybrid-set descriptions twice, before and after the selected cycle, we can represent all deletion choices by a single translated family with common $P$, $\Vec c$, and $K$.

\paragraph{Terminal Case: Pumpable Sources.}
\label{par:overview_pumpable_terminal}
A configuration $\Vec p(\Vec x)$ is \emph{pumpable} if it admits a cycle reaching $\Vec p(\Vec x')$ for some $\Vec x'\ge\Vec x$ with $\Vec x'\ne\Vec x$.
We show in~\cref{subsec:pumpable_sources} that, if the source is pumpable and the target is backward $C$-coverable, then
\begin{equation*}
    \MinLen_V(\Vec s,\Vec t)\le\bigl(B+C+\Size(V,\Vec s,\Vec t)\bigr)^{\bigO(1)}.
\end{equation*}

We focus on the coordinates not increased by the initial pumping cycle.
When the source is not forward $C$-coverable, bounded-coordinate extraction yields a prefix in which these coordinates remain bounded in alternating segments.
If one coordinate remains bounded throughout the prefix, we can apply $2$-VASS techniques.
Otherwise, a hill-cutting argument yields two ordered cycles whose effects, together with the initial pumping cycle, can be combined into a strictly positive vector.
We then use an auxiliary-VASS construction to ensure that this vector can be added as a free period to a hybrid representation of the reachable configurations, which can then be joined to the backward $C$-coverable target using the coverable-case argument.

\paragraph{Uniform Choice of \texorpdfstring{$\kappa$}{K} and \texorpdfstring{$R$}{R}.}
\label{par:overview_uniform_parameters}

The preceding arguments give
\[
    \norm P\le S\le D^{\bigO(1)},\qquad
    \kappa_0\le(D+\norm P)^{\bigO(1)},\qquad
    R\le(D+\kappa)^{\bigO(1)}.
\]
Here $S$ bounds the period norms of all translated families extracted by the $R$-large dichotomy.
Crucially, this bound is independent of $\kappa$, $R$, and the current source.
We can therefore choose the parameters in the order
\[
    D
    \xrightarrow{\text{\cref{lem:large_prefix_dichotomy}}} S
    \xrightarrow{\text{\cref{lem:translated_join}}} \kappa
    \xrightarrow{\text{\cref{lem:large_prefix_dichotomy}}} R,
\]
where each arrow represents a polynomial dependence on the previously chosen parameters.
Recall also that $B$, $C$, and $D$ are polynomially bounded in the original instance size $\Size(V,\Vec{s},\Vec{t})$.
All these parameters remain polynomial in $\Size(V,\Vec{s},\Vec{t})$ and can be fixed \emph{before the induction begins}.

\subsubsection{General Case}
\label{subsubsec:overview_removing_coverability}
We now establish the polynomial reachability-distance bound without any coverability assumption.
Our approach is to extend the additive-growth induction developed for semi-$C$-coverable instances.
The bounded-prefix extraction and the $R$-large dichotomy remain valid for arbitrary targets.
Moreover, combining the pumpability argument with the semi-$C$-coverable result allows us to handle pumpable sources and hybrid sets with a nonzero free period without target coverability.
Thus, the only remaining difficulty is an $R$-large prefix yielding a translated family $\mathcal J$.

Our construction ensures that some translation of $\mathcal J$ is forward $C$-coverable.
By reversing the remaining suffix, we can therefore regard $\mathcal J$ as a fixed target family with a backward $C$-coverable translation.
We then repeat the additive-growth induction from an arbitrary source, keeping $\mathcal J$ fixed.
The main difficulty arises when an $R$-large prefix produces another translated family $\mathcal H$.
Our construction guarantees that the two families are \emph{$C$-joinable}: some members of $H_0$ and $J_0$ are connected by the original run, and both families admit $C$-coverable translations at indices at most $D$, forward for $\mathcal H$ and backward for $\mathcal J$.
We show that $C$-joinability suffices to obtain a polynomial-length run between suitable members of the two families.

The same additive recurrence thus yields a polynomial reachability-distance bound for Case II of the CFZ characterization, completing the proof of~\cref{thm:polynomial_witness}.

%% file: figures/r_large_dichotomy.tikz
\definecolor{rldpath}{RGB}{35,91,153}
\definecolor{rldcycle}{RGB}{32,119,88}
\definecolor{rldshift}{RGB}{28,118,151}
\definecolor{rlddelete}{RGB}{180,87,47}
\definecolor{rldcover}{RGB}{151,118,57}
\definecolor{rldlarge}{RGB}{75,125,151}
\tikzset{
	rld axis/.style={black!75,line width=.55pt,-{Stealth[length=2.7pt,width=2pt]}},
	rld path/.style={rldpath,line width=.55pt,postaction={decorate},
		decoration={markings,mark=at position .46 with {\arrow{Stealth[length=3.5pt,width=2.8pt]}}}},
	rld cycle/.style={rldcycle,line width=.55pt,postaction={decorate},
		decoration={markings,mark=at position .55 with {\arrow{Stealth[length=3.5pt,width=2.8pt]}}}},
	rld translated/.style={rldshift,line width=.55pt,dashed,opacity=.65,
		postaction={decorate},decoration={markings,mark=at position .55 with {\arrow{Stealth[length=3.5pt,width=2.8pt]}}}},
	rld deleted/.style={rlddelete,line width=.55pt,dashed,postaction={decorate},
		decoration={markings,mark=at position .52 with {\arrow{Stealth[length=3.5pt,width=2.8pt]}}}},
	rld effect/.style={line width=.55pt,-{Stealth[length=3.2pt,width=2.4pt]}},
	rld hybrid/.style={rlddelete,line width=.55pt,-{Stealth[length=2.5pt,width=1.8pt]}},
	rld point/.style={circle,fill=black,inner sep=1.05pt},
	rld label/.style={fill=none,inner sep=1pt},
}

\newcommand{\rldbackground}{
	\path[use as bounding box] (-.65,-1.15) rectangle (11.6,11.25);
	\fill[black!2] (0,0) rectangle (3,3);
	\draw[gray!50,line width=.55pt] (0,3)--(3,3)--(3,0);
	\draw[gray!65,densely dashed,line width=.55pt] (0,7)--(7,7)--(7,0);
	\draw[rld axis] (0,0)--(10.9,0) node[right] {$x_j$};
	\draw[rld axis] (0,0)--(0,10.95) node[above] {$x_k$};
	\node[rld label,below left] at (0,0) {$0$};
	\node[rld label,left=3pt] at (0,3) {$\norm{\Vec x}$};
	\draw[gray!55,line width=.55pt] (3,-.1)--(3,-.75) (7,-.1)--(7,-.75);
	\draw[black!65,line width=.55pt,{Stealth[length=2.2pt,width=1.6pt]}-{Stealth[length=2.2pt,width=1.6pt]}]
		(3,-.55)--(7,-.55) node[midway,below=2pt,rld label] {$R$};
}

\newcommand{\rldhybridshape}{
	\fill[rlddelete!12] (.12,.15)--(.39,1.50)--(1.45,.385)--cycle;
	\draw[rld hybrid] (.12,.15)--(.39,1.50);
	\draw[rld hybrid] (.12,.15)--(1.45,.385);
}

\newcommand{\rldpositivecycle}{
	(0,0) .. controls (.9,-.15) and (1.35,.55) .. (.65,.95)
	.. controls (.25,1.35) and (-.1,.95) .. (0,1.35)
}
\newcommand{\rldpositivesuffix}{
	(5.65,4.5) .. controls (5.7,5.4) and (4.4,5.1) .. (3.8,5.75)
	.. controls (3.2,6.5) and (1.1,6.65) .. (1.3,8.1)
}
\newcommand{\rldmixedcycle}{
	(0,0) .. controls (.25,.38) and (-.2,.55) .. (-.55,.25)
	.. controls (-.9,-.05) and (-1.35,.25) .. (-1.2,.65)
}
\newcommand{\rldmixedsuffix}{
	(1.5,5.9) .. controls (1.95,6.55) and (1.05,7.2) .. (1.3,8.1)
}

%% file: figures/r_large_free_period.tikz
\begin{tikzpicture}[x=.5cm,y=.5cm,font=\scriptsize]
	\fill[rldlarge!7] (0,7) rectangle (10.7,10.85);
	\fill[rldlarge!7] (7,0) rectangle (10.7,7);
	\rldbackground
	\node[rldlarge,rld label,anchor=north east] at (10.55,10.6) {$R$-large};
	\coordinate (rldx) at (3,1.15);
	\coordinate (rldz) at (1.3,8.1);
	\coordinate (rldza) at (1.3,9.45);

	\draw[rld path] (rldx) .. controls (3.45,1.3) and (3.8,1.05) .. (4.5,1.65)
		.. controls (5.35,2.15) and (5.2,2.85) .. (5.65,3.15);
	\begin{scope}[shift={(5.65,3.15)}]
		\draw[rld cycle] \rldpositivecycle;
	\end{scope}
	\draw[rld path] \rldpositivesuffix;

	\begin{scope}[shift={(5.65,4.5)}]
		\draw[rld cycle,opacity=.65] \rldpositivecycle;
	\end{scope}
	\begin{scope}[shift={(0,1.35)}]
		\draw[rld translated] \rldpositivesuffix;
	\end{scope}
	\begin{scope}[shift={(5.65,5.85)}]
		\draw[rld cycle,opacity=.2] \rldpositivecycle;
	\end{scope}
	\draw[rldcycle!45,dotted,rld effect] (5.65,7.35)--(5.65,8.2);

	\foreach \p in {(5.65,3.15),(5.65,4.5),(5.65,5.85)} {
		\node[rld point,fill=rldcycle] at \p {};
	}
	\node[rldcycle,rld label] at (5,3.2) {$\beta$};
	\draw[rldcycle,rld effect,dashed] (5.65,3.15)--(5.65,4.5)
		node[midway,right,rld label] {$\Vec a$};
	\draw[rldcycle!70,rld effect,dashed] (5.65,4.5)--(5.65,5.85)
		node[midway,right,rld label] {};

	\draw[rldshift,rld effect,dashed] (rldz)--(rldza)
		node[midway,right=5pt,rld label] {$\Vec a$};
	\draw[rldshift!65,dotted,rld effect] (1.3,9.65)--(1.3,10.3);
	\node[rldshift,rld label,right=3pt] at (1.3,10.3) {$\Vec z+\Vec a^\ast$};
	\node[rld point] at (rldx) {};
	\node[rld label,below right=2pt] at (rldx) {$\Vec x$};
	\node[rld point] at (rldz) {};
	\node[rld label,left=3pt] at (rldz) {$\Vec z$};
	\node[rld point,fill=rldshift] at (rldza) {};
\end{tikzpicture}

%% file: figures/r_large_translated_family.tikz
\begin{tikzpicture}[x=.5cm,y=.5cm,font=\scriptsize]
	\fill[rldcover!5] (2,2) rectangle (10.7,10.6);
	\rldbackground
	\node[rldcover,rld label,anchor=north east] at (10.55,10.4) {$\ge(D,D)$};
	\draw[rldcover!65,densely dotted,line width=.55pt] (2,0)--(2,10.6) (0,2)--(10.7,2);
	\node[rldcover,rld label,left=3pt] at (0,2) {$D$};
	\node[rldcover,rld label,below=4pt] at (2,0) {$D$};
	\coordinate (rldx) at (3,1.15);
	\coordinate (rldz) at (1.3,8.1);

	\draw[rld path] (rldx) .. controls (4,1.05) and (4.2,2.6) .. (5.2,2.4)
		.. controls (5.9,2.25) and (6.6,2.7) .. (6.3,3.3);
	\foreach \s in {0,1,2,3} {
		\begin{scope}[shift={({6.3-1.2*\s},{3.3+.65*\s})}]
			\draw[rld cycle] \rldmixedcycle;
			\draw[rldcycle,rld effect,dashed] (0,0)--(-1.2,.65);
		\end{scope}
	}
	\draw[rld path] \rldmixedsuffix;

	\foreach \s/\fade in {1/.7,2/.45,3/.25} {
		\begin{scope}[shift={({1.2*\s},{-.65*\s})}]
			\draw[rld deleted,opacity=\fade] \rldmixedsuffix;
		\end{scope}
		\node[rld point,fill=rlddelete]
			at ({1.3+1.2*\s},{8.1-.65*\s}) {};
	}
	\node[rldcycle,rld label] at (5.75,4.25) {$-\Vec\delta$};
	\node[rldcycle,rld label] at (5.05,3.0) {$\beta$};
	\draw[rlddelete,rld effect,dashed] (rldz)--(2.5,7.45);
	\node[rlddelete,rld label] at (1.7,7.4) {$\Vec\delta$};
	\draw[rlddelete!60,densely dotted,line width=.55pt] (2.5,7.45)--(6.1,5.5);
	\foreach \s/\hlab in {0/H_0,1/H_1,2/H_2,3/\ldots} {
		\begin{scope}[shift={({1.3+1.2*\s},{8.1-.65*\s})}]
			\rldhybridshape
			\node[rlddelete,rld label,anchor=south west,font=\tiny]
				at (.85,1.55) {$\hlab$};
		\end{scope}
	}

	\node[rld point] at (rldx) {};
	\node[rld label,below right=2pt] at (rldx) {$\Vec x$};
	\node[rld point] at (rldz) {};
	\node[rld label,left=3pt] at (rldz) {$\Vec z$};
\end{tikzpicture}

%% file: contents/3-Preliminary.tex
\section{Preliminaries}
\label{sec:preliminary}

Let $\NN, \ZZ, \QQ$ denote the sets of natural numbers, integers, and rational numbers, respectively.
We use, for example, $\QQ_{\ge 0}$ for the set of rationals restricted to nonnegative values.
For integers $m\le n$, write $[m,n]=\{m,m+1,\ldots,n\}$.
The interval $[1, n]$ is abbreviated as $[n]$.
Vectors are written in boldface $\Vec{u}, \Vec{v}, \Vec{x}, \Vec{y}$, etc.
For a vector $\Vec{x} \in \QQ^d$, we write $\Vec{x}(i)$ for its $i$-th component, where $i \in [d]$.
Its \emph{support} is $\Supp(\Vec x):=\{i\in[d]:\Vec x(i)\ne0\}$.
We compare vectors component-wise, so that $\Vec{x} \le \Vec{y}$ means $\Vec{x}(i) \le \Vec{y}(i)$ for all $i \in [d]$.
We use $\norm{\Vec{x}} := \max_{i \in [d]}|\Vec{x}(i)|$ for the max norm of $\Vec{x}$, and $\norm{\Vec{x}}_1 := \sum_{i \in [d]}|\Vec{x}(i)|$ for the $L_1$-norm of $\Vec{x}$.
For a finite set $X \subseteq \QQ^d$, we write $\norm{X} := \max_{\Vec{x} \in X}\norm{\Vec{x}}$ with the understanding that $\max\emptyset:=0$.
When the dimension $d$ is clear from the context, we write $\Vec{1} := (1, 1, \ldots, 1)$ and $\Vec{0} := (0, 0, \ldots, 0)$.
For a finite set $X \subseteq \QQ^d$ of vectors where $X = \{\Vec{x}_1, \ldots, \Vec{x}_n\}$, the \emph{cone} generated by $X$ is $\Cone(X) := \{\lambda_1\cdot\Vec{x}_1 + \cdots + \lambda_n\cdot\Vec{x}_n : \lambda_1, \ldots, \lambda_n \in \QQ_{\ge 0}\}$.

\paragraph{Linear Set and Hybrid Set.}
\label{par:linear_and_hybrid_sets}
For a finite set $P=\{\Vec p_1,\ldots,\Vec p_m\}\subseteq\NN^d\setminus\{\Vec0\}$, write $P^*:=\{\sum_{i=1}^m\lambda_i\cdot\Vec p_i:\Vec\lambda\in\NN^m\}$, with $\emptyset^*=\{\Vec0\}$.
A \emph{linear set} is a set $\Vec b+P^*$, where $\Vec b\in\NN^d$ is its \emph{base} and the vectors in $P$ are its \emph{periods}.
A \emph{semilinear set} is a finite union of linear sets.
For a single period $\Vec a$, write $\Vec a^{\ast}:=\{\Vec a\}^{\ast}$.

Given $\Vec{x}, \Vec{y} \in \QQ^d$, let $\Inner{\Vec{x}, \Vec{y}}$ denote the inner product of these two vectors.
We use the notation for hybrid sets from~\cite[Section~4]{DBLP:journals/corr/abs-2607-14983}.
Given $\Vec{c}\in\NN^m$ and $K\in\NN$, the associated \emph{hybrid set} is
\begin{equation*}
    H=\Vec{b}+P^{\langle\Vec{c},\cdot\rangle\le K}, \text{ where }P^{\langle\Vec{c},\cdot\rangle\le K} =\left\{\sum_{i=1}^m\lambda_i\cdot\Vec{p}_i: \Vec{\lambda}\in\NN^m,\ \langle\Vec{c},\Vec{\lambda}\rangle\le K\right\}.
\end{equation*}
We call $\Vec c(i)$ the \emph{cost} of $\Vec p_i$.
A period is \emph{free} if its cost is zero and \emph{restricted} if its cost is positive.
The \emph{representation norm} is $\hnorm(H):=\max\{\norm{\Vec b},\norm P,\norm{\Vec c}\}$, where $K$ is excluded from the norm.
The case $\Vec{c}=\Vec{0}$ gives the ordinary linear set.

\paragraph{VASS and VAS.}
\label{par:vass_and_vas}

A $d$-dimensional \emph{vector addition system with states} ($d$-VASS) is a pair $V=(Q,T)$, where $Q$ is a finite set of \emph{states} and $T\subseteq Q\times\ZZ^d\times Q$ is a finite set of \emph{transitions}.
A \emph{configuration} is a pair $p(\Vec x)$ with $p\in Q$ and $\Vec x\in\NN^d$.
A transition $u=(p,\Vec a,q)$ induces a one-step run $p(\Vec x)\xrightarrow{u}q(\Vec y)$ if $\Vec x+\Vec a=\Vec y$ and both vectors are nonnegative.
A \emph{path} is a transition sequence $\pi=(p_1,\Vec a_1,q_1)\cdots(p_k,\Vec a_k,q_k)$ with $q_i=p_{i+1}$ for $i<k$.
It is a \emph{cycle} if it is nonempty and $q_k=p_1$, and a \emph{simple cycle} if it repeats no state except its initial and final state.
A one-transition cycle is a \emph{self-loop}.
Empty paths are allowed at every state and have length zero and effect $\Vec0$.
The effect and length are defined by $\Eff(\pi) := \sum_{i=1}^k\Vec a_i$ and $|\pi| := k$.

Fix an enumeration $T=\{u_1,\ldots,u_m\}$, where $u_i=(p_i,\Vec a_i,q_i)$.
The \emph{transition matrix} of $V$ is
\begin{equation*}
	\Vec{T}:=
	\begin{pmatrix}
		\Vec{a}_1 & \Vec{a}_2 & \cdots & \Vec{a}_m
	\end{pmatrix}
	\in\ZZ^{d\times m}.
\end{equation*}
The \emph{Parikh image} of a path $\pi\in T^*$ is the vector $\Parikh{\pi}\in\NN^m$ whose $i$-th component counts the number of occurrences of $u_i$ in $\pi$.
Hence, $\Eff(\pi) = \Vec{T}\Parikh{\pi}$.
We also write $\Parikh{\pi}(u)$ for the number of occurrences of a transition $u\in T$ in $\pi$ when the chosen enumeration is irrelevant.

A path induces an \emph{$\NN$-run} (or simply a run) when all intermediate counter values remain nonnegative.
For example,
\begin{equation*}
	\pi:p_1(\Vec{x}_0) \xrightarrow{t_1} q_1(\Vec{x}_1) \xrightarrow{t_2} \cdots \xrightarrow{t_k} q_k(\Vec{x}_k),
\end{equation*}
where $\Vec{x}_0,\ldots,\Vec{x}_k\in\NN^d$.
We write $p(\Vec{x})\xrightarrow{\pi}q(\Vec{y})$ if $\pi$ induces a run from $p(\Vec{x})$ to $q(\Vec{y})$, and write $p(\Vec{x})\xrightarrow{*}q(\Vec{y})$ if such a path exists.
Two paths $\pi_1$ and $\pi_2$ may be concatenated as $\pi_1\pi_2$ whenever the target state of $\pi_1$ is the source state of $\pi_2$.
If $\pi$ is a cycle, then $\pi^m$ denotes its $m$-fold concatenation.

A looser semantics is defined on \emph{$\ZZ$-configurations} of the form $p(\Vec{z})$, where $\Vec{z}\in\ZZ^d$.
A \emph{$\ZZ$-run} or \emph{integer run} $p(\Vec{z})\xrightarrow{\pi}_{\ZZ}q(\Vec{w})$ is defined analogously to a run, with configurations replaced by $\ZZ$-configurations.
We also write $p(\Vec{z})\xrightarrow{*}_{\ZZ}q(\Vec{w})$ for $\ZZ$-runs.
For a $\ZZ$-configuration $p(\Vec{z})$, a path $\pi=t_1\cdots t_k$ from $p$ to $q$ induces a $\ZZ$-run $p(\Vec{z})\xrightarrow{\pi}_{\ZZ}q(\Vec{w})$, where $\Vec{w} = \Vec{z} + \Eff(\pi)$.
We define the \emph{drop} of $\pi$ as the vector $\Drop(\pi)\in\NN^d$ given, for every $i\in[d]$, by
\begin{equation*}
	\Drop(\pi)(i):=-\min_{j\in[0,k]}\Eff(t_1\cdots t_j)(i).
\end{equation*}
Thus, $\Drop(\pi)(i)$ is the maximum amount by which the $i$-th counter falls below its initial value along $\pi$.
Consequently, this $\ZZ$-run is a run if and only if $\Vec z\ge\Drop(\pi)$.
The \emph{reverse} $V^{\rev}$ replaces every transition $(p,\Vec a,q)$ by $(q,-\Vec a,p)$.
Clearly, $s\xrightarrow{*}t$ in $V$ if and only if $t\xrightarrow{*}s$ in $V^{\rev}$.
For a VAS, we also write $-V$ for the reversed system.

A $d$-dimensional \emph{vector addition system} ($d$-VAS) is a $d$-VASS with a single state.
In this case, transitions are determined by their effects, so we identify $T$ with a finite subset of $\ZZ^d$ and omit the unique state from configurations and runs.

\paragraph{Size and Reachability Problems.}
\label{par:size_and_reachability}

Let $V$ be a VASS and let $s,t$ be two configurations of $V$.
We call the triple $(V,s,t)$ a \emph{reachability instance}, with $s$ and $t$ referred to as its source and target configurations.
The unary-encoded \emph{size} of $V$ is defined by
\begin{equation*}
	\Size(V) := |Q| + d\cdot |T|\cdot(\norm{T}+1),
\end{equation*}
where $\norm{T} := \max_{(p,\Vec{a},q)\in T}\norm{\Vec{a}}$.
Unless stated otherwise, we consider VASS of dimension $d\ge1$ with at least one state and one transition.
For these systems, $\Size(V)\ge2$, so constant factors can be absorbed by increasing a polynomial exponent.
The unary size of a reachability instance $(V,s,t)$ is defined by
\begin{equation*}
	\Size(V,s,t) := \Size(V) + d\cdot\bigl(\norm{s}+\norm{t}+1\bigr).
\end{equation*}
The norm of a configuration $p(\Vec{x})$, and likewise of a $\ZZ$-configuration, is defined by $\norm{p(\Vec{x})}:=\norm{\Vec{x}}$.
Sometimes we also consider a pair $(V, s)$ where the target is not fixed; its size $\Size(V, s)$ is defined similarly.
For brevity, we also refer to such triples and pairs as VASS.
For a VASS $V$, its \emph{reachability relation} is
\begin{equation*}
	\Reach_V := \left\{ \bigl(p(\Vec{x}),q(\Vec{y})\bigr) : p(\Vec{x})\xrightarrow{*}q(\Vec{y}) \text{ in } V \right\}.
\end{equation*}
We compose two relations from left to right: $(s,t)\in R_1\circ R_2$ if there exists $u$ such that $(s,u)\in R_1$ and $(u,t)\in R_2$.
For a fixed source configuration $s$ of $V$, its \emph{reachability set} is $\Reach_V(s) := \left\{ t : s\xrightarrow{*}t \text{ in } V \right\}$.
For a VAS, where states are omitted, we regard $\Reach_V$ as a subset of $\NN^d\times\NN^d$ and $\Reach_V(\Vec{s})$ as a subset of $\NN^d$.
The reachability problem in VASS can be formulated as follows:
\defproblem{Reachability in VASS}{
A VASS $(V,s,t)$.
}{
Does $s \xrightarrow{*} t$ hold in $V$?
}

Complexity bounds for reachability problems are usually obtained from length bounds on runs witnessing the reachability.
We write $\Len(V,s,t)  := \left\{ |\pi| : s\xrightarrow{\pi}t \text{ in } V \right\}$ for the set of lengths of runs from $s$ to $t$.
The \emph{reachability distance} from $s$ to $t$ in $V$ is defined by $\MinLen_V(s,t) := \min \Len(V,s,t)$, with the convention that $\MinLen_V(s,t):=\infty$ if $t$ is not reachable from $s$.
We use $\bigO(n)$ and $\PolyF{n}$ to denote unspecified linear and polynomial functions on integers, respectively.

\paragraph{Mixed Linear Path Schemes.}
\label{par:linear_path_schemes}
A \emph{linear path scheme} (LPS) is an expression
\begin{equation*}
    \Lambda=\alpha_0\beta_1^*\alpha_1\cdots\beta_k^*\alpha_k,
\end{equation*}
where each $\alpha_i$ is a finite sequence of vectors in $\ZZ^d$, and each $\beta_i$ is a nonempty such sequence.
Define $|\Lambda|:=|\alpha_0\beta_1\alpha_1\cdots\beta_k\alpha_k|$ and $\norm\Lambda := \max\{\norm{\Vec a} : \Vec{a} \text{ appears in }\Lambda\}$.
A path captured by $\Lambda$ has the form $\alpha_0\beta_1^{n_1}\alpha_1\cdots\beta_k^{n_k}\alpha_k$ for some $n_1,\ldots,n_k\in\NN$.
For $\XX \in\{\NN,\ZZ\}$, write $\Vec s\xrightarrow{\Lambda}_{\XX}\Vec t$ if $\Lambda$ captures a sequence of vectors whose sum is $\Vec t-\Vec s$ and whose prefix sums, when added to $\Vec s$, all lie in $\XX^d$.
The subscript $\NN$ is often omitted.

A \emph{mixed linear path scheme} (mixed LPS) is an expression
\begin{equation*}
    \Lambda=(\Lambda_1,\XX_1)\cdots(\Lambda_m,\XX_m),
\end{equation*}
where $m\ge1$, $\Lambda_1\Lambda_2\cdots \Lambda_m$ is an LPS, and $\XX_i\in\{\NN,\ZZ\}$.
Define $|\Lambda| := \sum_{i \in [m]}|\Lambda_i|$, and $\norm\Lambda := \max_{i\in[m]}\norm{\Lambda_i}$.
We write $\Vec c_0\xrightarrow{\Lambda}\Vec c_m$ if there exist nonnegative vectors $\Vec c_0,\ldots,\Vec c_{m}$ such that $\Vec c_{i-1}\xrightarrow{\Lambda_i}_{\XX_i}\Vec c_i$ for all $i$.
The counter values inside a $\ZZ$-component may be negative.

A witness $w$ specifies the repetition counts of all cycles and the intermediate vectors between consecutive components.
These choices must satisfy the semantics of each component.
The length $|w|$ is the total number of vectors in the resulting sequences.
Write $\Reach_\Lambda:=\{(\Vec s,\Vec t):\Vec s\xrightarrow{\Lambda}\Vec t\}$ and let $\MinLen_\Lambda(\Vec s,\Vec t)$ be the minimum witness length, with value $\infty$ when there is no witness.

\paragraph{Cone and Cycle-Dimension.}
\label{par:cones_and_cycle_dimension}
Let $V$ be a $d$-VASS.
Its \emph{cycle space} $\CycleSpace(V)\subseteq\QQ^d$ is the span of the effects of simple cycles in the VASS $V$, i.e.,
\begin{equation*}
    \CycleSpace(V) := \Span\{\Eff(\pi) : \pi \text{ is a simple cycle in } V\}.
\end{equation*}
Its \emph{cycle-dimension} $\dimcyc(V) := \dim(\CycleSpace(V))$ is the dimension of its cycle space.
Its \emph{component-dimension} $\dimcom(V):=\max\{\dimcyc(U):U\text{ is an SCC of }V\}$ is the maximum cycle-dimension of its SCCs.
In particular, for a $d$-VAS $V$ with transition set $T$, we have $\CycleSpace(V) = \Span(T)$.
We remark that the cycle-dimension is also known as the geometric dimension and has been used as a key parameter in the analysis of reachability and coverability~\cite{DBLP:conf/lics/LerouxS19,DBLP:conf/icalp/FuYZ24,DBLP:conf/icalp/CzerwinskiGOSZ26}.

\paragraph{\texorpdfstring{$B$}{B}-Encoding and Sequential VASS.}
\label{par:encodings_and_sequential_vass}
For a vector $\Vec{x} \in \ZZ^d$ and a subset $I \subseteq [d]$, we denote by $\restr{\Vec{x}}{I}$ the restriction of $\Vec{x}$ to the coordinates in $I$.
Let $V = (Q,T)$ be a $d$-VASS, $I \subseteq [d]$, and $B \in \NN$.
The $(d - |I|)$-VASS $\enc{V}_{I,B}$ has state set $Q' := Q \times [0,B]^I$.
For every $(p,\Vec a,q) \in T$ and $p(\Vec u)\in Q'$, it has the transition
\begin{equation*}
    \bigl(p(\Vec u),\restr{\Vec a}{[d]\setminus I}, q(\Vec v)\bigr)
    \quad\text{precisely when}\quad
    \Vec v=\Vec u+\restr{\Vec a}{I}\in [0,B]^I.
\end{equation*}
For $p(\Vec x)$ with $\restr{\Vec x}{I}\in [0,B]^I$, set $\enc{p(\Vec x)}_{I, B} =(p,\restr{\Vec x}{I})(\restr{\Vec x}{[d] \setminus I})$.
In the subsequent discussion, we only consider the case where $|I| = 2$, and assume, w.l.o.g., that $I$ indexes the first two coordinates of the VAS $V$ when there is no ambiguity.
In such cases, we abbreviate $\enc{V}_{I,B}$ as $\enc{V}_{B}$, and $\enc{\Vec x}_{I, B}$ as $\enc{\Vec x}_{B}$.
We call $\enc{V}_{B}$ a \emph{$B$-encoding (of a VAS)} when $V$ is a VAS, and a \emph{$B$-encoding of a VASS} otherwise.
We also allow $0$-VASS, which are finite directed graphs with no counters.

Let $U_i=(Q_i,T_i)$, for $i\in[k]$, be strongly connected $d$-VASS with pairwise disjoint state sets, and, for $i\in[k-1]$, let $u_i=(q_i,\Vec a_i,p_{i+1})$ with $q_i\in Q_i$ and $p_{i+1}\in Q_{i+1}$ be a \emph{bridge}.
The \emph{sequential VASS}
\begin{equation*}
    U=(U_1)u_1(U_2)\cdots u_{k-1}(U_k)
\end{equation*}
has state set $\bigcup_iQ_i$ and transition set $\bigcup_iT_i\cup\{u_1,\ldots,u_{k-1}\}$.

\paragraph{Covering Runs and Rackoff's Extraction.}
\label{par:rackoff_extraction}

We recall two results for extracting short runs reaching sufficiently large configurations.
A configuration $q(\Vec z)$ is \emph{coverable} from $p(\Vec x)$ if $p(\Vec x)\xrightarrow{*}q(\Vec z')$ for some $\Vec z'\geq \Vec z$.
The classical result of Rackoff gives a bound on the length of a covering run.

\begin{lemma}[{\cite[Theorem 3.5]{DBLP:journals/tcs/Rackoff78}}]
	\label{lem:rackoff_cover}
	Let $V=(Q,T)$ be a $d$-VASS.
	If $q(\Vec z)$ is coverable from $p(\Vec x)$, there is a covering run ending at state $q$ of length less than $N^{\bigO((d+1)!)}$, where $N:=\max\{2,|Q|,\norm{T},\norm{\Vec z}\}$.
\end{lemma}

We also use a related extraction result for a given run.
For a run $\pi$ and a threshold $D\in\NN$, let $\BDCounters(\pi,D)$ denote the set of coordinates whose values remain below $D$ throughout $\pi$.
Hence, $\BDCounters(\pi,D)=\emptyset$ means that every coordinate reaches a value at least $D$ somewhere along $\pi$, although these occurrences need not be simultaneous.
The following lemma shows that, for a sufficiently large threshold $D$, one can instead find a short run from the same source reaching a configuration in which all coordinates are simultaneously large.

\begin{lemma}[{\cite[Lemma 3.3]{DBLP:conf/concur/ChenFZ26}}]
	\label{lem:coordinate_extraction}

	For every $d$-VASS $(V,s)$ and number $C \in \NN$, there exists a threshold
	\begin{equation*}
		D_0:=(C + \Size(V))^{\bigO((d+1)!)}
	\end{equation*}
	such that, if every coordinate reaches a value at least $D_0$ somewhere along a run $\pi$ from $s$ in $V$, i.e., $\BDCounters(\pi, D_0)=\emptyset$, there is also a run $\rho$ from $s$ with $|\rho| \le D_0$ and ending at $q(\Vec{x})$ satisfying $\Vec{x} \ge C \cdot \Vec{1}$.
	Moreover, $\rho$ is obtained by deleting cycles from $\pi$ and therefore has the same final state as $\pi$.
\end{lemma}

\paragraph{Integer Programming.}
\label{par:integer_programming}
Let $\Vec{A} \in \mathbb{Z}^{m \times n}$ be a matrix, $\Vec{b} \in \mathbb{Z}^m$ be a vector, and $\Vec{x}$ be a vector of variables.
A \emph{linear Diophantine system} $\Vec{A} \Vec{x} = \Vec{b}$ is a set of $m$ linear equations with $n$ variables, whose coefficients are integers.
A \emph{minimal solution} of the system $\Vec{A}\Vec{x} = \Vec{b}$ is a minimal element, under the component-wise order $\le$, among the nonnegative integer solutions of the system.
Any nonnegative integer solution can be written as $\Vec x=\Vec x_0+\sum_j\lambda_j\cdot\Vec h_j$, where $\Vec x_0$ is minimal for $\Vec A\Vec x=\Vec b$, the $\lambda_j$ are nonnegative integers, and the $\Vec h_j$ range over the finitely many minimal elements of $\{\Vec h\in\NN^n\setminus\{\Vec0\}:\Vec A\Vec h=\Vec0\}$.
We write $\norm{\Vec{A}}_{1,\infty} := \max_{i \in [m]} \norm{\Vec{a}_i}_1$ and $\norm{\Vec{A}} := \max_{i \in [m]} \norm{\Vec{a}_i}$, where $\Vec{a}_i$ is the $i$-th row of $\Vec{A}$.
Pottier's theorem~\cite{DBLP:conf/rta/Pottier91} is used to bound the size of minimal solutions:

\begin{theorem}[Pottier's Theorem~{\cite[Theorem 1]{DBLP:conf/rta/Pottier91}}]
	\label{thm:pottier}
	Let $\Vec{A}\Vec{x} = \Vec{b}$ be a linear Diophantine system, where $\Vec{A} \in \mathbb{Z}^{m \times n}$ and $\Vec{b} \in \mathbb{Z}^m$.
	Each nonzero minimal solution $\Vec{x}_0 \in \NN^n$ of the system satisfies
	\begin{equation}\label{eq:pottier}
		\norm{\Vec{x}_0}_1 \le \left( 1 + \norm{\Vec{A}}_{1,\infty} + \lVert \Vec{b} \rVert \right)^{r},
	\end{equation}
	where $r:=\Rank\begin{pmatrix}\Vec{A}&\Vec{b}\end{pmatrix}$ is the rank of the augmented matrix $\begin{pmatrix}\Vec{A}&\Vec{b}\end{pmatrix}$.
\end{theorem}

Sometimes we bound $\norm{\Vec{A}}_{1,\infty}$ by $n\cdot \norm{\Vec{A}}$ to simplify \eqref{eq:pottier}.
For a system of linear inequalities of the form $\Vec{A}\Vec{x} \ge \Vec{b}$, we can introduce $m$ slack variables and reduce it to a linear Diophantine system without increasing the number of equations.
By~\cref{thm:pottier}, the $L_1$-norm of any minimal solution of the system is bounded by $\left(2 + \norm{\Vec{A}}_{1,\infty} + \norm{\Vec{b}} \right)^m$, hence the following.

\begin{corollary}\label{cor:pottier}
	Let $\Vec{A}\Vec{x} \ge \Vec{b}$ be a system of $m$ linear inequalities, where $\Vec{A} \in \mathbb{Z}^{m \times n}$ and $\Vec{b} \in \mathbb{Z}^m$.
	Each minimal solution $\Vec{x}_0 \in \NN^n$ of the system satisfies $\norm{\Vec{x}_0}_1 \le \left( 2 + \norm{\Vec{A}}_{1,\infty} + \lVert \Vec{b} \rVert \right)^{m}$.
\end{corollary}

%% file: contents/4-Quantitative.tex
\section{Generalized Transformation Areas for \texorpdfstring{$B$}{B}-Encodings of VAS}
\label{sec:generalized_transformation_areas}

In this section, we establish a generalized transformation-area theorem for $B$-encodings of VAS.
Classically, a transformation area is a region in which every $\ZZ$-run between two configurations can be rearranged into a run with the same Parikh image~\cite{DBLP:conf/stoc/Leeuwen74,DBLP:journals/tcs/HopcroftP79}\footnote{We also give an explicit bound for van Leeuwen's classical transformation areas in~\cref{app:classical_transformation_area}.}.
We now formally define the generalization to $B$-encodings outlined in~\cref{sec:tech_overview}.

Let $V$ be a $(d+2)$-VAS with transition set $T$ and $\enc{V}_{B}$ its $B$-encoding.
We identify each configuration $\Vec p(\Vec x)$ of $\enc{V}_{B}$ with $(\Vec p,\Vec x)\in\NN^{d+2}$ in $V$.
Throughout this section, Parikh images of paths in $\enc{V}_{B}$ count their underlying transitions in $T$.

\begin{definition}
	\label{def:gta}
	A set $A \subseteq [0, B]^2 \times \NN^d$ is a \emph{generalized transformation area} for $\enc{V}_{B}$ if for every $\Vec{s}, \Vec{t}\in A$ and every $\ZZ$-run $\Vec{s}\xrightarrow{\pi}_{\ZZ}\Vec{t}$ in $\enc{V}_{B}$, there is a run $\Vec{s}\xrightarrow{\rho}\Vec{t}$ in $V$ with $\Parikh\rho=\Parikh\pi$.
\end{definition}

Recall that $C$ is a transformation-area threshold for $\enc{V}_{B}$ if
\begin{equation*}
	\GTA_C = [0, B]^2 \times \{\Vec x\in\NN^d:\Vec x\ge C\cdot \Vec{1}\}
\end{equation*}
is a generalized transformation area for $\enc{V}_{B}$.
The main result of this section gives an explicit polynomial threshold and, more strongly, guarantees that the new run stays within the $2B$-encoding.

\begin{theorem}[Generalized Transformation Area Theorem]
	\label{thm:generalized_transformation_area}
	Let $V$ be a $(d+2)$-VAS with transition set $T$ and let $B\in \NN$.
	Set
	\begin{equation*}
		C := (d+2)(B+1)^2\norm{T}.
	\end{equation*}
	Then $\GTA_C$ is a generalized transformation area.
	More precisely, for every $\Vec s,\Vec t\in\GTA_C$ and every $\ZZ$-run $\Vec s\xrightarrow{\pi}_{\ZZ}\Vec t$ in $\enc{V}_{B}$, there is a run $\Vec s\xrightarrow{\rho}\Vec t$ in $\enc{V}_{2B}$ with $\Parikh\rho=\Parikh\pi$.
\end{theorem}

The proof adapts the cycle-relocation technique of Hopcroft and Pansiot~\cite[Lemma 3.9]{DBLP:journals/tcs/HopcroftP79}.
We first extract and relocate simple cycles to one common state, keeping the two encoded coordinates within $[0,2B]^2$.
We then apply the nonzero-sum form of Steinitz's lemma in~\cref{cor:steinitz_partial_sums} to reorder their effects so that the $d$ remaining counters stay nonnegative.

\subsection{Steinitz Reordering and Cycle Relocation}
\label{subsec:steinitz_and_cycle_relocation}

We use Steinitz's lemma~\cite{Steinitz+1913+128+176} in the form of Grinberg and Sevast'yanov~\cite{GrinbergSevastyanov1980}, which bounds the partial sums of a reordered zero-sum sequence.

\begin{lemma}[Steinitz Lemma~{\cite{GrinbergSevastyanov1980}}]
	\label{lem:steinitz}
	Let $\Vec{v}_1,\ldots,\Vec{v}_n\in\QQ^d$, with $\sum_{\ell=1}^n\Vec{v}_\ell=\Vec0$ and $\norm{\Vec{v}_\ell}\le M$.
	There is a permutation $\sigma$ of $[n]$ such that $\norm{\sum_{\ell=1}^k\Vec v_{\sigma(\ell)}}\le dM$ for every $k\in[0,n]$.
\end{lemma}

The total effect of the extracted cycles need not vanish.
We also need a version of Steinitz's lemma for a nonzero total sum, which is a corollary of the above lemma.

\begin{corollary}
	\label{cor:steinitz_partial_sums}
	Let $\Vec v_1,\ldots,\Vec v_n\in\QQ^d$ have norm at most $M$, and let $\Vec z :=\sum_{i=1}^n\Vec v_i$.
	There is a permutation $\sigma$ such that every partial sum $\Vec p_k=\sum_{i=1}^k\Vec v_{\sigma(i)}$ satisfies
	\begin{equation*}
		\Vec p_k\ge\min(\Vec0,\Vec z)-dM \cdot \Vec1,
	\end{equation*}
	where the minimum is componentwise.
\end{corollary}

\begin{proof}
	The empty case is immediate.
	Otherwise append $n$ auxiliary copies of $-\Vec z/n$.
	Since $\norm{\Vec z}\le nM$, all $2n$ vectors have norm at most $M$ and sum to zero.
	Apply~\cref{lem:steinitz}, then delete the auxiliary vectors.
	For any prefix of the resulting order, consider the augmented prefix ending at its last original vector.
	If it contains $m$ auxiliary vectors, its sum is $\Vec p_k-(m/n) \cdot \Vec z$, where $k$ is the number of original vectors in the prefix.
	Since $(m/n) \cdot \Vec z\ge\min(\Vec0,\Vec z)$ and $\norm{\Vec p_k-(m/n) \cdot \Vec z} \le dM$, one has the required bound $\Vec p_k\ge\min(\Vec0,\Vec z)-dM \cdot \Vec1$.
	The empty prefix also satisfies it.
\end{proof}

\paragraph{Cycle Relocation.}
\label{par:cycle_relocation}
A state $\Vec{r}$ visited by $\pi$ is a \emph{peak state} of a path $\pi$ if 
\[
	\Vec{r}(1) = \max\{\Vec{p}(1) : \Vec p\text{ appears in }\pi\}.
\]
We first show that a simple cycle can be relocated to a peak state while preserving its Parikh image.

\begin{lemma}[Cycle Relocation]
	\label{lem:cycle_relocation}
	Let $\pi$ be a path in $\enc{V}_B$.
	Suppose that $\beta$ is a simple cycle in $\pi$ and that $\Vec r$ is a peak state of $\pi$.
	Then there exists a simple cycle $\widehat\beta$ at $\Vec r$ in $\enc{V}_{2B}$ with $\Parikh\beta=\Parikh{\widehat{\beta}}$.
\end{lemma}

\begin{proof}
	Rotate $\beta$ to start at a state $\Vec w$ minimizing its second coordinate, and translate its encoded states by $\Vec r-\Vec w$.
	Since $\beta$ has zero total effect on the encoded coordinates, every resulting state is $\Vec r+\Vec v-\Vec w$ for a state $\Vec v$ of $\beta$.
	Here $\Vec r(1)\ge\Vec w(1)$ by the peak-state condition and $\Vec v(2)\ge\Vec w(2)$ by the choice of $\Vec w$.
	Together with $\Vec r,\Vec v,\Vec w\in[0,B]^2$, this gives $0\le\Vec r(i)+\Vec v(i)-\Vec w(i)\le2B$ for $i\in\{1,2\}$.
	Thus the same underlying VAS transitions define a cycle $\widehat\beta$ at $\Vec r$ in $\enc{V}_{2B}$.
	Rotation and translation preserve state distinctness and underlying transition counts, so $\widehat\beta$ is simple and $\Parikh\beta=\Parikh{\widehat\beta}$.
\end{proof}

The next corollary decomposes a path in $\enc{V}_{B}$ into two short paths and simple cycles that can be relocated to a peak state in $\enc{V}_{2B}$.
Write $Q_B=[0,B]^2$ for the states of $\enc V_B$.

\begin{corollary}
	\label{cor:cycle_relocation}
	Let $\pi$ be a path in $\enc{V}_B$ from $\Vec p$ to $\Vec q$.
	There exists a path $\alpha\gamma_1\cdots\gamma_n\beta$ from $\Vec p$ to $\Vec q$ in $\enc{V}_{2B}$ such that:
	\begin{enumerate}
		\item The prefix $\alpha$ and the suffix $\beta$ are paths in $\enc{V}_B$, with $|\alpha|,|\beta|\le (B + 1)^2-1$.
		\item Each $\gamma_i$ is a simple cycle in $\enc{V}_{2B}$ of length at most $(B + 1)^2$.
		\item $\Parikh\pi=\Parikh\alpha+\sum_{i=1}^n\Parikh{\gamma_i}+\Parikh\beta$.
	\end{enumerate}
	Consequently $\alpha\gamma_{\sigma(1)}\cdots\gamma_{\sigma(n)}\beta$ is a path in $\enc{V}_{2B}$ for every permutation $\sigma$.
\end{corollary}

\begin{proof}
	Split $\pi$ at an occurrence of a peak state $\Vec r$ and obtain two segments of $\pi$.
	Whenever either segment is not simple, select a simple cycle $\theta$ and use~\cref{lem:cycle_relocation} to obtain a simple cycle $\widehat \theta$ at $\Vec r$ with $\Parikh{\theta} = \Parikh{\widehat{\theta}}$.
	We also have $|\widehat{\theta}| = |\theta| \le |Q_B| \le (B+1)^2$.
	Delete $\theta$ from the residual path.
	The remaining encoded states are unchanged, so the residual path stays in $\enc{V}_B$ and still passes through the peak state $\Vec r$.
	Hence, the same lemma applies at every step.
	Each deletion shortens the path, so the process ends with two simple paths $\alpha$ and $\beta$ with $|\alpha|,|\beta|\le (B + 1)^2-1$.

	Let $\gamma_1,\ldots,\gamma_n$ be the constructed cycles.
	The successive deletions give the stated Parikh identity.
	Finally, $\alpha$ ends at $\Vec r$, all $\gamma_i$ start and end there, and $\beta$ starts there, so their concatenation is a path in $\enc{V}_{2B}$ with the cycles in any order.
\end{proof}

\subsection{Proof of the Transformation-Area Theorem}
\label{subsec:transformation_area_proof}

Combining the previous lemmas, we can now prove~\cref{thm:generalized_transformation_area}.
Write $I=\{1,2\}$ for the encoded coordinates and $J=[d+2]\setminus I$ for the remaining coordinates.

\begin{proof}[Proof of~\cref{thm:generalized_transformation_area}]
	Assume $\Vec s=(\Vec p,\Vec x)$ and $\Vec t=(\Vec r,\Vec y)$, where $\Vec p,\Vec r\in Q_B$ and $\Vec x,\Vec y\ge C\cdot\Vec1$.
	Put $M:=(B+1)^2\norm{T}$, so $C=(d+2)M$.
	Apply~\cref{cor:cycle_relocation} to obtain $\alpha,\gamma_1,\ldots,\gamma_n,\beta$ from $\pi$.
	For a permutation $\sigma$ to be chosen, consider
	\begin{equation*}
		\rho=\alpha\gamma_{\sigma(1)}\cdots\gamma_{\sigma(n)}\beta.
	\end{equation*}
	This path has the same Parikh image as $\pi$.
	Let \(\Vec u=\Vec x+\restr{\Eff(\alpha)}J\) and \(\Vec v=\Vec y-\restr{\Eff(\beta)}J\) be the \(J\)-coordinates immediately after \(\alpha\) and immediately before \(\beta\), respectively.
	Every prefix of $\alpha$ and every suffix of $\beta$ has effect of norm at most $M$ on $J$.
	Thus $\alpha$ and $\beta$ remain nonnegative on $J$, and
	\begin{equation*}
		\Vec u,\Vec v\ge(C - M) \cdot \Vec1=(d+1) M \cdot \Vec1.
	\end{equation*}

	Let $\Vec a_i=\restr{\Eff(\gamma_i)}J$ for $i\in[n]$.
	Since $|\gamma_i|\le(B+1)^2$, we have $\norm{\Vec a_i}\le M$, and the Parikh identity gives $\sum_i\Vec a_i=\Vec v-\Vec u$.
	By~\cref{cor:steinitz_partial_sums}, there is a permutation $\sigma$ such that every boundary between the reordered middle blocks has its $J$-coordinates at least
	\begin{equation*}
		\Vec u+\min(\Vec0,\Vec v-\Vec u)-dM\cdot\Vec1
		=\min(\Vec u,\Vec v)-dM\cdot\Vec1\ge M\cdot\Vec1.
	\end{equation*}
	A prefix of any $\gamma_i$ decreases each coordinate by at most $M$.
	Hence $\rho$ remains nonnegative on $J$.

	It remains to check the two encoded coordinates.
	By~\cref{cor:cycle_relocation}, $\rho$ stays in $[0,2B]^2$ on $I$ when started at $\Vec p$.
	Finally, $\Parikh{\rho}=\Parikh{\pi}$ gives the target $\Vec t$.
	Thus $\rho$ is the required run in $\enc{V}_{2B}$.
\end{proof}

\subsection{Stability of Transformation-Area Thresholds}
\label{subsec:threshold_stability}

Later arguments may reverse the system or add transitions with zero effect on $I$ and nonnegative or nonpositive effect on $J$.
\cref{cor:threshold_closure} allows us to use the same threshold under these operations.

\begin{corollary}
	\label{cor:threshold_closure}
	Let $V$ be a $(d+2)$-VAS and $B\in\NN$.
	Let $C$ be any transformation-area threshold for $\enc V_B$.
	The same $C$ is a transformation-area threshold for $\enc{V^{\rev}}_B$ and for $\enc{V'}_B$ whenever $V'$ is obtained from $V$ by adjoining transitions in $\{(0,0)\}\times\bigl(\NN^d\cup(-\NN^d)\bigr)$.
\end{corollary}

\begin{proof}
	For reversal, we reverse a $\ZZ$-run in $\enc{V^{\rev}}_B$ to be a $\ZZ$-run in $\enc{V}_B$, apply the transformation-area property for $\enc V_B$, and reverse the resulting run.
	For the second assertion, consider a $\ZZ$-run $\pi$ from $\Vec s$ to $\Vec t$ in $\enc{V'}_B$ with endpoints in $\GTA_C$.
	Partition the added occurrences into nonnegative and nonpositive ones, assigning every occurrence of a zero-effect transition to the former group.
	Let $\Vec a,\Vec b\in\{(0,0)\}\times\NN^d$ be their respective total effect and negated total effect.
	The remaining path $\eta$ is still a path in $\enc{V}_B$, which induces a $\ZZ$-run $\Vec s + \Vec a \xrightarrow{\eta}_\ZZ \Vec{t} + \Vec{b}$.
	Since $\Vec s+\Vec a,\Vec t+\Vec b\in\GTA_C$, the transformation-area property gives a run in $V$ between these two configurations with the same Parikh image as $\eta$.
	The deleted nonnegative and nonpositive transitions give a run from $\Vec s$ to $\Vec s+\Vec a$ and a run from $\Vec t+\Vec b$ to $\Vec t$, respectively.
	Concatenating these three runs gives a run from $\Vec s$ to $\Vec t$ in $V'$ with the same Parikh image as $\pi$.
\end{proof}

%% file: contents/5-Flatness.tex
\section{Polynomial Factorization of VAS Reachability}
\label{sec:reachability_factorization}
We now establish the polynomial factorization of VAS reachability stated in \cref{thm:intro_factorization}.
The proof uses the following linear factorization of $B$-encodings, with their reachability relations henceforth viewed on configurations of $V$ by restoring the counter values stored in their states.

Recall that a mixed LPS consists of LPS components interpreted under $\NN$- or $\ZZ$-semantics.
For a VASS or a mixed LPS $V_i$, we write $\Reach_{V_i}$ for its reachability relation and $\MinLen_{V_i}(s,t)$ for its minimum witness length.
In particular, for a VASS this is the ordinary reachability distance.

\begin{restatable}
{lemma}{LinearFactorabilityVB}
\label[lemma]{lem:encoding_factorability}
    For every fixed $d\ge1$, there exists a polynomial $P$ such that, for every $(d+2)$-VAS $V$, every $I\subseteq[d+2]$ with $|I|=2$, and every $B\in\NN$, there is a finite family $\mathcal{F}$ of sequences $(R_{1},\ldots,R_{m})$ satisfying
    \begin{equation*}
        \Reach_{\enc{V}_{I,B}}
        \subseteq
        \bigcup_{(R_{1},\ldots,R_{m})\in\mathcal{F}}R_{1}\circ\cdots\circ R_{m}
        \subseteq
        \Reach_{V}.
    \end{equation*}
    Every sequence in $\mathcal{F}$ is nonempty and has length at most $2B+2$.
    Each factor is $R_i=\Reach_{V_i}\subseteq\Reach_V$, where $V_i$ has one of the following forms:
    \begin{enumerate}
        \item $V_i=\enc{V}_{I\cup\{j\},D}$ for some $j\in[d+2]\setminus I$ and $D\le P(B+\Size(V))$;
        \item $V_i=\Lambda$ for some mixed LPS $\Lambda$, with $\norm{\Lambda},|\Lambda|\le P(B+\Size(V))$.
    \end{enumerate}
    Moreover, for every factor $R_i$ and every $(\Vec{x},\Vec{y})\in R_i$, 
    \begin{equation*}
        \MinLen_{V}(\Vec{x},\Vec{y}) \leq \MinLen_{V_{i}}(\Vec{x},\Vec{y}).
    \end{equation*}
\end{restatable}

\begin{proof}[Proof sketch of \cref{thm:intro_factorization}]
    By adding dummy counters, any encoding of $V$ with fewer than two bounded coordinates can be regarded as one with exactly two bounded coordinates, so~\cref{lem:encoding_factorability} applies.
    Starting from $V$, we repeatedly apply~\cref{lem:encoding_factorability} to $B$-encoding factors with fewer than three bounded original coordinates, leaving mixed-LPS factors unchanged.
    Each application replaces such a factor by a finite union of compositions of mixed LPS and $B$-encoding factors with one additional bounded original coordinate.
    After three levels, every remaining $B$-encoding factor has three bounded coordinates, and the procedure terminates.

    Since each application produces compositions of polynomial length with polynomially bounded parameters, the fixed recursion depth ensures that the resulting factorization has the required polynomial bounds.
    Every reachable pair in $V$ is covered by a resulting composition $R_{1}\circ\cdots\circ R_{m}$.
    Conversely, for any resulting composition, witnesses for its factors can be realized as runs of $V$, which concatenate to a run between the same endpoints.
\end{proof}

The full proof is given in \cref{app:reachability_factorization}.
The rest of this section proves \cref{lem:encoding_factorability} and presents applications of the polynomial factorization.

\subsection{Mixed LPS from Generalized Transformation Areas}
\label{subsec:mixed_lps_transformation_areas}

To apply~\cref{thm:generalized_transformation_area}, we represent $\ZZ$-reachability in $\enc V_B$ by mixed LPS whose endpoints are restricted to $\GTA_C$.
We enforce this restriction using a checking gadget under $\NN$-semantics.

\paragraph{Endpoint Checks.}
\label{par:endpoint_checks}
Let $V$ be a $(d+2)$-VAS with transition set $T$, and let $C$ be a transformation-area threshold for $\enc{V}_B$.
For each $\Vec p\in[0,B]^2$, define $\Vec g_{\Vec p}^C:=(\Vec p,C\cdot\Vec1)\in\NN^{d+2}$ and a checking path
\[
    \chi_{\Vec p}^C:=(-\Vec g^C_{\Vec p})(+\Vec g^C_{\Vec p}).
\]
Under $\NN$-semantics, this gadget is enabled from $\Vec z\in\NN^{d+2}$ exactly when $\Vec z\ge\Vec g^C_{\Vec p}$.
Given a $\ZZ$-LPS over $T$, we place $\chi_{\Vec p}^C$ and $\chi_{\Vec q}^C$ at its two ends under $\NN$-semantics.
The checks ensure that, after subtracting a common nonnegative offset on the first two coordinates, the endpoints of the middle $\ZZ$-run lie in $\GTA_C$.
Thus~\cref{thm:generalized_transformation_area} applies to every witness of the resulting mixed LPS.
This yields the following lemma.
In what follows, we only consider mixed LPS over $T\cup\{\pm\Vec g^C_{\Vec p}:\Vec p\in[0,B]^2\}$.

\begin{lemma}
\label[lemma]{lem:mixed_lps_transformation_area}
If $\Vec p(\Vec x)\xrightarrow{*}_{\ZZ}\Vec q(\Vec y)$ in $\enc{V}_B$ with $\Vec p(\Vec x), \Vec q(\Vec y) \in \GTA_C$, then some mixed LPS $\Lambda$ satisfies
\begin{equation*}
    (\Vec p, \Vec x)\xrightarrow{\Lambda} (\Vec q,\Vec y),\qquad
    \Reach_\Lambda\subseteq\Reach_V,\qquad
    \norm\Lambda,|\Lambda|\le\PolyF{B+C+\Size(V)}.
\end{equation*}
Every $(\Vec s,\Vec t)\in\Reach_\Lambda$ satisfies $\MinLen_V(\Vec s,\Vec t)\le\MinLen_\Lambda(\Vec s,\Vec t)$.
\end{lemma}

\begin{proof}
The $\ZZ$-reachability relation of a VASS $\enc{V}_B$ is represented by finitely many LPS of length at most $\PolyF{\Size(\enc{V}_B)}$~\cite[Lemma~4.10 and Proposition~4.11]{DBLP:journals/jacm/BlondinEFGHLMT21}.
Choose such an LPS $\Gamma$ from $\Vec p$ to $\Vec q$ with $\Vec x\xrightarrow{\Gamma}_{\ZZ}\Vec y$.
We restore its original transition effects in $\ZZ^{d+2}$ and set
\begin{equation*}
    \Lambda:=(\chi_{\Vec p}^C,\NN)(\Gamma,\ZZ)(\chi_{\Vec q}^C,\NN).
\end{equation*}
Both checks are enabled at the given endpoints, so $(\Vec p,\Vec x)\xrightarrow{\Lambda}(\Vec q,\Vec y)$.

Consider any witness $w$ of $\Lambda$, with middle $\ZZ$-run $\rho$.
Since $\Gamma$ comes from paths from $\Vec p$ to $\Vec q$ in $\enc{V}_B$, the checks ensure that its endpoints have the form $(\Vec p+\Vec\eta,\Vec x')$ and $(\Vec q+\Vec\eta,\Vec y')$, where $\Vec\eta\in\NN^2$ and $\Vec x',\Vec y'\ge C\cdot\Vec1$.
Subtracting $(\Vec\eta,\Vec0)$ gives a $\ZZ$-run in $\enc{V}_B$ between two configurations in $\GTA_C$.
The threshold property yields a run of $V$ between these configurations of length $|\rho|$.
Adding back $(\Vec\eta,\Vec0)$ by monotonicity gives a run between the original endpoints of length $|\rho|\le |w|$.
This proves the containment and the distance inequality.
Since $|\Lambda|\le |\Gamma|+4$ and $\norm\Lambda\le\max\{\norm\Gamma,B,C\}$, the parameters of $\Lambda$ are polynomial in $B+C+\Size(V)$.
\end{proof}

\subsection{Linear Factorization of \texorpdfstring{$B$}{B}-Encodings of VAS}
\label{subsec:encoding_factorability}
Assume $I=\{1,2\}$ and put $W:=\enc V_B$.
Set $C:=(d+2)(2B+1)^2\norm T$, where $T$ is the transition set of $V$.
All LPS below are over $T\cup\{\pm\Vec g^C_{\Vec p}:\Vec p\in[0,B]^2\}$ as in \cref{par:endpoint_checks}.

By~\cref{lem:coordinate_extraction}, we can choose $H\ge C$, polynomial in $B+\Size(V)$, with the following property for both $W$ and its reverse.
If every counter reaches at least $H$ somewhere along a run, then deleting cycles yields a run of length at most $H$ that ends in the same state with all counters at least $C$.

\subsubsection{Normalization at High Peak Positions}
\label{subsubsec:high_peak_normalization}
A \emph{position} of a run $\pi$ is an occurrence of a configuration along it.
A position is \emph{high} if every counter reaches at least $H$ both in the prefix ending there and in the suffix starting there.
If $\pi$ visits a \hyperref[par:overview_peak_states]{peak state} at a high position, Rackoff's extraction gives short runs covering $C\cdot\Vec{1}$ at that state, yielding configurations in $\GTA_C$.
Cycle relocation then collects the removed cycles into a middle $\ZZ$-run.

\begin{lemma}
\label[lemma]{lem:high_peak_normalization}
Let $\Vec p(\Vec x)\xrightarrow{\pi}\Vec q(\Vec y)$ be a run in $W$ that visits a peak state at a high position.
Then some mixed LPS $\Lambda$ satisfies $(\Vec p, \Vec x)\xrightarrow{\Lambda}(\Vec q, \Vec y)$, $\Reach_\Lambda\subseteq\Reach_V$, and $\norm\Lambda,|\Lambda|\le\PolyF{B+\Size(V)}$.
Every $(\Vec u,\Vec v)\in\Reach_\Lambda$ satisfies $\MinLen_V(\Vec u,\Vec v)\le\MinLen_\Lambda(\Vec u,\Vec v)$.
\end{lemma}

\begin{proof}
    Choose such a high position, let $\Vec p^*$ be its state, and split $\pi=\pi^-\pi^+$ there.
    By the definition of $H$ and applying~\cref{lem:coordinate_extraction} to $\pi^-$ and to the reverse of $\pi^+$, we obtain
    \begin{equation*}
        \Vec p(\Vec x)\xrightarrow{\alpha}\Vec p^*(\Vec z),\qquad
        \Vec p^*(\Vec u)\xrightarrow{\delta}\Vec q(\Vec y),\qquad
        |\alpha|,|\delta|\le H,\qquad \Vec z,\Vec u\ge C\cdot\Vec1.
    \end{equation*}
    Both runs are obtained by deleting cycles from the corresponding segments.
    Decompose the deleted cycles into simple cycles.
    These cycles have total effect $\Vec u-\Vec z$.
    Since $\Vec p^*$ is a peak state of $\pi$, the construction in \cref{lem:cycle_relocation} rotates and relocates them to $\Vec p^*$ in $\enc V_{2B}$.
    Concatenating them gives a $\ZZ$-run $\Vec p^*(\Vec z)\xrightarrow{\tau}_{\ZZ}\Vec p^*(\Vec u)$.

    Apply \cref{lem:mixed_lps_transformation_area} with encoding bound $2B$ and threshold $C$ to obtain
    \begin{equation*}
        \Theta=(\chi_{\Vec p^*}^C,\NN)(\Gamma,\ZZ)(\chi_{\Vec p^*}^C,\NN),
        \qquad \Lambda:=(\alpha,\NN)\Theta(\delta,\NN).
    \end{equation*}
    Then $(\Vec p,\Vec x)\xrightarrow{\Lambda}(\Vec q,\Vec y)$ holds.
    For every witness of $\Lambda$, the $\alpha$- and $\delta$-components are runs of $V$, and the $\Theta$-component yields a run of no greater length by \cref{lem:mixed_lps_transformation_area}.
    This proves containment and the distance inequality.
    The bounds on $\alpha,\delta,C$ and $\Gamma$ give the parameter bounds.
\end{proof}

\subsubsection{Recursive Factorization}
\label{subsubsec:linear_factorization}
We first isolate the interval of high positions, leaving a prefix and suffix that each have a bounded counter.
If this interval contains a peak state, the preceding lemma yields a mixed-LPS factor.
Otherwise, we either obtain bounded-counter factors directly or recurse on the middle subrun, whose maximum first encoded coordinate is strictly smaller.

\begin{lemma}
    \label{lem:high_interval_decomposition}
    Every run $\pi$ in $W$ can be partitioned into $\pi=\alpha\pi'\beta$ such that:
    \begin{enumerate}
        \item Every nonempty path among $\alpha,\beta$ has a counter that stays below $H+\norm T$.
        \item Every position of $\pi'$ is high in $\pi$, unless $\pi'$ is empty.
    \end{enumerate}
\end{lemma}

\begin{proof}
	Number the positions of $\pi$ from $0$ to $\ell$.
	If some counter stays below $H$ throughout $\pi$, take $\alpha=\pi$ and let $\pi',\beta$ be empty at its final configuration.
	Otherwise, let $a_j,b_j$ be the first and last positions where counter $j$ reaches at least $H$, and set $a:=\max_j a_j$ and $b:=\min_j b_j$.
	The high positions are precisely $[a,b]$ when $a\le b$, and there are none otherwise.
    Choose $j_-$ and $j_+$ attaining $a$ and $b$, respectively.
    Let $\alpha$ be the subrun from position $0$ to $a$, let $\pi'$ be empty if $a>b$ and the subrun from $a$ to $b$ otherwise, and let $\beta$ be the subrun from $\max\{a,b\}$ to $\ell$.
    By construction, every position of $\pi'$ is high in $\pi$ whenever $\pi'$ is nonempty.
    Moreover, counter $j_-$ stays below $H$ before position $a$, while counter $j_+$ stays below $H$ after position $b$.
    Since each transition changes a counter by at most $\norm T$, every nonempty path among $\alpha,\beta$ has a counter bounded by $H-1+\norm T$, as required.
\end{proof}

We now combine the preceding two lemmas to prove the linear factorization.

\begin{proof}[Proof of \cref{lem:encoding_factorability}]
	Fix a run $\pi$ in $W$.
    For a run $\rho$ in $W$, let $\operatorname{peak}(\rho)\in[0,B]$ be the maximum first coordinate of its states.
    If $\pi$ is empty, it is captured by an empty $\NN$-LPS, so we assume that $\pi$ is nonempty.
    Put $D:=\max\{B,H-1+\norm T\}$ and set $\rho_0:=\pi$.
    We recursively process $\rho_i$ as follows.

    At stage $i$, apply~\cref{lem:high_interval_decomposition} to obtain $\rho_i=\alpha_i\rho_{i+1}\beta_i$, where $\rho_{i+1}$ spans the high interval of $\rho_i$, or is empty if no high position exists.
    We distinguish two cases:
    \begin{itemize}
        \item If the high interval contains a peak state of $\rho_i$, then~\cref{lem:high_peak_normalization} gives a mixed LPS $\Lambda$ containing the endpoints of $\rho_i$ in its reachability relation, with
        \[
            \norm\Lambda,|\Lambda|\le\PolyF{B+\Size(V)},
            \qquad
            \Reach_\Lambda\subseteq\Reach_V,
        \]
        and $\MinLen_V(\Vec x,\Vec y)\le\MinLen_\Lambda(\Vec x,\Vec y)$ for all $(\Vec x,\Vec y)\in\Reach_\Lambda$.
        We record $\Reach_\Lambda$ as one factor and terminate.

        \item Otherwise, each nonempty path among $\alpha_i,\beta_i$ has some counter $j$ bounded by $D$, and is therefore captured by $\enc V_{I\cup\{j\},D}$.
        We record the corresponding factor
        \[
            R=\Reach_{\enc V_{I\cup\{j\},D}}\subseteq\Reach_V,
        \]
        which satisfies $\MinLen_V(\Vec x,\Vec y)\le\MinLen_{\enc V_{I\cup\{j\},D}}(\Vec x,\Vec y)$ for all $(\Vec x,\Vec y)\in R$.
        If $\rho_{i+1}$ is empty, we terminate.
        Otherwise, we continue with $\rho_{i+1}$.
    \end{itemize}

    At every nonterminal stage, the high interval contains no peak state of $\rho_i$, so
    \[
        \operatorname{peak}(\rho_{i+1})
        <\operatorname{peak}(\rho_i).
    \]
    Since $\operatorname{peak}(\rho_i)\in[0,B]$, there are at most $B$ nonterminal stages.
    Each such stage contributes at most two bounded-counter factors, while the terminal stage contributes at most two factors.
    Thus the resulting composition has at most $2B+2$ factors, ordered according to their corresponding subruns.
	The factors produced above have uniformly bounded parameters, including at most five components per mixed LPS, so they belong to a finite family $\mathcal R$.
    
	Include in $\mathcal R$ the identity relation represented by the empty $\NN$-LPS, and choose a common polynomial $P$ for the parameter bounds.
	Let $\mathcal F$ contain all nonempty sequences of at most $2B+2$ relations from $\mathcal R$.
    The recursive construction gives coverage of $\Reach_{\enc V_{I,B}}$, and transitivity of $\Reach_V$ gives the other inclusion.
\end{proof}

The first inclusion in \cref{lem:encoding_factorability} may be strict: cycle relocation uses the enlarged encoding $\enc V_{2B}$, and reconstruction in $V$ need not preserve the original bound $B$.
This suffices for the polynomial factorization, since every resulting factor is contained in $\Reach_V$.

\subsection{Applications of Polynomial Factorization}
\label{subsec:factorization_applications}

\paragraph{Mixed LPS for \texorpdfstring{$5$}{5}-VAS Reachability.}
\label{subsubsec:five_vas_mixed_lps}
In a $5$-VAS, encoding the three bounded counters of a factor into states leaves a $2$-VASS of polynomial size.
By~\cite[Theorem~3.1]{DBLP:journals/jacm/BlondinEFGHLMT21}, its reachability relation is covered by finitely many $\NN$-LPS with polynomially bounded parameters.
The decoded LPS cover the factor, and each defines a subrelation of $\Reach_V$.
Concatenating these LPS with the mixed-LPS factors proves \cref{cor:flatness_of_5vas}.
See \cref{app:mixed_flattenability} for the full proof.

Fixing which cycles occur in a mixed-LPS witness yields a system $\mathcal E_\Lambda$ of linear equations over nonnegative integers, with the source and target vectors as variables.
The system has at most $\PolyF{\Size(V)}$ variables and equations, and its coefficients and constants have norm at most $\PolyF{\Size(V)}$.
By~\cref{thm:pottier}, its solution set is a finite union of linear sets whose bases and periods have norm at most $2^{\PolyF{\Size(V)}}$.
Projecting onto the source and target coordinates and taking the union over all mixed LPS and cycle choices yields
\[
    \Reach_V=\bigcup_{i\in\mathcal I}(\Vec b_i+P_i^*),
    \qquad
    \norm{\Vec b_i},\norm{P_i}\le 2^{\PolyF{\Size(V)}}.
\]
Since the family is effectively enumerable, $\Reach_V$ is effectively semilinear.

Now let $M:=\Size(V,\Vec s,\Vec t)$, and let $n$ be the binary input length.
Fixing $\Vec s, \Vec t$ in $\mathcal E_\Lambda$ gives a system $\mathcal E_{\Lambda,\Vec s, \Vec t}$ with polynomially bounded size and coefficients in $M$.
By~\cref{thm:pottier}, a feasible system has a solution of binary length polynomial in $M$, yielding a $\NEXPT$ algorithm since $M$ is at most exponential in $n$.
Moreover, the cycle repetition counts are at most $2^{\PolyF{M}}$, and every witness yields a run in $V$ of no greater length, so $\MinLen_V(\Vec s,\Vec t)\le 2^{\PolyF{M}}$.
The $\PSPACE$ upper bound uses the stronger polynomial reachability-distance bound of~\cref{thm:polynomial_witness}, proved in the following sections.

\paragraph{Reachability in \texorpdfstring{$(d+3)$}{(d+3)}-VAS Is in \texorpdfstring{$\mathsf F_d$}{Fd}.}
\label{subsubsec:higher_dimensional_complexity}
For a $(d+3)$-VAS, encode the three bounded counters of each bounded-counter factor into states while preserving all original counter effects.
Every cycle returns to the same encoded state, so it has zero effect on these three counters and the resulting gadget has cycle-dimension at most $d$.
The bound of~\cite[Theorem~2.1]{DBLP:conf/concur/ChenFZ26} concerns the maximum cycle-dimension of an SCC, so the mixed-LPS factors must also be converted into VASS gadgets with controlled cycle-dimension.
To build a VASS gadget, give each path and cycle of the LPS its own copy and handle its $\NN$- and $\ZZ$-components differently:
\begin{itemize}
    \item an $\NN$-component needs no further change, and
    \item a $\ZZ$-component is padded by two self-loops, one raising ($+\mathbf{1}$) and one restoring ($-\mathbf{1}$) all original counters, with an auxiliary counter forcing the two repetition counts to be equal.
\end{itemize}
The auxiliary counter adds a counter, but each SCC of the resulting gadget has at most one simple cycle, so its cycle-dimension is at most $1$.
Connecting the gadgets in sequence keeps every SCC inside a single gadget, so the compiled VASS has component-dimension at most $d$, and the $\mathsf{F}_d$ bound of~\cite[Theorem~2.1]{DBLP:conf/concur/ChenFZ26} applies.
Since every witness becomes a run of $V$ of no greater length, the bound transfers back and proves \cref{thm:higher_dimensional_complexity} for $d\ge3$.
The full proof is in \cref{app:higher_dimensional_complexity}.

%% file: contents/6-Coverable.tex
\section{CFZ Characterization and Basic Reachability-Distance Bounds}
\label{sec:basic_reachability_bounds}

We now turn to the proof of~\cref{thm:polynomial_witness} starting from a decomposition of $5$-VAS reachability into two classes of instances.
Recall that for a $d$-VASS $V=(Q,T)$, the encoding $\enc{V}_{I,B}$ has state set $Q'=Q\times[0,B]^I$ and removes the coordinates in $I$.
We also use $\widehat{\enc V}_{I,B}$, with the same state set but retaining all $d$ coordinates.
For every $(p,\Vec a,q) \in T$ and $p(\Vec u)\in Q'$, it has the transition
\[
    \bigl(p(\Vec u),\Vec a, q(\Vec v)\bigr)
    \quad\text{precisely when}\quad
    \Vec v=\Vec u+\restr{\Vec a}{I}\in [0,B]^I.
\]
Retaining the original counters allows two components encoding different coordinate sets to share the same configuration across a zero-effect transition.
From a configuration $p(\Vec{u})(\Vec{x})$ satisfying $\Vec x|_I=\Vec u$, every reachable configuration also preserves this equality.

\paragraph{The CFZ Characterization.}
\label{par:cfz_characterization}
Based on these notions, we now instantiate~\cite[Proposition 5.1]{DBLP:conf/concur/ChenFZ26} for $d=5$ as follows, and refer the reader to its full version for the construction.

\begin{theorem}[CFZ Characterization~{\cite[Proposition~5.1]{DBLP:conf/concur/ChenFZ26}}]
	\label{thm:cfz}
	For every $5$-VAS $(V,\Vec{s},\Vec{t})$ with $\Vec{s}\xrightarrow{*}\Vec{t}$, there exists $B \le \Size(V,\Vec{s},\Vec{t})^{\bigO(1)}$ and a VASS $(V',s',t')$ such that $s'\xrightarrow{*}t'$ in $V'$ and $\MinLen_{V}(\Vec{s},\Vec{t}) \le \MinLen_{V'}(s',t')$.
	Moreover, one of the following alternatives holds:
	\begin{caselist}
		\item\label{item:cfz_two_components}
		$V'$ consists of two disjoint $5$-VASS $\widehat{\enc{V''}}_{I_1,B},\widehat{\enc{V''}}_{I_2,B}$, connected by a zero-effect transition from the first to the second.
		Here, $V''$ is a $5$-VASS of size $\Size(V,\Vec{s},\Vec{t})^{\bigO(1)}$, and for each $i = 1,2$, we have $|I_i| = 2$ and $\dimcyc(\widehat{\enc{V''}}_{I_i,B})\le2$.
		The source and target configurations of $V'$ are $s' = (p, \restr{\Vec{s}}{I_1})(\Vec{s})$ and $t' = (q, \restr{\Vec{t}}{I_2})(\Vec{t})$, where $p$ and $q$ are two states in $V''$.
		\item\label{item:cfz_encoding}
		$(V',s',t')$ is a single $B$-encoding VAS $(\enc{V}_{B}, \enc{\Vec s}_B, \enc{\Vec t}_B)$ for some $B\le\Size(V,\Vec s,\Vec t)^{\bigO(1)}$.
	\end{caselist}
\end{theorem}

The constant hidden in the $\bigO(1)$ exponent above is independent of the input.
It suffices to establish polynomial reachability-distance bounds for these two classes.

\subsection{Case I: Two Components with Cycle-Dimension at Most 2}
\label{subsec:two_component_case}

We now treat the two cases separately, beginning with~\cref{item:cfz_two_components}.
For clarity, we fix a $5$-VASS $V''$, $I_1, I_2 \subseteq [5]$ containing two elements, and $B \in \NN$, and consider the reachable instance $(V',s',t')$ obtained from $V''$ by the construction in~\cref{item:cfz_two_components}.
In this part, we establish the following proposition.

\begin{restatable}{proposition}{lenCaseI}
	\label{prop:two_component_reachability}
	Every reachable instance $(V',s',t')$ in~\cref{item:cfz_two_components} satisfies 
	\[\MinLen_{V'}(s',t')\le (B+\Size(V',s',t'))^{\bigO(1)}.\]
\end{restatable}

It is crucial to observe that both $\widehat{\enc{V''}}_{I_1,B}$ and $\widehat{\enc{V''}}_{I_2,B}$ have cycle-dimension at most $2$.
Removing the coordinates already recorded in the states yields a $3$-VASS $\enc{V''}_{I_i,B}$ with cycle-dimension at most two.
We use two existing results for such systems.

For a $d$-VASS $V$ and a state $q$, write $\Reach^q_V(s) := \{\Vec{x} : q(\Vec{x}) \in \Reach_V(s)\}$ for the slice of $\Reach_V(s)$ at state $q$.
More generally, for $X\subseteq\NN^d$, we define $\Reach^q_V(p(X)) := \bigcup_{\Vec{x} \in X} \Reach_V^q(p(\Vec{x}))$, namely, the set of vectors reachable at $q$ from some configuration $p(\Vec{x})$ with $\Vec{x}\in X$.
The first result represents configurations reachable from a hybrid-set source by a finite union of hybrid sets with polynomially bounded representation norms.

\begin{theorem}[Efficient Hybrid-Set Representations~{\cite[Theorem~4.2]{DBLP:journals/corr/abs-2607-14983}}]
	\label{thm:hybrid_description}
	Let $V$ be a $3$-VASS with $\dimcyc(V)\le2$.
	Let $H=\Vec{b}+P^{\langle\Vec{c},\cdot\rangle\le K}$ be a hybrid set whose periods satisfy $P\subseteq\CycleSpace(V)$.
	For any states $p,q$, the set $\Reach_V^q(p(H))$ is a finite union of hybrid sets $H'$ satisfying
	\begin{equation*}
		\hnorm(H')\le\PolyF{\Size(V)+\hnorm(H)}.
	\end{equation*}
\end{theorem}

To bound $\MinLen_V(s,t)$, we also need the following reachability-distance bound for $3$-VASS of cycle-dimension at most $2$ with a hybrid source.

\begin{lemma}[{\cite[Lemma~4.3]{DBLP:journals/corr/abs-2607-14983}}]
	\label{lem:hybrid_source}
	Let $V$ be a $3$-VASS with $\dimcyc(V)\le2$, $t$ be a configuration, and $H$ be a hybrid set.
	If $p(\Vec{x})\xrightarrow{*} t$ for some $\Vec{x}\in H$, there exists $\Vec{x}'\in H$ with
	\begin{equation*}
		\norm{\Vec{x}'},\MinLen_V(p(\Vec{x}'), t)\le\PolyF{\Size(V,t)+\hnorm(H)}.
	\end{equation*}
\end{lemma}

Together, these results allow us to describe the reachable midpoints of the first component by hybrid sets and select a midpoint admitting a short run through the second.
Since the two components may encode different coordinates when $I_1\ne I_2$, we must ensure that the selected midpoint is compatible with both encodings.

\begin{proof}[Proof of~\cref{prop:two_component_reachability}]
	Set $M:=B+\Size(V',s',t')$.
	For $i=1,2$, write $J_i=[5]\setminus I_i$ and $V_i:=\enc{V''}_{I_i,B}$ for the $3$-dimensional projection of the $i$-th component.
	Every cycle in $V_i$ lifts to a cycle in the original component, with zero effect on $I_i$.
	Hence $\dimcyc(V_i)\le2$.
	Let $s_1,t_2$ be the projections of $s',t'$ into $V_1,V_2$, respectively.
	Write $(r_1(\Vec u_1), \Vec{0}, r_2(\Vec u_2))$ for the zero-effect bridge, where $\Vec u_i\in[0,B]^{I_i}$.
	It is easy to see that, for any run in $(V',s',t')$, the configurations at the two ends of the bridge share the same vector $\Vec{w}$, with $\restr{\Vec w}{I_1}=\Vec u_1$ and $\restr{\Vec w}{I_2}=\Vec u_2$.

	Apply~\cref{thm:hybrid_description} to $V_1$, with the singleton source specified by $s_1$.
	This set satisfies the cycle-space hypothesis.
	Choose a hybrid set in the resulting union at state $r_1(\Vec u_1)$ that contains $\restr{\Vec w}{J_1}$.
	By inserting $\Vec u_1$ into its base on $I_1$ and zeros into its periods on $I_1$, we obtain a $5$-dimensional hybrid set
	\begin{equation*}
		H=\Vec b+P^{\langle\Vec c,\cdot\rangle\le K}, \qquad \hnorm(H)\le\PolyF{M},
	\end{equation*}
	containing $\Vec w$, all of whose vectors are reachable at the bridge through the first component.
	Then we impose the second restriction $\restr{\Vec x}{I_2}=\Vec u_2$ before projecting to $J_2$.

	Fix a representation $\Vec w=\Vec b+\sum_{\Vec p\in P}\lambda_{\Vec p}\cdot\Vec p$ with $\sum_{\Vec p\in P}\Vec c(\Vec p)\lambda_{\Vec p}\le K$, where $\Vec c(\Vec p)$ denotes the cost of period $\Vec p$.
	Partition $P$ into $P_0=\{\Vec p\in P:\restr{\Vec p}{I_2}=\Vec0\}$ and $P_1=P\setminus P_0$.
	Each period in $P_1$ contributes at least $1$ to the sum of the coordinates in $I_2$.
	Since all bases and periods are nonnegative, we have
	\begin{equation*}
		\sum_{\Vec p\in P_1}\lambda_{\Vec p} \le \sum_{j\in I_2}\Vec w(j)\le2 B.
	\end{equation*}
	Absorb these boundedly many period occurrences into the base by setting $\Vec b_0:=\Vec b+\sum_{\Vec p\in P_1}\lambda_{\Vec p}\cdot\Vec p$ and $K_0:=K-\sum_{\Vec p\in P_1}\Vec c(\Vec p)\lambda_{\Vec p}\ge0$.
	Then the resulting hybrid set
	\begin{equation*}
		H_0:=\Vec b_0+P_0^{\langle\restr{\Vec c}{P_0},\cdot\rangle\le K_0}\subseteq H,\qquad
		\hnorm(H_0)\le (2B + 1)\hnorm(H)\le\PolyF{M}.
	\end{equation*}
	The set $H_0$ contains $\Vec w$, and every $\Vec x\in H_0$ satisfies $\restr{\Vec x}{I_2}=\Vec u_2$.
	Its projection $H_2:=\{\restr{\Vec x}{J_2}:\Vec x\in H_0\}$ is therefore a hybrid set of no larger norm.
	The original suffix shows that some vector in $H_2$ reaches $t_2$ from state $r_2(\Vec u_2)$ in $V_2$.
	By~\cref{lem:hybrid_source}, there are $\Vec z\in H_2$ and a run in $V_2$ satisfying
	\begin{equation*}
		r_2(\Vec u_2)(\Vec z)\xrightarrow{\rho_2}t_2,
		\qquad \norm{\Vec z},|\rho_2|\le\PolyF{M}.
	\end{equation*}
	Insert $\Vec u_2$ into $\Vec z$ to obtain its unique lift $\Vec w'\in H_0$.
	Then $\norm{\Vec w'}\le\max\{B,\norm{\Vec z}\} \le \PolyF{M}$.

	Since $\Vec w'\in H_0\subseteq H$, the configuration $r_1(\Vec u_1)(\restr{\Vec w'}{J_1})$ is reachable from $s_1$ in $V_1$.
	Apply~\cref{lem:hybrid_source} once more, now with the singleton source specified by $s_1$ and this fixed target.
	A singleton leaves no choice of source vector, so this gives a run $\rho_1$ from $s_1$ to that target with $|\rho_1|\le\PolyF{M+\norm{\Vec w'}}\le\PolyF{M}$.
	Both runs lift to their respective five-dimensional components without changing their lengths.
	They meet at the same vector $\Vec w'$, so the zero-effect bridge joins them.
	Consequently,
	\begin{equation*}
		\MinLen_{V'}(s',t')\le |\rho_1|+1+|\rho_2|\le M^{\bigO(1)},
	\end{equation*}
	which completes the proof.
\end{proof}

\subsection{Case II: A Single Encoding VAS}
\label{subsec:encoding_case}

We now consider~\cref{item:cfz_encoding}.
Let $W=\enc{V}_B$ be the $B$-encoding of a $5$-VAS $V$, and fix a run $\Vec p(\Vec x)\xrightarrow{\pi}\Vec q(\Vec y)$ in $W$, where $\Vec s = (\Vec{p},\Vec x)$ and $\Vec t = (\Vec{q},\Vec y)$.
The SCCs traversed by $\pi$, together with their connecting transitions, form a sequential sub-VASS $U=(U_1)u_1(U_2)\cdots u_{k-1}(U_k)$.
Since $W$ has $(B+1)^2$ states, we have $k\le(B+1)^2$.
Our goal is to establish a polynomial bound on $\MinLen_V(\Vec s,\Vec t)$.

\begin{restatable}{proposition}{generalEncodingBound}
	\label{prop:encoding_reachability}
	Let $V$ be a $5$-VAS and let $(U,\Vec{p}(\Vec{x}), \Vec{q}(\Vec{y}))$ be a reachable instance specified above.
	Then
	\begin{equation}
		\MinLen_V(\Vec s,\Vec t)
		\le \bigl(B+\Size(V,\Vec s,\Vec t)\bigr)^{\bigO(1)}.
		\label{eq:encoding_reachability_length}
	\end{equation}
\end{restatable}

Together with~\cref{thm:cfz,prop:two_component_reachability}, \cref{prop:encoding_reachability} implies~\cref{thm:polynomial_witness}.
The formal deduction is given in~\cref{app:polynomial_witness}.
We first establish polynomial reachability-distance bounds for two basic cases: non-wide instances and $C$-coverable instances as specified in~\cref{sec:tech_overview}.
We formalize the related notions below.

\begin{definition}[Wideness]
	\label{def:wideness}
	Let $U$ be a $d$-VASS.
	Let $E(U)$ be the set of effects of simple cycles in $U$, and let $\Cone(U):=\Cone(E(U))$.
	A $d$-VASS $U$ is called \emph{wide} if $\Cone(U)=\QQ^d$.
\end{definition}

This notion differs from the wideness conditions of~\cite{DBLP:conf/icalp/CzerwinskiJ0O25,DBLP:journals/corr/abs-2607-14983}, which use positive-orthant containment $\QQ_{\ge 0}^3 \subseteq \Cone(U)$.
\cref{def:wideness} considers the full space $\QQ^3$ because it permits compensation in every direction.

\begin{definition}[$C$-Coverability]
	\label{def:c_coverability}
	Let $(U,p(\Vec{x}), q(\Vec{y}))$ be a $d$-VASS.
	For $C\in\NN$, the source $p(\Vec x)$ is \emph{forward $C$-coverable} in $U$ if $p(\Vec x)\xrightarrow{*}p(\Vec u)$ for some $\Vec u\ge C\cdot\Vec1$.
	The target $q(\Vec y)$ is \emph{backward $C$-coverable} if $q(\Vec v)\xrightarrow{*}q(\Vec y)$ for some $\Vec v\ge C\cdot\Vec1$.
	The VASS $(U,p(\Vec{x}), q(\Vec{y}))$ is \emph{$C$-coverable} if both conditions hold.
\end{definition}

\subsubsection{Non-Wide Case}
\label{subsubsec:nonwide_encoding}

The non-wide case follows the standard argument of~\cite[Lemma~9]{DBLP:conf/icalp/CzerwinskiJ0O25}.
If $\dimcyc(U)\le2$, the preceding low-dimensional result applies directly.
Otherwise a proper cycle cone is contained in a half-space with a nonzero normal.
The inner product with this normal is bounded along a given run.
Recording that inner product in the state reduces the entire system to cycle-dimension at most two.
We state the reduction below and defer its proof to~\cref{app:proper_cone_trim}.

\begin{lemma}
	\label{lem:proper_cone_trim}
	Let $(U,s,t)$ be a reachable $3$-VASS with $\Cone(U)\ne\QQ^3$.
	There exists a reachable $3$-VASS $(\widetilde U,\widetilde s,\widetilde t)$ of size at most $\Size(U,s,t)^{\bigO(1)}$ such that $\dimcyc(\widetilde U)\le2$ and $\MinLen_U(s,t)\le\MinLen_{\widetilde U}(\widetilde s,\widetilde t)$.
\end{lemma}

Apply~\cref{lem:hybrid_source} from the previous subsection to $\widetilde U$ with a singleton $\{\widetilde s\}$.
It gives a run from $\widetilde{s}$ to $\widetilde t$ with length at most $\PolyF{\Size(\widetilde U,\widetilde s,\widetilde t)}$, hence the following.

\begin{corollary}
	\label{cor:proper_cone_len}
	Let $(U,s,t)$ be a reachable $3$-VASS with $\Cone(U)\ne\QQ^3$.
	Its reachability distance satisfies
	\begin{equation*}
		\MinLen_U(s,t)\le\Size(U,s,t)^{\bigO(1)}.
	\end{equation*}
\end{corollary}

\subsubsection{Wide and Coverable Case}
\label{subsubsec:wide_coverable_case}

We now consider the wide case, where $\Cone(U)=\QQ^3$.
Combined with the transformation-area theorem, this gives a reachability-distance bound under $C$-coverability.
We state the result more generally for $B$-encodings of $(d+2)$-VAS.

\begin{theorem}[Wide and Coverable Case]
	\label{thm:wide_coverable}
	Let $V$ be a $(d+2)$-VAS, $B \in \NN$, and $(U,\Vec{p}(\Vec{x}),\Vec{q}(\Vec{y}))$ be a sequential sub-VASS of $\enc V_B$.
	Assume that the source and target lie in the first and last components of $U$, respectively.
	If $(U,\Vec{p}(\Vec{x}),\Vec{q}(\Vec{y}))$ is wide and $C$-coverable and $\Vec{p}(\Vec{x}) \xrightarrow{*}_\ZZ \Vec{q}(\Vec{y})$ in $U$, then
	\begin{equation}
		\MinLen_V(\Vec s,\Vec t)
		\le \bigl(B + C + \Size(V,\Vec s,\Vec t)\bigr)^{(d+1)^{\bigO(d)}},
		\label{eq:wide_coverable_length}
	\end{equation}
	where $C$ is a transformation-area threshold for $\enc{V}_{B}$, $\Vec{s}=(\Vec{p},\Vec x)$ and $\Vec{t}=(\Vec{q},\Vec y)$.
\end{theorem}

For $d=3$, the threshold supplied by~\cref{thm:generalized_transformation_area} is polynomial in $B+\Size(V)$, giving the required polynomial reachability-distance bound.

\paragraph{Short Runs Between Configurations in \texorpdfstring{$\GTA_C$}{GC}.}
\label{par:short_runs_transformation_area}
The transformation-area theorem converts a $\ZZ$-run between configurations in $\GTA_C$ into a run of the same length, but does not bound that length.
We therefore use the following shortening result for $\ZZ$-runs, whose proof is given in~\cref{app:short_integer_vass}.

\begin{restatable}{lemma}{shortIntVASS}
	\label{lem:short_integer_vass}
	Let $U = (Q, T)$ be a $d$-VASS.
	If $p(\Vec x)\xrightarrow{*}_{\ZZ}q(\Vec y)$ in $U$, there is such a $\ZZ$-run of length at most $(\Size(U) + \norm{\Vec{x}} + \norm{\Vec{y}})^{\PolyF{d}}$.
\end{restatable}

By~\cref{lem:short_integer_vass}, every $\ZZ$-run in $\enc{V}_B$ can be shortened to polynomial length in the size of the instance for fixed $d$.
If its endpoints lie in $\GTA_C$, \cref{thm:generalized_transformation_area} then converts the shortened $\ZZ$-run into a run in $V$ with the same Parikh image, and hence the same length.
This gives the following corollary.

\begin{corollary}
    \label{cor:short_high_encoding}
    Let $V$ be a $(d+2)$-VAS, $B\in\NN$, and $U$ a sub-VASS of $\enc{V}_B$.
    Let $C$ be a transformation-area threshold for $\enc{V}_B$.
    If $\Vec p(\Vec x)\xrightarrow{*}_{\ZZ}\Vec q(\Vec y)$ in $U$ and $\Vec x,\Vec y\ge C\cdot\Vec1$, then
    \begin{equation*}
        \MinLen_V\bigl((\Vec p,\Vec x),(\Vec q,\Vec y)\bigr)
        \le
        \bigl(\Size(U)+\norm{\Vec x}+\norm{\Vec y}\bigr)^{\PolyF{d}}.
    \end{equation*}
\end{corollary}

\paragraph{Compatibility of the Covering Configurations.}
\label{par:compatible_covering_configurations}

By~\cref{cor:short_high_encoding}, we can obtain a short run once the endpoints of a $\ZZ$-run lie in $\GTA_C$.
For the $C$-coverable case, we use~\cref{lem:rackoff_cover} to extract short covering runs.
However, the resulting configurations in $\GTA_C$ need not be reachable by the original $\ZZ$-run.
The following lemma uses wideness to provide a $\ZZ$-run by inserting suitable cycles.

\begin{lemma}
	\label{lem:wide_compatible_covers}
	Let $U=(U_1)u_1\cdots u_{k-1}(U_k)$ be wide, with $p$ in $U_1$ and $q$ in $U_k$.
	Suppose that
	\begin{equation*}
		p(\Vec x)\xrightarrow{\alpha}_{\ZZ}p(\Vec u),\qquad
		p(\Vec x)\xrightarrow{*}_{\ZZ}q(\Vec y),\qquad
		q(\Vec v)\xrightarrow{\beta}_{\ZZ}q(\Vec y)
	\end{equation*}
	in $U$.
	Then $p(\Vec u)\xrightarrow{*}_{\ZZ}q(\Vec v)$ in $U$.
\end{lemma}

\begin{proof}
	Let $E:=E(U)$ be the set of effects of simple cycles in $U$, and let $E^*$ denote the set of their nonnegative integer combinations.
	We claim that $E^*$ is closed under negation.
	Indeed, for every $\Vec e\in E$, wideness gives $-\Vec e\in\Cone(E)$.
	Since $\Cone(E)$ consists of nonnegative rational combinations of vectors in $E$, there exists an integer $m\ge1$ such that $-m\cdot\Vec e\in E^*$.
	Consequently, $-\Vec e=(-m\cdot\Vec e)+(m-1)\cdot\Vec e\in E^*$.

	Fix a $\ZZ$-run $p(\Vec x)\xrightarrow{\pi}_{\ZZ}q(\Vec y)$ in $U$.
	In each component $U_i$, choose a cycle $\omega_i$ based at a state visited by $\pi$ and traversing every state of $U_i$.
	Strong connectivity provides such a cycle.
	For a one-state component we may instead take the empty path.
	Inserting these paths into $\pi$ gives a path $\tau$.
	Let
	\begin{equation*}
		\Vec a:=\Eff(\alpha)=\Vec u-\Vec x,\qquad
		\Vec b:=\Eff(\beta)=\Vec y-\Vec v,\qquad
		\Vec h:=\sum_{i = 1}^k\Eff(\omega_i).
	\end{equation*}
	The paths $\alpha,\beta,\omega_i$ are cycles or empty paths, so their effects belong to $E^*$.
	Hence $-\Vec a-\Vec b-\Vec h \in E^*$, which implies that there exists a multiset of simple cycles with effect $-\Vec a-\Vec b-\Vec h$.

	Insert these cycles at visits of $\tau$ to their initial states, which exist because $\tau$ visits every state of $U$.
	This yields a path from $p$ to $q$, whose effect is
	\begin{equation*}
		(\Vec y-\Vec x)+\Vec h-\Vec a-\Vec b-\Vec h
		=\Vec v-\Vec u.
	\end{equation*}
	Hence, there is a $\ZZ$-run from $p(\Vec u)$ to $q(\Vec v)$ in $U$.
\end{proof}

\begin{proof}[Proof of~\cref{thm:wide_coverable}]
	Let $T$ be the transition set of $V$, $N:=(B+1)^2$, and $R:=\max\{\norm{\Vec x},\norm{\Vec y}\}$.
	Applying~\cref{lem:rackoff_cover} to the source covering run in $U$ and to the target covering run in the reversed system, we obtain two short covering runs
	\begin{equation*}
		\Vec{p}(\Vec x)\xrightarrow{\alpha}\Vec{p}(\Vec u),
		\qquad
		\Vec{q}(\Vec v)\xrightarrow{\beta}\Vec{q}(\Vec y),
	\end{equation*}
	with $\Vec u,\Vec v\ge C\cdot \Vec{1}$ and $|\alpha|,|\beta|\le L:=\bigl(\Size(U)+C\bigr)^{\bigO((d+1)!)}$.
	Consequently, $\norm{\Vec u},\norm{\Vec v}\le R+L\norm{T}$.
	By~\cref{lem:wide_compatible_covers}, there is a $\ZZ$-run $\pi$ in $U$ giving
	\begin{equation*}
		\Vec{p}(\Vec x)\xrightarrow{\alpha}\Vec{p}(\Vec u)
		\xrightarrow{\pi}_{\ZZ}\Vec{q}(\Vec v)
		\xrightarrow{\beta}\Vec{q}(\Vec y).
	\end{equation*}
	By~\cref{cor:short_high_encoding}, the $\ZZ$-run $\Vec{p}(\Vec{u}) \xrightarrow{\pi}_{\ZZ} \Vec{q}(\Vec{v})$ can be replaced by a run in $V$ of length at most $M\le(\Size(U) + \norm{\Vec{u}} + \norm{\Vec{v}})^{\PolyF{d}}$.
	The two covering runs also lift to $V$, so concatenation gives
	\begin{equation*}
		\MinLen_V(\Vec s,\Vec t)
		\le 2L+M \le 2L + (\Size(U) + 2 R + 2L\norm{T})^{\PolyF{d}}.
	\end{equation*}

	Finally, $U$ has at most $N$ states and $N|T|$ transitions, all of norm at most $\norm{T}$, whence
	\begin{equation*}
		\Size(U)\le N \Size(V).
	\end{equation*}
	Combining this with the bounds on $N$ and $C$ yields~\eqref{eq:wide_coverable_length}.
\end{proof}

\cref{thm:wide_coverable} shows that, in the wide and $C$-coverable case, even $\ZZ$-reachability suffices to obtain a short run.
Without wideness, ~\cref{lem:proper_cone_trim} only shortens an \emph{existing run}.
For a reachable sequential instance, first restrict $U$ to the components between its source and target.
The forward and backward covering runs stay in their respective components and remain available in this restriction.
Since every run in $U$ lifts to $V$ with the same length and $\Size(U)\le(B+1)^2\Size(V)$, combining the two results gives the following bound for all reachable $C$-coverable instances.

\begin{corollary}[Coverable Case]\label{cor:coverable_case}
	Let $V$ be a $5$-VAS, $B\in\NN$, and $(U,\Vec p(\Vec x),\Vec q(\Vec y))$ be a sequential sub-VASS of $\enc V_B$ such that $\Vec p(\Vec x)\xrightarrow{*}\Vec q(\Vec y)$ in $U$.
	If $(U,\Vec p(\Vec x),\Vec q(\Vec y))$ is $C$-coverable, then
	\begin{equation*}
		\MinLen_V(\Vec s,\Vec t)
		\le
		\bigl(B+C+\Size(V,\Vec s,\Vec t)\bigr)^{\bigO(1)},
	\end{equation*}
	where $C$ is a transformation-area threshold for $\enc V_B$, $\Vec s=(\Vec p,\Vec x)$, and $\Vec t=(\Vec q,\Vec y)$.
\end{corollary}

\subsection{Hybrid Sets as Sources and Targets}
\label{subsec:hybrid_sources_and_targets}

The preceding results give polynomial reachability-distance bounds for fixed configurations.
In the reductions used later, however, the source and target may range over hybrid sets.
We handle such endpoints by reducing the problem to reachability between fixed configurations in an augmented VASS.
The idea is to add the periods of the source hybrid set as positive self-loops at the source state, and those of the target hybrid set as negative self-loops at the target state.

Formally, for a VASS $U=(Q,T)$ with specified source and target states $\Vec p,\Vec q$ and period sets $P_1,P_2\subseteq\NN^d$, define the augmented system:
\[
	\begin{aligned}
		U[P_1,P_2]&:=\bigl(Q,T\cup\{(\Vec p,\Vec a,\Vec p):\Vec a\in P_1\}
		\cup\{(\Vec q,-\Vec a,\Vec q):\Vec a\in P_2\}\bigr),
	\end{aligned}
\]
and abbreviate $U[P,\varnothing]$ to $U[P]$.
We first extend $C$-coverability to hybrid sets using only their free periods, whose multiplicities are unrestricted.

\begin{definition}
	\label{def:coverable_hybrid}
	Let $H=\Vec b+P^{\langle\Vec c,\cdot\rangle\le K}$ be a hybrid set at state $\Vec p$ of a VASS $U$, and let $P_0$ be its free periods.
	We say that $H$ is \emph{forward $C$-coverable in $U$} if $\Vec p(\Vec b)$ is forward $C$-coverable in $U[P_0]$.
	Similarly, $H$ at target state $\Vec q$ is \emph{backward $C$-coverable in $U$} if $\Vec q(\Vec b)$ is backward $C$-coverable in $U[\varnothing,P_0]$.
\end{definition}

Consider a forward $C$-coverable hybrid set $H_1$ and a backward $C$-coverable hybrid set $H_2$ with bases $\Vec b_1,\Vec b_2$ and period sets $P_1,P_2$, respectively.
If $\Vec p(\Vec x)\xrightarrow{*}\Vec q(\Vec y)$ in $U$ for some $\Vec x\in H_1$ and $\Vec y\in H_2$, then $\Vec p(\Vec b_1)\xrightarrow{*}\Vec q(\Vec b_2)\text{}$ in $U[P_1,P_2]$.
Indeed, the added self-loops first generate $\Vec x$ from $\Vec b_1$ and then reduce $\Vec y$ to $\Vec b_2$.
Moreover, the new instance is $C$-coverable.
Thus, \cref{cor:coverable_case} applies.
The added periods may violate the cost bounds of $H_1$ or $H_2$.
However, this can happen only when $K_i<\norm{\Vec c_i}L$, where $L$ is the polynomial length bound for the augmented run.
In this case, we absorb the restricted periods of a fixed connected member into the base.
Since each restricted period has positive integer cost, their total multiplicity is at most $K$.
This gives the following lemma.

\begin{lemma}
	\label{lem:restricted_period_absorption}
	Let $H=\Vec b+P^{\langle\Vec c,\cdot\rangle\le K}$ be a hybrid set, and write $P=P_0\sqcup P_+$ for its free and restricted periods.
	Fix a representation $\Vec z=\Vec b+\sum_{\Vec a\in P}\eta_{\Vec a}\cdot\Vec a$, where $\eta_{\Vec a}\in\NN\text{ for all }\Vec a\in P$ and $\sum_{\Vec a\in P}\Vec c(\Vec a)\eta_{\Vec a}\le K$.
	Define
	\begin{equation*}
		\widehat{\Vec b}:=\Vec b+\sum_{\Vec a\in P_+}\eta_{\Vec a}\cdot\Vec a,
		\qquad \widehat H:=\widehat{\Vec b}+P_0^*.
	\end{equation*}
	Then $\Vec z\in\widehat H\subseteq H$ and $\norm{\widehat{\Vec b}}\le\norm{\Vec b}+K\norm P$.
\end{lemma}

\begin{proof}
	Every restricted cost is a positive integer, so
	$\sum_{\Vec a\in P_+}\eta_{\Vec a}\le\sum_{\Vec a\in P}\Vec c(\Vec a)\eta_{\Vec a}\le K$, which gives the norm bound.
	Fixing the restricted coefficients gives $\Vec z\in\widehat H\subseteq H$.
\end{proof}

Moreover, this operation preserves $C$-coverability.
Since $\widehat{\Vec b}\ge\Vec b$ and the free periods are unchanged, the covering runs can be shifted by $\widehat{\Vec b}-\Vec b$ using monotonicity.
The same argument applies to backward coverability.
We use these properties to handle the cost constraints when applying the reachability bounds for fixed configurations to augmented systems.
The resulting lemma is stated below, and its proof is deferred to~\cref{app:joining_coverable_hybrids}.

\begin{restatable}{lemma}{joinCoverableHybrid}\label{lem:joining_coverable_hybrids}
	Let $V$ be a $5$-VAS, $B\in\NN$, and $C$ a transformation-area threshold for $\enc V_B$.
	Let $H_1,H_2$ be hybrid sets at states $\Vec p,\Vec q$ of a sequential sub-VASS $U$ of $\enc{V}_B$.
	Suppose that there exist $\Vec x\in H_1$ and $\Vec y\in H_2$ such that $\Vec p(\Vec x)\xrightarrow{*}\Vec q(\Vec y)$ in $U$.
	Assume that $H_1$ is forward $C$-coverable in $U$ and $H_2$ is backward $C$-coverable in $U$.
	Then there exist $\Vec x'\in H_1$ and $\Vec y'\in H_2$ such that
	\begin{equation*}
		\norm{\Vec x'},\ \norm{\Vec y'},
		\MinLen_V\bigl((\Vec p,\Vec x'),(\Vec q,\Vec y')\bigr)
		\le
		\bigl(B+C+\Size(V)+\hnorm(H_1)+\hnorm(H_2)\bigr)^{\bigO(1)}.
	\end{equation*}
\end{restatable}

%% file: contents/7-Semicoverable.tex
\section{Additive-Growth Induction for Semi-Coverable Instances}
\label{sec:semi_coverable}

We have established polynomial reachability-distance bounds for $C$-coverable instances in~\cref{sec:basic_reachability_bounds}.
We now show that coverability at just one endpoint suffices.
We call an instance \emph{semi-$C$-coverable} if its source is forward $C$-coverable or its target is backward $C$-coverable.

\begin{restatable}[Semi-Coverable Case]{theorem}{semicoverabletheorem}
	\label{thm:semi_coverable}

	Let $V$ be a $5$-VAS, $B \in \NN$, and $(U, \Vec{p}(\Vec{x}), \Vec{q}(\Vec{y}))$ be a sequential sub-VASS of $\enc{V}_B$ with $\Vec{p}(\Vec{x}) \xrightarrow{*} \Vec{q}(\Vec{y})$ in $U$.
	If $(U, \Vec{p}(\Vec{x}), \Vec{q}(\Vec{y}))$ is semi-$C$-coverable, then
	\begin{equation}
		\MinLen_V(\Vec s,\Vec t)
		\le \bigl(B + C + \Size(V,\Vec s,\Vec t)\bigr)^{\bigO(1)},
		\label{eq:semi_coverable_length}
	\end{equation}
	where $C$ is a transformation-area threshold for $\enc{V}_B$, $\Vec{s}=(\Vec{p},\Vec x)$ and $\Vec{t}=(\Vec{q},\Vec y)$.
\end{restatable}

We prove the theorem when the target is backward $C$-coverable, without assuming coverability of the source.
The case of a forward $C$-coverable source follows by reversal.
The proof proceeds by induction on the number $k$ of SCCs.
We aim to ensure that each step increases the source norm by an additive polynomial bound, so that the final polynomial degree does not depend on $k$.
By \cref{cor:threshold_closure}, the same transformation-area threshold $C$ remains valid under reversal and when adjoining the nonnegative or nonpositive period transitions during the induction.

\paragraph{A Maximal Bounded-Coordinate Prefix.}
\label{par:maximal_bounded_prefix}

If the source is not forward $C$-coverable, Rackoff's extraction implies that some coordinate remains below a polynomial threshold along the first SCC segment of a witnessing run.
We follow this coordinate through successive SCCs and stop at the entry of the first later SCC in which it reaches the threshold, or at the target if no such SCC exists.
In the former case, any configuration at this entry state becomes forward $C$-coverable within that SCC if the other two coordinates are sufficiently large and the tracked coordinate is not decreased.

\begin{restatable}{lemma}{boundedPrefix}\label{lem:bounded_prefix}
	Let $U=(U_1)u_1(U_2)\cdots u_{k-1}(U_k)$ be a sequential $3$-VASS, and fix a run
	$p(\Vec x)\xrightarrow{\pi}q(\Vec y)$ whose source is not forward $C$-coverable.
	Then there exists an integer $D\le(C+\Size(U))^{\bigO(1)}$ such that one can find a coordinate $i\in[3]$ and a prefix $p(\Vec x)\xrightarrow{\pi_0}r(\Vec z)$ with $i\in\BDCounters(\pi_0,D)$ satisfying one of the following:
	\begin{enumerate}
		\item $r(\Vec z)=q(\Vec y)$;
		\item\label{item:bounded_prefix_entry} $r(\Vec z)$ is the entry configuration of a later component $U_h$, and every
		$\Vec w\in\NN^3$ satisfying
		\begin{equation*}
			\Vec w(i)\ge\Vec z(i),
			\qquad
			\Vec w(j)\ge D
			\quad\text{for all }j\ne i
		\end{equation*}
		has $r(\Vec w)$ forward $C$-coverable in $U_h$.
	\end{enumerate}
	Moreover, $\pi_0$ is captured by a $2$-VASS of size at most $D\Size(U)$.
\end{restatable}

The proof uses Rackoff's bounds to establish~\cref{item:bounded_prefix_entry} at the first component where the selected coordinate ceases to remain bounded.
We defer the details to~\cref{app:bounded_prefix}.

\paragraph{Additive-Growth Induction and Prefix Classification.}
\label{par:prefix_classification}
Throughout this section, fix a $5$-VAS $V$, an encoding bound $B\in\NN$, and a transformation-area threshold $C$ for $W:=\enc{V}_B$.
Let $U=(U_1)u_1(U_2)\cdots u_{k-1}(U_k)$ be a reachable sequential sub-VASS of $W$ whose source $s$ and target $t$ belong to $U_1$ and $U_k$, respectively.
Their corresponding configurations in $V$ are $\Vec s=(\Vec p,\Vec x)$ and $\Vec t=(\Vec q,\Vec y)$.
For the proof of~\cref{thm:semi_coverable}, assume that $t$ is backward $C$-coverable in $U$.
Write $U=(Q,T)$.
Since $\Size(W)\le(B+1)^2\Size(V)$, the threshold in~\cref{lem:bounded_prefix} can be chosen uniformly over all sequential subinstances of $W$.
By enlarging it if necessary, we fix 
\begin{equation}\label{eq:uniform_structural_threshold}
	B+C+\Size(V)\le D \le(B+C+\Size(V))^{\bigO(1)}.
\end{equation}

If $s$ is forward $C$-coverable, \cref{cor:coverable_case} gives the required polynomial bound.
Otherwise, \cref{lem:bounded_prefix} gives a partition
\[
    \Vec p(\Vec x)
    \xrightarrow{\pi_0}\Vec r(\Vec z)
    \xrightarrow{\pi_1}\Vec q(\Vec y),
\]
where some coordinate $i$ remains below $D$ throughout $\pi_0$.
The prefix is captured by a $2$-VASS of size at most $D\Size(U)$, so \cref{lem:hybrid_source} gives a polynomial-length replacement in terms of the system size and endpoint norms.
If $\pi_1$ is empty, lifting this replacement to $V$ proves the desired bound.
Otherwise, the suffix lies in a sequential subinstance with fewer SCCs.

Fix a structural polynomial threshold $R\ge D$, to be chosen in~\cref{subsec:one_sided_completion}.
We call $\pi_0$ \emph{$R$-small} if $\norm{\Vec z}\le\norm{\Vec x}+R$.
In this case, induction applies to the unchanged suffix, with the source norm increased by at most $R$.
Otherwise, we call $\pi_0$ \emph{$R$-large}.
Since coordinate $i$ remains below $D\le R$, this means that
\begin{equation*}
    \Vec z(j)>\norm{\Vec x}+R
    \qquad\text{for some }j\ne i.
\end{equation*}
For a suitable choice of $R$, every $R$-large prefix yields either a hybrid set with a nonzero free period or a translated family.
The tools for these two outcomes are developed in~\cref{subsec:pumpable_sources,subsec:joining_translated}, respectively.
In~\cref{subsec:one_sided_completion}, we establish this dichotomy and use it to complete the induction.

\subsection{Pumpable Source}
\label{subsec:pumpable_sources}

A configuration $p(\Vec x)$ is \emph{pumpable} in $U$ if an enabled cycle returns to $p(\Vec x+\Vec a)$ for some $\Vec a\ge\Vec0$ with $\Vec a\ne\Vec0$.
Such a cycle can be repeated indefinitely, but it may leave some coordinates unchanged.
Thus, pumpability does not necessarily imply forward $C$-coverability.
We show that it nevertheless suffices for a polynomial reachability-distance bound when the target is backward $C$-coverable.

\begin{restatable}{lemma}{pumpableinterfacelemma}
	\label{lem:pumpable_coverable_target}
	Let $U$ be a sequential sub-VASS of $\enc V_B$ and let $\Vec p(\Vec x)\xrightarrow{*}\Vec q(\Vec y)$ in $U$, with $\Vec s=(\Vec p,\Vec x)$ and $\Vec t=(\Vec q,\Vec y)$.
	If $\Vec p(\Vec x)$ is pumpable in $U$ and $\Vec q(\Vec y)$ is backward $C$-coverable, then
	\begin{equation}\label{eq:pumpable_coverable_distance}
		\MinLen_V(\Vec s,\Vec t)\le\bigl(B+C+\Size(V,\Vec s,\Vec t)\bigr)^{\bigO(1)}.
	\end{equation}
\end{restatable}

The proof starts with a short pumping cycle and considers the coordinates left unchanged by its effect.
Using bounded-coordinate extraction, we either find a prefix with one of these coordinates bounded throughout or obtain two ordered cycles whose effects, together with the initial pumping cycle, yield a strictly positive free period.
In particular, the positive free period ensures forward $C$-coverability, allowing us to apply~\cref{lem:joining_coverable_hybrids} to the backward $C$-coverable target.

\paragraph{$D$-Alternating Prefixes.}
\label{par:d_alternating_prefixes}

We first strengthen~\cref{lem:bounded_prefix} using the effect $\Vec a_0$ of the initial pumping cycle.
By temporarily adding $\Vec a_0$ as a self-loop at each component entry, we ensure that every coordinate in $\Supp(\Vec a_0)$ becomes large.
Thus, whenever forward $C$-coverability fails, the bounded coordinate extracted by~\cref{lem:bounded_prefix} must lie outside $\Supp(\Vec a_0)$.
Since this bounded coordinate may change between components, we introduce the following alternating condition.

\begin{definition}
\label{def:d_alternating}
	For $D\in\NN$ and $I\subseteq[3]$, we say that a run $\pi$ is \emph{$D$-alternating on $I$} if it admits a partition $\pi=\pi_1\cdots\pi_m$, where $m\ge0$, such that for $S_j:=\BDCounters(\pi_j,D)\cap I$,
	\begin{equation}\label{eq:d_alternating}
		S_j\ne\varnothing\text{ for all }j\in[m],\qquad S_j\cap S_{j+1}=\varnothing\text{ for all }j\in[m-1].
	\end{equation}
\end{definition}

When $|I|=1$, a nonempty $D$-alternating run keeps that coordinate bounded throughout.
When $|I|=2$, the bounded coordinates alternate between consecutive segments.

\begin{restatable}{lemma}{alternatingBoundedness}
	\label{lem:alternating_boundedness}
	Let $U=(U_1)u_1\cdots u_{k-1}(U_k)$ be a sequential $3$-VASS, let $p(\Vec x)\xrightarrow{\pi}t$ be a run in $U$, and let $\Vec a_0\in\NN^3\setminus\{\Vec0\}$ be the effect of an enabled cycle at $p(\Vec x)$.
	There exist
	\begin{equation}\label{eq:alternating_threshold_bound}
		D\le(C+\Size(U)+\norm{\Vec a_0})^{\bigO(1)}
	\end{equation}
	and a possibly empty prefix $p(\Vec x)\xrightarrow{\pi_*}r_*(\Vec z_*)$ that is $D$-alternating on $[3]\setminus\Supp(\Vec a_0)$.
	Either $r_*(\Vec z_*)=t$, or it is the entry of some component $U_*$ such that $r_*(\Vec z_*)$ is forward $C$-coverable in $U_*[\{\Vec a_0\}]$.
\end{restatable}

The proof repeatedly applies~\cref{lem:bounded_prefix} after adding $\Vec{a}_0$ as a self-loop at each SCC entry until the target is reached or a forward $C$-coverable entry is found.
Merging consecutive segments with a common bounded coordinate gives the required alternation.
See~\cref{app:alternating_boundedness} for its proof.

\subsubsection{Extracting Two Ordered Cycles for a Positive Free Period}
\label{subsubsec:ordered_pumping_cycles}
By~\cref{lem:alternating_boundedness}, the coordinates left unchanged by the initial pumping cycle are $D$-alternating along a prefix.
After permuting coordinates so that the pumping cycle increases coordinate $3$, this prefix is also $D$-alternating on $\{1,2\}$: if only one coordinate is left unchanged, the partition has at most one segment.
The following lemma shows that either some coordinate remains bounded throughout the prefix under a larger polynomial threshold, or there are two ordered cycles, each increasing one coordinate while leaving the other unchanged.
In the latter case, repeating the initial pumping cycle sufficiently many times makes the combined effect positive on coordinate $3$ as well, yielding a strictly positive free period.

\begin{lemma}
	\label{lem:hill_cuts}
	Let $U=(Q,T)$ be a $3$-VASS and let $p(\Vec x)\xrightarrow{\pi}t$ be a run $D$-alternating on $\{1,2\}$.
	There exist
	\begin{equation*}
		D_1,L\le \bigl(D+\Size(U)+\norm{\Vec x}\bigr)^{\bigO(1)}
	\end{equation*}
	such that either $\BDCounters(\pi,D_1) \neq \varnothing$, or, after exchanging the first two counters, there is a partition
	\begin{equation}
		p(\Vec x)\xrightarrow{\rho_1}q_1(\Vec z_1)
		\xrightarrow{\rho_2}q_2(\Vec z_2)
		\xrightarrow{\rho}t
		\label{eq:ordered_hill_cuts}
	\end{equation}
	and enabled cycles $\theta_1,\theta_2$ at $q_1(\Vec z_1)$ and $q_2(\Vec z_2)$, respectively, with $|\theta_\iota|\le L$, $\Eff(\theta_{\iota})(\iota)>0$ and $\Eff(\theta_{\iota})(3-\iota)=0$ for $\iota\in\{1,2\}$.
	Moreover, $\{1,2\}\subseteq\BDCounters(\rho_1,D_1)$ and $2\in\BDCounters(\rho_1\rho_2,D_1)$.
\end{lemma}

We first isolate a hill-cutting argument.
It extracts a short enabled cycle when one coordinate remains bounded while another increases sufficiently.
The large increase produces many positions where the growing coordinate reaches a new maximum.
There are at most $D|Q|$ pairs consisting of a state and the bounded-coordinate value.
Two records must therefore share such a pair while the growing coordinate increases, yielding a cycle that can be shortened using Rackoff's covering bound.
The proof is deferred to~\cref{app:hill_cutting}.

\begin{restatable}{lemma}{hillCutting}
    \label{lem:hill_cutting}
    Let $U=(Q,T)$ be a $3$-VASS, and let $i,j\in[3]$ be distinct coordinates.
    Suppose $p(\Vec x)\xrightarrow{\pi}q(\Vec y)$ is a run with $D\ge1$ satisfying $j\in\BDCounters(\pi,D)$ and 
	\[\Vec y(i)\ge D':=\Vec x(i)+D|Q|\norm T+1.\]
    Then there exist a configuration $r(\Vec z)$ along $\pi$ and a cycle $\theta$ fireable at $r(\Vec z)$ such that
    \[
        \Eff(\theta)(i)>0,\qquad
        \Eff(\theta)(j)=0,\qquad
        |\theta|\le(D+\Size(U)+\Vec x(i))^{\bigO(1)}.
    \]
\end{restatable}

\begin{proof}[Proof of~\cref{lem:hill_cuts}]
	Let $D_1:=\max\{\norm{\Vec x},D\}+D|Q|\norm T+1$.
	We call the first position where coordinate $i$ has value at least $D_1$ its \emph{$D_1$-crossing}.
	If either coordinate has no $D_1$-crossing, the first conclusion holds with $L:=1$.
	Otherwise, fix a partition witnessing $D$-alternation.

	For $i\in\{1,2\}$, let $r_i(\Vec u_i)$ be the source of the segment containing the incoming transition of its $D_1$-crossing.
	This segment bounds coordinate $3-i$ below $D$.
	If it is not the first segment, its predecessor bounds coordinate $i$ below $D$ by alternation.
	Otherwise, its source is the initial configuration.
	Hence $\Vec u_i(i)\le\max\{\norm{\Vec x},D\}$, so~\cref{lem:hill_cutting} applies to the segment prefix ending at the crossing.
	Its construction gives a configuration $q_i(\Vec z_i)$ before the $D_1$-crossing and an enabled cycle $\theta_i$ there with
	\[
		\Eff(\theta_i)(i)>0,
		\qquad \Eff(\theta_i)(3-i)=0.
	\]
	The bound on $\Vec u_i(i)$ gives a common length bound $L\le(D+\Size(U)+\norm{\Vec x})^{\bigO(1)}$.

	The two crossings lie in distinct segments, since each segment keeps one coordinate below $D$.
	Exchange coordinates so that the $D_1$-crossing of coordinate $1$ comes first.
	Then $q_1(\Vec z_1)$ precedes $q_2(\Vec z_2)$, giving~\eqref{eq:ordered_hill_cuts}.
	The first cut precedes both $D_1$-crossings, and the second precedes that of coordinate $2$.
	Consequently, $\{1,2\}\subseteq\BDCounters(\rho_1,D_1)$ and $2\in\BDCounters(\rho_1\rho_2,D_1)$, as required.
\end{proof}

\subsubsection{Drop Tests and Hybrid Set Representation}
\label{subsubsec:drop_tests_hybrid_representation}

In the two-cycle case of~\cref{lem:hill_cuts}, coordinate $2$ remains bounded along $\rho_1\rho_2$.
By recording this coordinate in states, we can describe the reachability set at $q_2$ by a finite union of hybrid sets.
However, this alone does not allow us to obtain a strictly positive free period: a run reaching an arbitrary member of such a hybrid set need not visit $q_1$ with $\theta_1$ enabled or have $\theta_2$ enabled.
We therefore construct an auxiliary $3$-VASS $U^{(q_1,q_2)}$ that captures exactly the runs visiting $q_1$ before reaching $q_2$, with both cycles enabled at the corresponding configurations.
At each state $q_i$, subtracting and then restoring $\Drop(\theta_i)$ checks that $\theta_i$ is enabled without changing the counter values.

\begin{lemma}
    \label{lem:two_cycle_capture}
    Let $U=(Q,T)$ be a $3$-VASS, $p(\Vec x)\xrightarrow{\rho_1}q_1(\Vec z_1)\xrightarrow{\rho_2}q_2(\Vec z_2)$ be a run in $U$, and $\theta_0,\theta_1,\theta_2$ be three cycles enabled at $p(\Vec x)$, $q_1(\Vec z_1)$, and $q_2(\Vec z_2)$, respectively, with lengths at most $L$ and effects
    \[
        \Eff(\theta_0)=(0,0,a),\qquad
        \Eff(\theta_1)=(u,0,b),\qquad
        \Eff(\theta_2)=(0,v,c),
    \]
    where $a,u,v\ge1$.
    There exists a $3$-VASS $U^{(q_1,q_2)}$ of size polynomial in $\Size(U)+L$ with initial state $p^{(0)}$ and final state $f$, satisfying the following properties.
    \begin{enumerate}
        \item $p^{(0)}(\Vec x)\xrightarrow{\pi}f(\Vec{z}_2)$ in $U^{(q_1,q_2)}$ for some run $\pi$ with $\BDCounters(\rho_1\rho_2, D) \subseteq \BDCounters(\pi, D)$ for every $D\ge0$.
        \item There exists $\Vec a_+\ge\Vec1$ with $\norm{\Vec a_+}\le(L+\Size(U))^{\bigO(1)}$ such that
        \begin{equation*}
            \Reach^{f}_{U^{(q_1,q_2)}}(p^{(0)}(\Vec x))
            +\Vec a_+^*
            \subseteq
            \Reach^{q_2}_U(p(\Vec x)).
        \end{equation*}
		Moreover, if $f(\Vec h)$ is reachable in $U^{(q_1,q_2)}$ from $p^{(0)}(\Vec x)$, then for every $k\in\NN$,
		\[
            \MinLen_U\bigl(p(\Vec x),q_2(\Vec h+k\cdot\Vec a_+)\bigr)\le\MinLen_{U^{(q_1,q_2)}}\bigl(p^{(0)}(\Vec x),f(\Vec h)\bigr)+k(L+\Size(U))^{\bigO(1)}.
		\]
    \end{enumerate}
\end{lemma}

\begin{proof}
	Let $\Vec a_0:=\Eff(\theta_0)$, $\Vec e_j:=\Eff(\theta_j)$, and $\Vec d_j:=\Drop(\theta_j)$ for $j=1,2$.
	Then $\norm{\Vec a_0},\norm{\Vec e_j},\norm{\Vec d_j}\le L\norm T$ and $\Vec d_j\le\Vec z_j$.
	Take two disjoint copies $U^{(0)},U^{(1)}$ of $U$, with copied states $q^{(0)},q^{(1)}$, and three fresh states $\hat q_1,\hat q_2,f$.
	Besides the copied transitions, add only the following transitions $t_j^\pm$, of effects $\pm\Vec d_j$, respectively, for $j=1,2$:
	\begin{equation*}
		q_1^{(0)}\xrightarrow{t_1^-}\hat q_1\xrightarrow{t_1^+}q_1^{(1)},
		\qquad q_2^{(1)}\xrightarrow{t_2^-}\hat q_2\xrightarrow{t_2^+}f.
	\end{equation*}
	Denote the resulting system by $U^{(q_1,q_2)}$.
	It has $2|Q|+3$ states, $2|T|+4$ transitions, and transition norm at most $L\norm T$, so its size is polynomial in $\Size(U)+L$.
	Each test decreases and then restores the vector, and is enabled precisely when its source vector is at least $\Vec d_j$.
	Thus $\rho_1\rho_2$ lifts to a run $\pi$ reaching $f(\Vec z_2)$ with the same coordinatewise maxima, proving the first property.

	Every run $p^{(0)}(\Vec x)\xrightarrow{\gamma}f(\Vec h)$ traverses the two tests in order.
	Deleting their four transitions gives a run in $U$ of the form
	\begin{equation*}
		p(\Vec x)\xrightarrow{\eta_1}q_1(\Vec w_1)
		\xrightarrow{\eta_2}q_2(\Vec h),
		\qquad \Vec w_1\ge\Vec d_1,
		\quad \Vec h\ge\Vec d_2.
	\end{equation*}
	Let
	\[
		\lambda:=2L\norm T+1,
		\qquad \Vec a_1:=\lambda\cdot\Vec a_0+\Vec e_1,
		\qquad \Vec a_+:=\Vec a_1+\Vec e_2.
	\]
	The prescribed effects and $|b|,|c|\le L\norm T$ give $\Vec a_1\ge\Vec0$, $\Vec a_+\ge\Vec1$, and
	\[
		\norm{\Vec a_+}\le(2L\norm T+3)L\norm T
		\le(L+\Size(U))^{\bigO(1)}.
	\]
	For $k\in\NN$, consider $\theta_0^{k\lambda}\eta_1\theta_1^k\eta_2\theta_2^k$.
	The copies of $\theta_0$ and the translated $\eta_1,\eta_2$ are enabled by monotonicity.
	Before the $(j+1)$-st copies of $\theta_1$ and $\theta_2$, respectively, where $0\le j<k$, we have
	\[
		\begin{aligned}
			\Vec w_1+k\lambda\cdot\Vec a_0+j\cdot\Vec e_1
			&=\Vec w_1+(k-j)\lambda\cdot\Vec a_0+j\cdot\Vec a_1\ge\Vec d_1,\\
			\Vec h+k\cdot\Vec a_1+j\cdot\Vec e_2
			&=\Vec h+(k-j)\cdot\Vec a_1+j\cdot\Vec a_+\ge\Vec d_2.
		\end{aligned}
	\]
	Hence every cycle copy is enabled, and the run ends at $q_2(\Vec h+k\cdot\Vec a_+)$.
	Its length is at most
	\[
		|\gamma|+k\bigl(\lambda|\theta_0|+|\theta_1|+|\theta_2|\bigr)
		\le|\gamma|+k(2L\norm T+3)L.
	\]
	Taking $\gamma$ shortest proves the distance bound and the reachability inclusion.
\end{proof}

The two properties in~\cref{lem:two_cycle_capture} serve different purposes.
The first preserves the boundedness of the original run, allowing us to record a bounded coordinate in states without losing $\Vec z_2$ from the reachability set.
The second property applies to every vector reachable at $f$ in the auxiliary VASS: adding any multiple of $\Vec a_+$ yields a vector reachable at $q_2$ in $U$, with a polynomial bound on the additional run length.
Thus, after obtaining a hybrid set $G$ containing $\Vec z_2$ from the auxiliary reachability set, we can add $\Vec a_+$ as a free period.

Together with~\cref{thm:hybrid_description,lem:hybrid_source}, this yields a hybrid set with polynomially bounded representation norm and short prefix witnesses.
The complete proof is given in~\cref{app:pump_extraction}.

\begin{restatable}{lemma}{pumpExtraction}
	\label{lem:pump_extraction}
	Let $p(\Vec x)\xrightarrow{\pi}t=q(\Vec y)$ be a run in a sequential $3$-VASS $U=(Q,T)$, with $p(\Vec{x})$ pumpable, and let $C\in\NN$.
	There exist a partition of $\pi$
	\begin{equation*}
		p(\Vec x)\xrightarrow{\tau}r(\Vec z)\xrightarrow{\sigma}t,
	\end{equation*}
	and a hybrid set $H$ containing $\Vec z$, with $\hnorm(H) \le (C+\Size(U,p(\Vec x)))^{\bigO(1)}$, such that every $r(\Vec h)$ with $\Vec h\in H$ has a prefix witness $p(\Vec{x}) \xrightarrow{\rho} r(\Vec h)$ of length
	\begin{equation*}
		|\rho| \le (C + \Size(U,p(\Vec{x}),r(\Vec{h})))^{\bigO(1)}.
	\end{equation*}
	Either $\sigma$ is empty, or $H$ is forward $C$-coverable in the component containing $r$.
\end{restatable}

\subsubsection{Joining the Backward Coverable Target}
\label{subsubsec:joining_coverable_target}

Now we can show that the reachability distance from a pumpable source to a backward $C$-coverable target in $\enc V_B$ is polynomially bounded in $V$.

\begin{proof}[Proof of~\cref{lem:pumpable_coverable_target}]
	Let $M:=C+B+\Size(V,\Vec s,\Vec t)$.
	Since $\Size(U)\le(B+1)^2\Size(V)$, applying~\cref{lem:pump_extraction} to a witnessing run $\Vec p(\Vec x)\xrightarrow{\pi}\Vec q(\Vec y)$ gives a partition $\pi=\tau\sigma$ at $\Vec r(\Vec z)$ and a hybrid set $H$ containing $\Vec z$, with
	\begin{equation}\label{eq:pump_joining_prefix_data}
		\hnorm(H)\le M^{\bigO(1)},
		\qquad
		\MinLen_U\bigl(\Vec p(\Vec x),\Vec r(\Vec h)\bigr)
		\le(M+\norm{\Vec h})^{\bigO(1)}\quad\text{for all }\Vec h\in H.
	\end{equation}
	If $\sigma=\varepsilon$, take $\Vec h:=\Vec y\in H$ in~\eqref{eq:pump_joining_prefix_data} and lift the resulting run to $V$, obtaining the required bound.

	Otherwise, let $U'$ be the suffix of $U$ starting with the component containing $\Vec r$, including the original bridges.
	By~\cref{lem:pump_extraction}, $H$ is forward $C$-coverable in this component and hence in $U'$.
	The backward covering run at $\Vec q(\Vec y)$ stays in its component, so $\{\Vec y\}$ is backward $C$-coverable in $U'$.
	Since $\Vec z\in H$ and $\Vec r(\Vec z)\xrightarrow{\sigma}\Vec q(\Vec y)$ in $U'$, \cref{lem:joining_coverable_hybrids} applied to $H$ and $\{\Vec y\}$ gives $\Vec h\in H$ with
	\begin{equation*}
		\norm{\Vec h},\ \MinLen_V\bigl((\Vec r,\Vec h),\Vec t\bigr)
		\le M^{\bigO(1)}.
	\end{equation*}
	By~\eqref{eq:pump_joining_prefix_data}, we have
	\begin{equation*}
		\MinLen_V(\Vec s,\Vec t)\le \MinLen_U\bigl(\Vec p(\Vec x),\Vec r(\Vec h)\bigr)+\MinLen_V\bigl((\Vec r,\Vec h),\Vec t\bigr) \le M^{\bigO(1)},
	\end{equation*}
	which is the bound in~\eqref{eq:pumpable_coverable_distance}.
\end{proof}

\paragraph{Hybrid Sources with a Nonzero Free Period.}
\label{par:free_period_hybrid_sources}
The preceding result assumes a fixed pumpable source, whereas the $R$-large case may yield a hybrid set representation.
If this set has a nonzero free period, we can adjoin its periods as self-loops to make the base pumpable and then apply~\cref{lem:pumpable_coverable_target}.
By removing the added transitions and handling the restricted-period costs as in~\cref{lem:joining_coverable_hybrids}, we obtain the following result.
The proof can be found in~\cref{app:free_covered}.

\begin{restatable}{lemma}{freeCovered}
	\label{lem:free_covered}
	Let $U$ be a sequential sub-VASS of $\enc V_B$, let $C$ be a transformation-area threshold for $\enc V_B$, and let $H$ be a hybrid set at state $\Vec r$ with a nonzero free period.
	Suppose that some member of $H$ reaches a backward $C$-coverable target $\Vec q(\Vec y)$ in $U$.
	Then there exists $\Vec h\in H$ such that
	\begin{equation*}
		\norm{\Vec h},
		\MinLen_V\bigl((\Vec r,\Vec h),\Vec{t}\bigr)
		\le
		\bigl(
			B+C+\Size(V,\Vec t)+\hnorm(H)
		\bigr)^{\bigO(1)},
	\end{equation*}
	where $\Vec t=(\Vec q,\Vec y)$.
	The degree is independent of the number of components.
\end{restatable}

\subsection{Translated Families}
\label{subsec:joining_translated}

The preceding results handle hybrid sources with a nonzero free period.
We now consider another terminal case, based on hybrid sets whose bases vary along a common direction while their periods and cost constraints remain unchanged.

\begin{definition}[Translated Family]
	\label{def:translated_family}
	For a permutation $\sigma$ of $[3]$, a \emph{translated family} $\mathcal H=\{H_r\}_{r=0}^\kappa$ consists of hybrid sets
	$H_r=\Vec b+r\cdot\Vec\delta+P^{\langle\Vec c,\cdot\rangle\le K}$, where $\Vec\delta\in\ZZ^3$, $\Vec\delta(\sigma(3))=0$, $\Vec p(\sigma(3))=0$ for every $\Vec p\in P$, and $\Vec b+r\cdot\Vec\delta\ge\Vec0$ for each integer $r\in[0,\kappa]$.
\end{definition}

We next consider reachability between two translated families.
The key is to combine a known run between their unshifted members with $C$-coverability at suitable translation indices.
Recall that $C$ is the fixed transformation-area threshold for $\enc{V}_B$, and $D$ is the polynomial threshold fixed in~\eqref{eq:uniform_structural_threshold}.
Let $\mathcal H=\{H_r\}_{r = 0}^\kappa$ and $\mathcal J=\{J_r\}_{r = 0}^\kappa$ be two translated families, where
\begin{equation*}
	H_r=\Vec b+r\cdot\Vec\delta+P^{\langle\Vec c,\cdot\rangle\le K},\qquad
	J_r=\Vec b'+r\cdot\Vec\delta'+P'^{\langle\Vec c',\cdot\rangle\le K'}.
\end{equation*}
Let $\Vec p$ and $\Vec q$ be two states of a sequential $3$-VASS $U$.
We call $(U,\Vec p(\mathcal H),\Vec q(\mathcal J))$ \emph{$C$-joinable} if the following conditions hold:
\begin{enumerate}
	\item There exist $\Vec x\in H_0$ and $\Vec x'\in J_0$ such that $\Vec p(\Vec x)\xrightarrow{*}\Vec q(\Vec x')$ in $U$.
	\item There exist integers $0\le r,r'\le\min\{D,\kappa\}$ such that $H_r$ is forward $C$-coverable and $J_{r'}$ is backward $C$-coverable in $U$.
\end{enumerate}

The following lemma shows that $C$-joinability allows us to select members of the two translated families connected by a short run in the original VAS.

\begin{lemma}
	\label{lem:translated_join}
	Let $\mathcal H=\{H_r\}_{r = 0}^\kappa$ and $\mathcal J=\{J_r\}_{r = 0}^\kappa$ be two translated families at states $\Vec p$ and $\Vec q$ of $U$, where
	\begin{equation*}
		H_r=\Vec b+r\cdot\Vec\delta+P^{\langle\Vec c,\cdot\rangle\le K},
		\qquad
		J_\ell=\Vec b'+\ell\cdot\Vec\delta'+P'^{\langle\Vec c',\cdot\rangle\le K'}.
	\end{equation*}
	Suppose that $(U,\Vec p(\mathcal H),\Vec q(\mathcal J))$ is $C$-joinable.
	There exists $\kappa_0\le\bigl(D+\norm P+\norm{P'}\bigr)^{\bigO(1)}$ such that, if $\kappa\ge\kappa_0$, some $0\le r,\ell\le\kappa_0$ and $\Vec h\in H_r$, $\Vec h'\in J_\ell$ satisfy
	\begin{equation*}
		\norm{\Vec h},\quad
		\norm{\Vec h'},\quad
		\MinLen_V\bigl((\Vec p,\Vec h),(\Vec q,\Vec h')\bigr)
		\le
		\bigl(
			D+\hnorm(H_0)+\hnorm(J_0)
			+\norm{\Vec\delta}+\norm{\Vec\delta'}
		\bigr)^{\bigO(1)}.
	\end{equation*}
\end{lemma}

\begin{proof}
	Restrict $U$ to the components from $\Vec p$ to $\Vec q$, including their full transitions and intervening bridges.
	Let $\Vec z\in H_0$, $\Vec z'\in J_0$ and $a,a'\le\min\{D,\kappa\}$ witness $C$-joinability.
	Write $\Vec b_s:=\Vec b+s\cdot\Vec\delta$ and $\Vec b'_s:=\Vec b'+s\cdot\Vec\delta'$ for $0\le s\le\kappa$, and let $P_0,P'_0$ be the free periods.
	The covering runs stay in the components containing their respective states, so they are available in this restriction.
	Let $M:=D+\hnorm(H_0)+\hnorm(J_0)+\norm{\Vec\delta}+\norm{\Vec\delta'}$ and $n:=|Q|$.
	Set $\Delta_*:=\max\{1,\norm T,\norm P,\norm{P'}\}$ and $L_{\mathrm{cov}}:=\max\{2,n,\Delta_*,C\}^{c_{\mathrm{cov}}}$, where $c_{\mathrm{cov}}$ is the exponent for $3$-VASS from~\cref{lem:rackoff_cover}.
	Choose
	\begin{equation}\label{eq:bilateral_joining_range}
		\kappa_0:=D+6(n\Delta_*)^3+C+L_{\mathrm{cov}}\Delta_*+1
		\le\bigl(D+\norm P+\norm{P'}\bigr)^{\bigO(1)}.
	\end{equation}

	Fix representations of $\Vec z\in H_0$ and $\Vec z'\in J_0$.
	In the augmented VASS $W:=U[P,P']$, we can first execute the added positive loops at $\Vec p$ to reach $\Vec z$, follow the original run $\rho_0$, and then execute the negative loops at $\Vec q$ to reach $\Vec b'$.
	This gives
	\begin{equation}\label{eq:bilateral_augmented_connection}
		\Vec p(\Vec b)\xrightarrow{\rho_P}\Vec p(\Vec z)
		\xrightarrow{\rho_0}\Vec q(\Vec z')
		\xrightarrow{\rho_{P'}}\Vec q(\Vec b').
	\end{equation}
	The last segment is enabled because each intermediate vector equals $\Vec b'$ plus the periods not yet removed.
	Let $\widetilde V$ be obtained from $V$ by adjoining the corresponding transitions $(0,0,\Vec a)$ for $\Vec a\in P$ and $(0,0,-\Vec a')$ for $\Vec a'\in P'$.
	Both $W$ and $\widetilde V$ have size polynomial in $M$.

	If $W$ is non-wide, choose $r=\ell=0$ and apply~\cref{cor:proper_cone_len} to obtain a run between the unshifted bases of length $M^{\bigO(1)}$.
	Suppose now that $W$ is wide.
	Choose three linearly independent simple-cycle effects $\Vec d_1,\Vec d_2,\Vec d_3$ in $W$, and let
	\begin{equation*}
		\nu:=\bigl|\det(\Vec d_1,\Vec d_2,\Vec d_3)\bigr|,
		\qquad 1\le\nu\le6(n\Delta_*)^3.
	\end{equation*}
	The indices $a,a'$ witness forward and backward $C$-coverability, respectively, but translating the bases by arbitrary indices may destroy their connection.
	We therefore round them up to multiples of $\nu$, ensuring that the displacement between the new bases is an integer combination of cycle effects.
	Let
	\begin{equation*}
		r:=\nu\left\lceil\frac a\nu\right\rceil,
		\qquad \ell:=\nu\left\lceil\frac{a'}\nu\right\rceil,
		\qquad 0\le r,\ell\le D+\nu-1<\kappa_0\le\kappa.
	\end{equation*}
	For every $h$ with $\Vec\delta(h)\ge0$, we have $\Vec b_r(h)\ge\Vec b_a(h)$.
	If $\Vec\delta(h)<0$, the nonnegativity of $\Vec b_\kappa$ gives
	\begin{equation}\label{eq:translated_integer_reserve}
		\Vec b_r(h)\ge(\kappa-r)\lvert\Vec\delta(h)\rvert
		\ge C+L_{\mathrm{cov}}\Delta_*+1.
	\end{equation}
	By~\cref{lem:rackoff_cover}, choose a covering run of length less than $L_{\mathrm{cov}}$ from $\Vec p(\Vec b_a)$ in $U[P_0]$, and replay it from $\Vec p(\Vec b_r)$.
	Coordinates with nonnegative translation follow values no smaller than before, and every other coordinate stays above $C$ by~\eqref{eq:translated_integer_reserve}.
	Thus $\Vec p(\Vec b_r)$ is forward $C$-coverable in $W$.
	The same argument for $\Vec\delta'$ in the reversed system makes $\Vec q(\Vec b'_\ell)$ backward $C$-coverable in $W$.

	Cramer's rule and the divisibility of $r,\ell$ by $\nu$ give integers $t_1,t_2,t_3$ such that
	\begin{equation}\label{eq:bilateral_lattice_correction}
		\ell\cdot\Vec\delta'-r\cdot\Vec\delta
		=\sum_{m=1}^3t_m\cdot\Vec d_m.
	\end{equation}
	Apply the cycle-insertion argument from~\cref{lem:wide_compatible_covers} to the path in~\eqref{eq:bilateral_augmented_connection}, using~\eqref{eq:bilateral_lattice_correction} to obtain a $\ZZ$-run from $\Vec p(\Vec b_r)$ to $\Vec q(\Vec b'_\ell)$ in $W$.
	By~\cref{cor:threshold_closure}, $C$ remains a transformation-area threshold for $\enc{\widetilde V}_B$.
	Hence~\cref{thm:wide_coverable} supplies a run between these bases in $\widetilde V$.

	In both cases, write $\rho$ for the resulting run and $|\rho|\le L\le M^{\bigO(1)}$ for its length bound.
	Removing the added period transitions and applying restricted-period absorption yields the following claim, whose proof is in~\cref{app:translated_period_removal}.
	We are done by this claim.
	\begin{restatable}{claim}{translatedPeriodRemoval}\label{clm:translated_period_removal}
		There exist indices $0\le r,\ell\le\kappa_0$ and members $\Vec h\in H_r$, $\Vec h'\in J_\ell$ whose norms and reachability distance $\MinLen_V((\Vec p,\Vec h),(\Vec q,\Vec h'))$ are at most $M^{\bigO(1)}$.
	\end{restatable}
\end{proof}

\subsection{Completing the Induction on the Number of SCCs}
\label{subsec:one_sided_completion}

We have established reachability-distance bounds for two terminal cases: hybrid sources with a nonzero free period and translated families satisfying the joining conditions.
We now show that every $R$-large bounded-coordinate prefix yields one of these structures.

\paragraph{The \texorpdfstring{$R$}{R}-Large Dichotomy.}
\label{par:large_prefix_dichotomy}

Let $\Vec p(\Vec x)\xrightarrow{\pi_0}\Vec r(\Vec z)$ be a prefix with $i\in\BDCounters(\pi_0,D)$.
Recording coordinate $i$ in states gives a $2$-VASS of size at most $D\Size(U)$.
The following lemma extracts either a hybrid set with a nonzero free period or a translated family with $\kappa+1$ members.
Crucially, the period norms, costs, and translation direction are bounded polynomially in $D$, independently of $\kappa$.
This allows us to choose $\kappa$ before the growth threshold $R$.

\begin{restatable}{lemma}{mixeddichotomylemma}
	\label{lem:large_prefix_dichotomy}
	For every integer $\kappa\ge1$, there exists $D\le R\le(D+\kappa)^{\bigO(1)}$ such that the following holds.
	Let $\Vec p(\Vec x)\xrightarrow{\pi_0}\Vec r(\Vec z)$ be an $R$-large prefix in $U$ whose coordinate $i$ stays below $D$, and let $\Vec s=(\Vec p,\Vec x)$.
	Choose $k\ne i$ with $\Vec z(k)>\norm{\Vec x}+R$, and let $j$ be the remaining coordinate.
	Then one of the following holds:
	\begin{enumerate}
		\item $\Vec z\in H$ for a hybrid set $H$ with a nonzero free period and $\hnorm(H)\le(D+\norm{\Vec x})^{\bigO(1)}$.
		\item There is a translated family $\{H_r\}_{r=0}^{\kappa}$ with
		\begin{equation*}
			H_r=\Vec b+r\cdot\Vec\delta+P^{\langle\Vec c,\cdot\rangle\le K},
			\qquad \Vec\delta=u\cdot\Vec e_j-v\cdot\Vec e_k,
		\end{equation*}
		such that $\Vec z\in H_0$, $u, v \ge 1$, and
		\begin{equation*}
			\norm{\Vec{c}}, \norm{P}, \norm{\Vec{\delta}} \le D^{\bigO(1)}, \quad \norm{\Vec{b}} \le (D + \norm{\Vec{x}} + \kappa)^{\bigO(1)}.
		\end{equation*}
	\end{enumerate}
	Moreover, every member $\Vec h$ of the resulting hybrid set or family satisfies
	\begin{equation}\label{eq:large_prefix_distance}
		\MinLen_V\bigl(\Vec s,(\Vec r,\Vec h)\bigr)
		\le\bigl(D + \norm{\Vec{x}}+\norm{\Vec h}\bigr)^{\bigO(1)}.
	\end{equation}
\end{restatable}

We will first choose an LPS $\Lambda$ with polynomial parameters for the encoding induced by $\pi_0$.
We recall the relevant $2$-VASS results below.
The proof is in~\cref{app:two_counter_tools}.

\begin{restatable}{lemma}{twoCounterTools}
	\label{lem:two_counter_tools}
	Let $U$ be a $2$-VASS with size at most $S$, and let $p,q$ be states of $U$.
	\begin{itemize}
		\item There exists a finite family $\mathcal L$ of LPSs $\Lambda$ satisfying $|\Lambda|,\ \norm{\Lambda}\le S^{\bigO(1)}$ such that
		\[
			\Reach_U^q(p(\Vec x))
			=
			\bigcup_{\Lambda\in\mathcal L}
			\{\Vec y:(\Vec x,\Vec y)\in\Reach_\Lambda\}
			\qquad\text{for every }\Vec x\in\NN^2.
    	\]
		Moreover, whenever $p(\Vec x)\xrightarrow{*}q(\Vec y)$,
		$
			\MinLen_U\bigl(p(\Vec x),q(\Vec y)\bigr)
			\le
			\bigl(S+\norm{\Vec x}+\norm{\Vec y}\bigr)^{\bigO(1)}.
		$
		\item Let
		$
			A=\Vec b_A+P_A^{\langle\Vec c_A,\cdot\rangle\le K_A}
		$
		be a hybrid set at $p$.
		There is a finite family $\mathcal G$ of hybrid sets at $q$ such that
		$
			\Reach_U^q(p(A))
			=
			\bigcup_{G\in\mathcal G}G,
		$
		and every $G = \Vec{b}_G + P_G^{\langle\Vec c_G,\cdot\rangle\le K_G}\in\mathcal G$ satisfies 
		\[
			\hnorm(G)\le\bigl(S+\hnorm(A)\bigr)^{\bigO(1)},\qquad
			\norm{P_G},\norm{\Vec c_G},|P_G|
			\le\bigl(S+\norm{P_A}+\norm{\Vec c_A}\bigr)^{\bigO(1)}.
		\]
	\end{itemize}
\end{restatable}

We now apply~\cref{lem:two_counter_tools} to prove the $R$-large dichotomy.

\begin{proof}[Proof of~\cref{lem:large_prefix_dichotomy}]
	Encode coordinate $i$ and apply~\cref{lem:two_counter_tools}.
	Let $L\le D^{\bigO(1)}$ bound the lengths of all resulting LPSs, and let $E:=L\norm{T}$.
	For the prescribed $\kappa\ge1$, choose
	\begin{equation*}
		R:=\max\{D,\kappa E\}\le(D+\kappa)^{\bigO(1)}.
	\end{equation*}
	Choose an LPS $\Lambda=\alpha_0\beta_1^*\alpha_1\cdots\beta_m^*\alpha_m$ containing a witnessing run.
	Delete unused cycle blocks and zero-effect cycles, preserving a witnessing run.
	In the following claims, vectors are projected onto coordinates $j,k$, while the states record coordinate $i$.
	Let $\mathcal S_\Lambda\subseteq\NN^2$ be the set of vectors reached from the encoded source through $\Lambda$ at its fixed final state.
	The first claim handles the free-period case.

	\begin{restatable}{claim}{mixedFreePeriod}
		\label{clm:mixed_free_period}
		If a cycle used by the chosen run has nonzero nonnegative effect, there exists a hybrid set $H$ with a nonzero free period such that $\Vec z\in H\subseteq\mathcal S_\Lambda$ and $\hnorm(H)\le(D+\norm{\Vec x})^{\bigO(1)}$.
	\end{restatable}

	For the remaining case, assume that no used cycle has nonzero nonnegative effect.
	Every used cycle increasing coordinate $k$ must then decrease coordinate $j$.
	The large growth in coordinate $k$ forces such a cycle to occur at least $\kappa$ times.
	Choosing the last cycle block with this property ensures that the remaining suffix has enough reserve in coordinate $k$ to delete up to $\kappa$ repetitions, even at intermediate positions within cycles.
	Hence, we have the following claim.

	\begin{restatable}{claim}{mixedLongCycle}
		\label{clm:mixed_long_cycle}
		There exists a run $\eta_-\beta^\kappa\eta_+$ from the encoded source to $\Vec z$, with vector $\Vec a$ immediately before $\beta^\kappa$, satisfying the following properties:
		\begin{enumerate}
			\item $\beta$ is a cycle with $\Eff(\beta)=-u\cdot\Vec e_j+v\cdot\Vec e_k$, where $1\le u,v\le E$.
			\item The paths $\eta_-,\eta_+$ belong to LPSs $\Lambda_-,\Lambda_+$ obtained by cutting $\Lambda$ after the selected $\beta$-star, with that star included in $\Lambda_-$.
			Their sizes are at most $D^{\bigO(1)}$, independently of $\kappa$.
			\item Every vector $\Vec w$ along $\eta_+$ satisfies $\Vec w(k)\ge\kappa v$.
		\end{enumerate}
	\end{restatable}

	Fix the run $\eta_-\beta^\kappa\eta_+$ and the LPSs $\Lambda_-,\Lambda_+$ supplied by~\cref{clm:mixed_long_cycle}.
	Deleting copies of $\beta$ shifts the final vector along $\Vec\delta=-\Eff(\beta)=u\cdot\Vec e_j-v\cdot\Vec e_k$.
	The lower bound along $\eta_+$ ensures that these deletions preserve nonnegativity.
	Using hybrid descriptions of $\Lambda_-$ and $\Lambda_+$, we obtain the following translated family.
	The full proofs are given in~\cref{app:mixed_dichotomy_claims}.

	\begin{restatable}{claim}{mixedTranslatedFamily}\label{clm:mixed_translated_family}
		There is a translated family $\{H_s\}_{s=0}^{\kappa}$ with $\Vec z\in H_0$ and direction $\Vec\delta=u\cdot\Vec e_j-v\cdot\Vec e_k$.
		Its periods and costs have norm at most $D^{\bigO(1)}$, and its base has norm at most $(D+\norm{\Vec x}+\kappa)^{\bigO(1)}$.
		Every member is reachable from the encoded source at the fixed encoded target state.
	\end{restatable}

	These claims establish the two outcomes of the dichotomy.
	In both cases, every member of the resulting hybrid set or translated family is reachable from the fixed source in the bounded-coordinate $2$-VASS.
	The short-witness bound in~\cref{lem:two_counter_tools}, followed by lifting to $V$, gives~\eqref{eq:large_prefix_distance} and completes the proof.
\end{proof}

\paragraph{Completion of the Induction.}
\label{par:completing_semi_coverable_induction}
The two outcomes of~\cref{lem:large_prefix_dichotomy} yield a polynomial distance bound for every $R$-large prefix satisfying the bounded-prefix stopping condition, leaving only $R$-small prefixes for the inductive step.

\begin{restatable}[$R$-large: terminal case]{lemma}{phaseLargeTerminal}
	\label{lem:large_prefix_terminal}

	There exists $D\le R\le D^{\bigO(1)}$ such that the following holds.
	Let $\Vec p(\Vec x)\xrightarrow{\pi_0}\Vec r(\Vec z)\xrightarrow{\pi_1}\Vec q(\Vec y)$ be a run in $U$ whose target $\Vec q(\Vec y)$ is backward $C$-coverable, where $\Vec s=(\Vec p,\Vec x)$ and $\Vec t=(\Vec q,\Vec y)$.
	Assume that $\pi_0$ is $R$-large and $i\in\BDCounters(\pi_0,D)$.
	If the configuration $\Vec r(\Vec w)$ is forward $C$-coverable whenever $\Vec w(i)\ge\Vec z(i)$ and $\Vec w(j)\ge D$ for all $j\ne i$, then
	\begin{equation*}
		\MinLen_V(\Vec s,\Vec t)
		\le
		\bigl(D+\norm{\Vec x}+\norm{\Vec t}\bigr)^{\bigO(1)}.
	\end{equation*}
\end{restatable}
\begin{proof}[Proof sketch]
	The extracted period norms are polynomially bounded in $D$, independently of $\kappa$.
	We can therefore choose $\kappa\ge \max\{2D,\kappa_0\}$ given in~\cref{lem:translated_join}, and then choose $R$ by~\cref{lem:large_prefix_dichotomy}, both polynomially bounded in $D$.
	Apply the dichotomy.
	In the free-period case, invoke~\cref{lem:free_covered}.
	In the translated-family case, the covering hypothesis and the choice of $\kappa$ give $C$-joinability with the constant target family $J_\ell=\{\Vec y\}$, so invoke~\cref{lem:translated_join}.
	In either case, prepend the short prefix supplied by~\eqref{eq:large_prefix_distance}.
	The full proof is in~\cref{app:large_prefix_terminal}.
\end{proof}

\begin{proof}[Proof of~\cref{thm:semi_coverable}]
	Fix $V,B,C$ and $\Vec t=(\Vec q,\Vec y)$, and assume that $\Vec q(\Vec y)$ is backward $C$-coverable.
	Choose $D$ uniformly for all sequential subinstances by~\cref{lem:bounded_prefix}, and then choose $\kappa,R$ as in the proof of~\cref{lem:large_prefix_terminal}, so that
	\begin{equation}\label{eq:semi_uniform_parameters}
		B+C+\Size(V)\le D\le(B+C+\Size(V))^{\bigO(1)},
		\qquad D\le R\le D^{\bigO(1)}.
	\end{equation}
	For $k\ge1$, we construct a nondecreasing bound $L_k(N)$ on $\MinLen_V(\Vec s,\Vec t)$, uniform over reachable sequential subinstances with at most $k$ components, source norm at most $N$, and backward $C$-coverable target $\Vec q(\Vec y)$.
	Let $M:=D+\norm{\Vec t}$, and choose a fixed nondecreasing polynomial $P$ large enough for the bounds below.
	Its degree and the choices of $D,R$ do not depend on $k$.

	\emph{Base Case: $k=1$.}
	If $\Vec p(\Vec x)$ is forward $C$-coverable, apply~\cref{cor:coverable_case}.
	Otherwise,~\cref{lem:bounded_prefix} captures the entire run in a $2$-VASS of size at most $D\Size(U)\le D^{\bigO(1)}$, since there is no later component.
	Apply~\cref{lem:two_counter_tools} and lift the resulting run to $V$.
	In both cases,
	\begin{equation*}
		L_1(N)\le P(M+N).
	\end{equation*}

	For $k \ge 1$, assume that $L_k$ has been constructed.
	Consider a run from $\Vec p(\Vec x)$ to $\Vec q(\Vec y)$ in a sequential subinstance $U$ with at most $k+1$ components and $\norm{\Vec s}\le N$.
	We distinguish the terminal and inductive cases.

	\emph{Terminal Cases.}
	If $\Vec p(\Vec x)$ is forward $C$-coverable, apply~\cref{cor:coverable_case}.
	Otherwise, write the run as $\pi_0\pi_1$ by~\cref{lem:bounded_prefix}.
	If $\pi_1=\varepsilon$, apply~\cref{lem:two_counter_tools} to the $2$-VASS encoding of $\pi_0$ and lift the run to $V$.
	If $\pi_1\ne\varepsilon$ and $\pi_0$ is $R$-large, its target satisfies the covering condition supplied by~\cref{lem:bounded_prefix}.
	By applying~\cref{lem:large_prefix_terminal}, the reachability-distance bound holds.
	Each case gives a run in $V$ with
	\begin{equation*}
		\MinLen_V(\Vec s,\Vec t)\le P(M+N).
	\end{equation*}

	\emph{Inductive Case.}
	The remaining case has an $R$-small prefix $\pi_0$ ending at a later component entry $\Vec r(\Vec z)$.
	Let $\Vec u:=(\Vec r,\Vec z)$.
	Since $R\ge D\ge B$,
	\begin{equation}\label{eq:source_restart_invariant}
		\norm{\Vec u}\le\max\{B,\norm{\Vec x}+R\}
		\le\norm{\Vec s}+R\le N+R.
	\end{equation}
	By~\cref{lem:two_counter_tools}, replace $\pi_0$ by a run in $V$ from $\Vec s$ to $\Vec u$ of length at most $P(M+N+R)$.
	The unchanged suffix has at most $k$ components and preserves the target's backward covering run, which stays in its component.
	Apply the induction hypothesis to this suffix and prepend the shortened prefix.
	Together with the terminal-case bounds, this allows us to choose
	\begin{equation}\label{eq:one_covered_component_recurrence}
		L_{k+1}(N)\le P(M+N+R)+L_k(N+R).
	\end{equation}
	Note that $k\le(B+1)^2$.
	By~\eqref{eq:semi_uniform_parameters}, solving this recurrence gives
	\begin{equation*}
		\begin{aligned}
			\MinLen_V(\Vec s,\Vec t) \le L_k(\norm{\Vec s})\le(B+C+\Size(V,\Vec s,\Vec t))^{\bigO(1)}.
		\end{aligned}
	\end{equation*}
	The degree is independent of $k$, proving~\eqref{eq:semi_coverable_length}.
\end{proof}

%% file: contents/8-Noncoverable.tex
\section{Removing the Coverability Assumption}
\label{sec:removing_coverability}

We now remove the coverability assumption from~\cref{thm:semi_coverable}, using the same additive-growth induction.
Pumpable sources and hybrid sets with a nonzero free period can be handled using the semi-coverable result.
The remaining difficulty arises when an $R$-large prefix yields a translated family but the target is not backward $C$-coverable.

To handle this case, we reverse the remaining suffix and regard the extracted translated family as a fixed target family.
We show that an arbitrary source can reach a suitable member of such a family by a polynomial-length run, using~\cref{lem:translated_join} when another translated family arises.
This provides the missing terminal case for the final induction.

Throughout this section, fix a $5$-VAS $V$, an encoding bound $B\in\NN$, and a sequential sub-VASS $U$ of $\enc V_B$.
Let $C$ be a transformation-area threshold for $\enc V_B$, and choose $D$ uniformly as in~\eqref{eq:uniform_structural_threshold}.

\subsection{Pumpable Sources and Hybrid Sets with a Nonzero Free Period}
\label{subsec:terminal_hybrids}

We first extend the pumpable-source bound to arbitrary targets.
The extraction in~\cref{lem:pump_extraction} does not require backward coverability of the target.
It produces a forward $C$-coverable hybrid set, unless the extraction process reaches the target, which already gives a short witness.
We can therefore apply~\cref{thm:semi_coverable} after adjoining its periods as self-loops.
The proof is given in~\cref{app:pumpable_arbitrary}.

\begin{restatable}{lemma}{pumpableArbitrary}
	\label{lem:pumpable_arbitrary}
	Suppose $\Vec p(\Vec x)$ is pumpable in $U$ and $\Vec p(\Vec x)\xrightarrow{*}\Vec q(\Vec y)$ in $U$.
	For $\Vec s=(\Vec p,\Vec x)$ and $\Vec t=(\Vec q,\Vec y)$,
	\[
		\MinLen_V(\Vec s,\Vec t)
		\le\bigl(D+\norm{\Vec s}+\norm{\Vec t}\bigr)^{\bigO(1)}.
	\]
\end{restatable}

The preceding lemma handles fixed pumpable sources without requiring target coverability.
We now extend this result to hybrid sources and targets, allowing the source base to be either forward $C$-coverable or pumpable after adjoining its free periods.

\begin{restatable}{lemma}{hybridTerminal}
	\label{lem:hybrid_terminal}
	Let $H=\Vec a+P^{\langle\Vec c,\cdot\rangle\le K}$ and $J$ be hybrid sets at the source state $\Vec p$ and target state $\Vec q$ of $U$, respectively.
	Write $P_0$ for the free periods of $H$.
	Suppose that $\Vec p(\Vec x)\xrightarrow{*}\Vec q(\Vec y)$ in $U$ for some $\Vec x\in H$ and $\Vec y\in J$.
	If $\Vec p(\Vec a)$ is forward $C$-coverable or pumpable in $U[P_0]$, then there exist $\Vec h\in H$ and $\Vec h'\in J$ such that
	\begin{equation}\label{eq:hybrid_terminal_member}
		\norm{\Vec h},\ \norm{\Vec h'},\quad
		\MinLen_V\bigl((\Vec p,\Vec h),(\Vec q,\Vec h')\bigr)
		\le\bigl(D+\hnorm(H)+\hnorm(J)\bigr)^{\bigO(1)}.
	\end{equation}
\end{restatable}

The proof adds the source periods and the negatives of the target periods as self-loops.
The augmented source is forward $C$-coverable or pumpable, so~\cref{thm:semi_coverable} or the preceding pumpable-source bound gives a short run between the bases.
Absorbing restricted periods yields the required members of the original hybrid sets.
The details are given in~\cref{app:hybrid_terminal}.

In particular, taking the target hybrid set to be a singleton handles the free-period outcome of the $R$-large dichotomy without any target coverability assumption.
It remains to handle the translated-family outcome.

\subsection{Reaching a Fixed Translated Family}
\label{subsec:family_arbitrary_source}

We now extend the induction of~\cref{sec:semi_coverable} to a fixed translated-family target with a backward $C$-coverable member.
The first ingredient is a small-intersection property for hybrid sets, with a bound independent of their cost bounds.
Its proof is given in~\cref{app:hybrid_intersection}.

\begin{restatable}{lemma}{phaseIntersection}\label{lem:hybrid_intersection}
	Let $H,J\subseteq\NN^3$ be hybrid sets with $H\cap J\ne\emptyset$.
	There is $\Vec h\in H\cap J$ such that
	\begin{equation*}
		\norm{\Vec h}\le\bigl(D+\hnorm(H)+\hnorm(J)\bigr)^{\bigO(1)}.
	\end{equation*}
\end{restatable}

We now establish the distance bound for an arbitrary source and a fixed translated-family target.
The induction follows~\cref{thm:semi_coverable}, keeping the target family fixed and replacing its terminal cases by the preceding hybrid-set bounds.

\begin{restatable}{lemma}{translatedArbitrary}
	\label{lem:translated_arbitrary}
	Let $\mathcal J=\{J_\ell\}_{0\le\ell\le\kappa}$ be a translated family at the target state $\Vec r'$ of $U$ with
	\[
		J_\ell
		=
		\Vec b'+\ell\cdot\Vec\delta'
		+P'^{\langle\Vec c',\cdot\rangle\le K'}.
	\]
	Suppose that $J_a$ is backward $C$-coverable in $U$ for some $0\le a\le\min\{D,\kappa\}$.
	Suppose that $\Vec p(\Vec x)\xrightarrow{*}\Vec r'(\Vec z')$ in $U$ for some $\Vec z'\in J_0$.
	There exists $\kappa_0 \le \bigl(D+\norm{P'}\bigr)^{\bigO(1)}$ such that, if $\kappa\ge\kappa_0$, then some $0\le\ell\le\kappa_0$ and $\Vec h'\in J_\ell$ satisfy
	\begin{equation}\label{eq:translated_arbitrary_length}
		\norm{\Vec h'},
		\quad
		\MinLen_V\bigl(\Vec s,(\Vec r',\Vec h')\bigr)
		\le
		\bigl(
			D+\norm{\Vec s}
			+\hnorm(J_0)
			+\norm{\Vec\delta'}
		\bigr)^{\bigO(1)},
		\qquad
		\Vec s=(\Vec p,\Vec x).
	\end{equation}
\end{restatable}

\begin{proof}[Proof sketch]
	Follow the induction on the number of SCCs in~\cref{thm:semi_coverable}, keeping $\mathcal J$ fixed.
	There are four terminal cases.
	A forward $C$-coverable source is handled by~\cref{lem:hybrid_terminal}.
	A bounded-coordinate prefix reaching $J_0$ is handled by~\cref{lem:hybrid_intersection,lem:two_counter_tools}.
	An $R$-large prefix yielding a nonzero free period is handled by~\cref{lem:hybrid_terminal}.
	An $R$-large prefix yielding a translated family is handled by~\cref{lem:translated_join}.
	Only an $R$-small prefix ending at a later component requires induction, with an additive increase of at most $R$ in the source norm.
	The full case analysis is in~\cref{app:translated_arbitrary}.
\end{proof}

\subsection{The Final Induction}
\label{subsec:final_induction}

We are now ready to complete the proof of~\cref{prop:encoding_reachability}.

\begin{proof}[Proof sketch of~\cref{prop:encoding_reachability}]
	Use additive-growth induction as in~\cref{thm:semi_coverable}.
	A forward $C$-coverable source is terminal by~\cref{thm:semi_coverable}.
	Otherwise, extract a bounded-coordinate prefix.
	If it reaches the target, apply~\cref{lem:two_counter_tools}.
	If it is $R$-large, use~\cref{lem:hybrid_terminal} for a nonzero free period, or reverse the suffix and use~\cref{lem:translated_arbitrary} for a translated family.
	Only an $R$-small prefix ending at a later component requires induction, increasing the source norm by at most $R$.
	The full proof is in~\cref{app:encoding_reachability}.
\end{proof}

%% file: appendix/Appendix-A-Overview.tex
\section{Proofs for Main Results}
\label{app:overview}

\subsection{Proof of~\texorpdfstring{\cref{thm:main_complexity}}{Theorem~\ref*{thm:main_complexity}}}
\label{app:main_complexity}

\mainComplexity*

Recall that the size $\Size(V,\Vec{s}, \Vec{t})$ is defined under unary encoding as in~\cref{sec:preliminary}.
For $d$-VAS, we use the following measure, which is polynomially equivalent to the binary encoding length:
\begin{equation*}
	\Size_{\textsc{bin}}(V,\Vec s,\Vec t) := 1 + d|T|\lceil\log_2(\norm T+2)\rceil + d\lceil\log_2(\norm{\Vec s}+2)\rceil + d\lceil\log_2(\norm{\Vec t}+2)\rceil.
\end{equation*}
The main theorem follows by a nondeterministic brute-force algorithm as follows.

\begin{proof}[Proof of \cref{thm:main_complexity}]
    The \NL{}-hardness and the \PSPACE{}-hardness are given in ~\cite[Theorem 1.2]{DBLP:conf/concur/ChenFZ26} and~\cite[Theorem 1]{DBLP:conf/rp/KaminskiL26}, respectively.
    We only focus on the upper bound.

	Let $c\ge1$ be an integer constant satisfying~\cref{thm:polynomial_witness}.
	Construct a nondeterministic Turing machine that, given a $5$-VAS $(V,\Vec s,\Vec t)$ of size $M:=\Size(V,\Vec s,\Vec t)$, guesses at most $\ell:=M^c$ transitions and verifies that they induce a run from $\Vec s$ to $\Vec t$.
	By this theorem, if $\Vec{s} \xrightarrow{*} \Vec{t}$, then there exists a witness of length at most $\ell$.
	Hence the machine accepts if and only if $\Vec{s} \xrightarrow{*} \Vec{t}$ in $V$.

	Then we analyze the space complexity of this algorithm.
	The machine needs to store the current configuration, which has a norm at most $\norm{\Vec{s}} + \ell \cdot M  \le (\ell + 1) \cdot M \le 2M ^ {c + 1}$.
	The five counters, a step counter, and an index of the current transition therefore require $O(\log M)$ space.
	The transition can be read from the input using that index.
	Also $\ell$ can be computed within the same space.
	For unary encoding, this uses logarithmic space, and hence reachability in unary $5$-VAS is in $\NL$.
	For binary encoding, note that $\log M = O(\Size_{\textsc{bin}}(V, \Vec{s}, \Vec{t}))$.
	So the same algorithm uses polynomial space, and hence reachability in binary $5$-VAS is in $\NPSPACE = \PSPACE$.
\end{proof}

\subsection{Proof of~\texorpdfstring{\cref{cor:flatness_of_5vas}}{Corollary~\ref*{cor:flatness_of_5vas}}}
\label{app:mixed_flattenability}

\flatnessFiveVAS*

\begin{proof}[Proof of \cref{cor:flatness_of_5vas}]
Fix a polynomial factorization of $\Reach_V$ given by \cref{thm:intro_factorization}.
Let $R$ be a factor captured by $\enc{V}_{I, B}$ where $|I| = 3$ and $B$ is polynomial in $\Size(V)$.
Applying~\cite[Theorem~3.1]{DBLP:journals/jacm/BlondinEFGHLMT21} to this $2$-VASS and restoring the original transition effects of $V$ gives a finite family $\mathcal S_R$ of $\NN$-LPSs satisfying
\begin{equation*}
    R\subseteq \bigcup_{\Lambda\in \mathcal S_R}\Reach_\Lambda\subseteq\Reach_V.
\end{equation*}
The second holds because the resulting $\NN$-runs use only transitions of $V$.
We still have $|\Lambda| \le \PolyF{\Size(V)}$ and $\norm{\Lambda} \le \Size(V)$.

Let $\mathcal S$ be the finite family of concatenations obtained by choosing an LPS from $\mathcal S_R$ for each bounded-counter factor and keeping the mixed-LPS factors.
These concatenations cover $\Reach_V$ by the original factorization.
Conversely, each chosen mixed LPS captures a subset of $\Reach_V$, so transitivity gives $\Reach_V=\bigcup_{\Lambda\in \mathcal S}\Reach_\Lambda$.
Since the number of factors is polynomially bounded, so are the parameters of the resulting mixed LPSs.
\end{proof}

\subsection{Proof of~\texorpdfstring{\cref{thm:higher_dimensional_complexity}}{Theorem~\ref*{thm:higher_dimensional_complexity}}}
\label{app:higher_dimensional_complexity}

\dVASReachInFdmThree*

We use $F_d$ for the fast-growing functions and $\mathsf F_d$ for the associated complexity classes.
See~\cite[Section~2]{DBLP:journals/toct/Schmitz16} for more details.

\begin{proof}[Proof of \cref{thm:higher_dimensional_complexity}]
Let $n$ be the bit-length of the input.
The unary size is at most $2^{\PolyF{n}}$, so \cref{lem:higher_dimensional_shortest_runs} gives a run of length at most $F_d(E(n))$ for some elementary function $E$ whenever the target is reachable.
Exhaustive search in $V$ up to this length decides reachability.
For fixed $d\ge3$, the preprocessing and search overhead are absorbed in $\mathsf F_d$~\cite[Lemma~4.6 and Section~4.2.1]{DBLP:journals/toct/Schmitz16}.
\end{proof}

\begin{lemma}
\label[lemma]{lem:higher_dimensional_shortest_runs}
For every fixed $d\ge3$, there is an elementary function $g_d$ such that every reachable $(d+3)$-VAS instance $(V,\Vec s,\Vec t)$ satisfies
\begin{equation*}
    \MinLen_V(\Vec s,\Vec t)\le F_d\bigl(g_d(\Size(V,\Vec s,\Vec t))\bigr).
\end{equation*}
\end{lemma}

\begin{proof}
Put $D:=d+3$ and $M:=\Size(V,\Vec s,\Vec t)$.
Apply \cref{thm:intro_factorization} and choose a witnessing composition for $\Vec s\xrightarrow{*}\Vec t$ with polynomially bounded factor count and parameters.
We build a VASS $W$ using the original $D$ counters and dedicated auxiliary counters that start and finish at zero.

\begin{claim}
\label{clm:integer_lps_simulation}
    Let $\Lambda$ be a $\ZZ$-LPS over $d$ counters.
    There exists a $(d+1)$-dimensional $\NN$-LPS $\widehat\Lambda$ such that
    \[
        \Reach_{(\Lambda,\ZZ)}
        =
        \left\{
            (\Vec s,\Vec t):
            ((\Vec s,0),(\Vec t,0))
            \in\Reach_{\widehat\Lambda}
        \right\}.
    \]
    Moreover, $|\widehat\Lambda|\le|\Lambda|+2$,
    $\norm{\widehat\Lambda}\le\max\{\norm\Lambda,1\}$.
\end{claim}

\begin{claimproof}
    We append a zero auxiliary coordinate to every vector of $\Lambda$ to obtain a $(d+1)$-dimensional $\ZZ$-LPS $\Lambda^\uparrow$, and define
    \[
        \widehat\Lambda
        :=(\Vec1,1_h)^*\Lambda^\uparrow(-\Vec1,-1_h)^*.
    \]
    Every $\ZZ$-witness of $\Lambda$ with nonnegative endpoints can be executed under $\NN$-semantics by repeating the positive vector sufficiently many times before it and the negative vector equally many times afterward.
    Conversely, a witness of $\widehat\Lambda$ whose auxiliary coordinate starts and ends at zero uses these two vectors equally often.
    Removing their occurrences gives a $\ZZ$-witness of $\Lambda$.
    The parameter bounds follow directly from the construction.
\end{claimproof}

For a mixed LPS $\Lambda$, discard empty components and apply \cref{clm:integer_lps_simulation} to each $\ZZ$-component, using distinct auxiliary counters that are unchanged elsewhere and start and finish at zero.
Write the resulting $\NN$-LPS as $\widehat{\Lambda}=\alpha_0\beta_1^*\alpha_1\cdots\beta_k^*\alpha_k$.
We construct a VASS $V_{\widehat{\Lambda}}$ by realizing each $\alpha_i$ as a directed path and each $\beta_i$ as a simple cycle, using fresh states for distinct occurrences.
Connect these paths and cycles in their order in $\widehat{\Lambda}$ by forward zero-effect transitions, and let $p_{\widehat{\Lambda}}$ and $q_{\widehat{\Lambda}}$ be the entry and exit states.
Every run from $p_{\widehat{\Lambda}}$ to $q_{\widehat{\Lambda}}$ follows the fixed paths and repeats each cycle an arbitrary number of times.
Consequently, for all $\Vec s,\Vec t\in\NN^D$,
\[
    (\Vec s,\Vec t)\in\Reach_\Lambda
    \iff
    p_{\widehat{\Lambda}}(\Vec s, \Vec 0)\xrightarrow{*}q_{\widehat{\Lambda}}(\Vec t, \Vec 0)
    \text{ in }V_{\widehat{\Lambda}}.
\]
Deleting the connecting transitions recovers an LPS witness of no greater length.
Moreover, $V_{\widehat{\Lambda}}$ has $O(|{\Lambda}|+1)$ states and transitions, with transition norm at most $\max\{\norm{{\Lambda}},1\}$.
Each SCC contains at most one simple cycle, so its cycle-dimension is at most one.

For a bounded-counter factor captured by $\enc{V}_{I,B}$ with $|I|=3$, we restore all $D$ original transition coordinates, and add a private entry and exit connected to the encoded states by zero-effect edges.
This gadget has polynomial size and cycle-dimension at most $d$, since every cycle has zero effect on the three encoded coordinates.
Its runs become runs of $V$ after deleting the entry and exit edges, and it admits every witness of the factor.

Give the factor gadgets disjoint states and auxiliary counters, and connect them in sequence by forward zero-effect edges to obtain $W$.
Let $s':=p_{\mathrm{in}}(\Vec s,\Vec0)$ and $t':=p_{\mathrm{out}}(\Vec t,\Vec0)$ be the global entry and exit configurations.
The chosen witnesses give $s'\xrightarrow{*}t'$ in $W$, with $\Size(W,s',t')\le\PolyF{M}$ and $\dimcom(W)\le d$, since every SCC stays inside one gadget.
Every run from $s'$ to $t'$ restores each auxiliary counter to zero.
Deleting the padding loops and connecting edges therefore recovers mixed-LPS witnesses and runs of $V$ for the bounded-counter gadgets.
By \cref{thm:intro_factorization}, these yield a run of $V$ of no greater length.
Thus~\cite[Theorem~2.1]{DBLP:conf/concur/ChenFZ26} gives
\begin{equation*}
    \MinLen_V(\Vec s,\Vec t)
    \le\MinLen_W(s',t')
    \le F_d\bigl(g(\Size(W,s',t'))\bigr),
\end{equation*}
where $g$ is a fixed nondecreasing elementary function, independent of the number of counters.
Composing $g$ with the polynomial size bound gives $g_d$.
\end{proof}

%% file: appendix/Appendix-B-Transformation-Area.tex
\section{Quantifying van Leeuwen's Transformation Area}
\label{app:classical_transformation_area}

Fix a $d$-VAS $V$ with transition set $T$.
A set $A\subseteq\NN^d$ is a transformation area if every $\ZZ$-run between configurations in $A$ can be rearranged into a run in $V$ with the same Parikh image.
For $i\in[d]$ and $C\in\NN$, define
\begin{equation*}
	\TA_C^i := \{\Vec{x}\in\NN^d:\Vec{x}(j)\ge C\text{ for every }j\in[d]\setminus\{i\}\},
	\qquad
	\TA_C := \bigcup_{i\in[d]}\TA_C^i.
\end{equation*}
The original result of van Leeuwen~\cite[Theorem 5.3]{DBLP:conf/stoc/Leeuwen74} states that for every $i\in[d]$, there exists an effectively computable threshold $C\in\NN$ such that $\TA_C^i$ is a transformation area.
We give an explicit threshold linear in $d$ and prove that the union $\TA_C$ is also a transformation area.

\begin{theorem}
	\label{thm:coor_ta_bound}
	Let $C:=2d\norm{T}(\norm{T}+1)$.
	Then $\TA_C$ is a transformation area of $V$.
\end{theorem}

We first show how to rearrange a $\ZZ$-run whose endpoints lie in the same $\TA_C^i$ into a valid run.

\begin{lemma}
	\label{lem:coor_ta_bound}
	For every $i\in[d]$, the set $\TA_{C_0}^i$ is a transformation area of $V$, where $C_0:=2d\norm{T}$.
\end{lemma}

\begin{proof}
	Fix $i\in[d]$ and a $\ZZ$-run $\Vec s\xrightarrow{\pi}_{\ZZ}\Vec t$ with $\Vec s,\Vec t\in\TA_{C_0}^i$.
	Write $\pi=\Vec a_1\cdots\Vec a_n$ and $\Vec z:=\Vec t-\Vec s=\sum_{\ell=1}^n\Vec a_\ell$.
	The case $n=0$ is immediate.
	As in the proof of~\cref{cor:steinitz_partial_sums}, append $n$ copies of $-\Vec z/n$ and apply~\cref{lem:steinitz}.
	Let $\Vec P_0,\ldots,\Vec P_{2n}$ be the partial sums in the resulting order, so that
	\begin{equation*}
		\Vec P_0=\Vec P_{2n}=\Vec0,
		\qquad
		\norm{\Vec P_k}\le d\norm{T}.
	\end{equation*}
	Choose $r\in[0,2n-1]$ minimizing $\Vec P_r(i)$, and rotate the sequence to begin just after $\Vec P_r$.
	Since the total sum is zero, every partial sum $\Vec R$ of the rotated sequence has the form $\Vec P_k-\Vec P_r$, including those after the sequence wraps around.
	Thus
	\begin{equation}\label{eq:rotated_steinitz}
		\Vec R(i)\ge0,
		\qquad
		\Vec R(j)\ge-2d\norm{T}\text{ for }j\ne i.
	\end{equation}
	Delete the auxiliary vectors from the rotated sequence, and let $\rho$ be the resulting path.
	For a nonempty prefix of $\rho$, let $\Vec p$ be its effect and let $\Vec R$ be the partial sum of the rotated augmented sequence ending at the same original occurrence.
	If that augmented prefix contains $m$ auxiliary vectors, then
	\begin{equation*}
		\Vec{s} + \Vec p=\Vec{s} + \Vec R+(m/n)\cdot\Vec z \ge \Vec{R} + \Vec{s} + \min(\Vec{0}, \Vec{t} - \Vec{s})= \Vec{R} + \min(\Vec{s}, \Vec{t}).
	\end{equation*}
	Using~\eqref{eq:rotated_steinitz}, we obtain $\Vec s(i)+\Vec p(i)\ge\Vec R(i)+\min(\Vec s(i),\Vec t(i))\ge0$.
	For every $j\ne i$, the endpoints instead satisfy $\min(\Vec s(j),\Vec t(j))\ge C_0$, so $\Vec s(j)+\Vec p(j)\ge-2d\norm{T}+C_0=0$.
	Hence $\rho$ is a run from $\Vec s$ to $\Vec t$.
	Every original occurrence is used exactly once, so $\Parikh{\rho}=\Parikh{\pi}$.
\end{proof}

For the union, a short prefix brings every coordinate to at least $C_0$, after which~\cref{lem:coor_ta_bound} applies.
The larger threshold $C=(\norm{T}+1)C_0$ pays for the decrease along that prefix.

\begin{proof}[Proof of~\cref{thm:coor_ta_bound}]
	Let $C_0:=2d\norm{T}$, so $C=(\norm{T}+1)C_0$.
	Consider a $\ZZ$-run $\Vec s\xrightarrow{\pi}_{\ZZ}\Vec t$ with $\Vec s\in\TA_C^i$ and $\Vec t\in\TA_C^j$.
	If $\Vec s,\Vec t\in\TA_{C_0}^j$, then~\cref{lem:coor_ta_bound} rearranges $\pi$ into the required run.

	Otherwise $\Vec s\in\TA_{C_0}^i\setminus\TA_{C_0}^j$, which implies that $i\ne j$ and $\Vec s(i)<C_0$.
	Since $\Vec t\in\TA_C^j$ and $i\ne j$, we have $\Vec t(i)\ge C$.
	For the occurrences of $\pi$ with strictly positive $i$-th component, the sum of these components is therefore at least $\Vec t(i)-\Vec s(i)\ge C_0-\Vec s(i)$.
	Choose $\alpha$ as a subsequence consisting only of such occurrences.
	Take such occurrences until that coordinate first reaches $C_0$.
	Each increases it by at least one, so $|\alpha|\le C_0-\Vec s(i)\le C_0$.
	Coordinate $i$ is nondecreasing along $\alpha$, and every other coordinate stays at least
	\begin{equation*}
		C-\norm{T}|\alpha|\ge C-\norm{T}\cdot C_0=C_0.
	\end{equation*}
	Thus $\alpha$ induces a run to some $\Vec w\ge C_0\cdot\Vec1$.

	The Parikh image $\Parikh{\pi}-\Parikh{\alpha}$ gives a $\ZZ$-run from $\Vec w$ to $\Vec t$ in $V$.
	Since $\Vec{w},\Vec{t}\in\TA_{C_0}^j$, \cref{lem:coor_ta_bound} rearranges this $\ZZ$-run into a run $\beta$.
	Then $\alpha\beta$ is a run from $\Vec s$ to $\Vec t$ with Parikh image $\Parikh{\pi}$.
	This proves that $\TA_C$ is a transformation area.
	Each $\TA_C^i$ is a subset of $\TA_C$, so it is a transformation area as well.
\end{proof}

%% file: appendix/Appendix-C-Factorization.tex
\section{Proofs for~\texorpdfstring{\cref{sec:reachability_factorization}}{Section~\ref*{sec:reachability_factorization}}}
\label{app:reachability_factorization}

The following is the formal version of \cref{thm:intro_factorization}.

\begin{restatable}{theorem}{FactorabledVAS}
\label{thm:vas_factorability}
For every fixed dimension $d\ge3$, there exists a polynomial $P_d$ such that every $d$-VAS $V$ admits a finite family $\mathcal{F}$ of sequences of relations $(R_{1},\ldots,R_{m})$ satisfying the following properties:
\begin{equation*}
    \Reach_{V}
    =
    \bigcup_{(R_{1},\ldots,R_{m})\in\mathcal{F}}
    R_{1}\circ\cdots\circ R_{m}.
\end{equation*}
Every sequence in $\mathcal{F}$ has $1\le m\le P_d(\Size(V))$.
Each factor is $R_i=\Reach_{V_i}\subseteq\Reach_V$, where $V_i$ has one of the following forms:
\begin{enumerate}
    \item $V_i=\enc{V}_{I,B}$ for some $I\subseteq[d]$ with $|I|=3$ and $B\le P_d(\Size(V))$;
    \item $V_i=\Lambda$ for some mixed LPS $\Lambda$, with $\norm{\Lambda},|\Lambda|\le P_d(\Size(V))$.
\end{enumerate}
Moreover, for every factor $R_i$ and every $(\Vec{x},\Vec{y})\in R_i$,
\begin{equation*}
    \MinLen_{V}(\Vec{x},\Vec{y}) \leq \MinLen_{V_{i}}(\Vec{x},\Vec{y}).
\end{equation*}
\end{restatable}

\begin{proof}
    Let $N:=\Size(V)$ and fix a nondecreasing polynomial $P_0$ sufficiently large for all applications of~\cref{lem:encoding_factorability}, including those with at most two zero-effect dummy coordinates.
    Set $B_0:=0$ and $B_{k+1}:=P_0(B_k+N)$ for $k\in\{0,1,2\}$.

    We construct finite families $\mathcal V_0,\ldots,\mathcal V_3$ of sequences of factors, starting with $\mathcal V_0:=\{(V)\}$.
    At level $k\in\{0,1,2\}$, consider a sequence $(V_1,\ldots,V_m)\in\mathcal V_k$.
    We leave its mixed-LPS factors unchanged and replace each factor of the form $V_i=\enc V_{I,B_k}$ with $|I|=k$ as follows.

    If $k<2$, add $2-k$ zero-effect dummy coordinates and encode them together with the coordinates in $I$.
    We can then apply~\cref{lem:encoding_factorability}, since exactly two coordinates are encoded.
    Projecting away the dummy coordinates yields a finite family of compositions whose factors are either mixed LPS or encodings $\enc V_{I\cup\{j\},D}$ for some original coordinate $j\notin I$ and $D\le B_{k+1}$.
    Every resulting mixed LPS has zero effect on the dummy coordinates, so this projection preserves its witnesses and their lengths.
    We enlarge $D$ to $B_{k+1}$ and replace $V_i$ by each of these compositions.

    Performing these replacements for every nonterminal factor and taking all possible choices gives $\mathcal V_{k+1}$.
    Thus every nonterminal factor in $\mathcal V_{k+1}$ has $k+1$ bounded original coordinates.
    After three levels, every factor is either a mixed LPS or an encoding with three bounded coordinates, and the recursion terminates.
    Each level-$k$ replacement introduces at most $2B_k+2$ factors.
    Hence every resulting sequence contains at most
    \[
        \prod_{k=0}^{2}(2B_k+2)
    \]
    factors, and all counter bounds and mixed-LPS parameters are at most $B_3$.
    Since the recursion has constant depth, these bounds are polynomial in $N$.

    Every reachable pair remains represented after each replacement by~\cref{lem:encoding_factorability}.
    Conversely, every resulting factor defines a subset of $\Reach_V$, and the distance inequality is preserved by the lemma and by projection of the dummy coordinates.
    Thus $\mathcal V_3$ gives a polynomial factorization.
\end{proof}

\paragraph{Effective Construction.}
\label{app:effective_factorization}
We consider candidate mixed LPSs with polynomially bounded parameters.
Each $\NN$-component may be an arbitrary LPS using transitions of $V$.
For each $\ZZ$-component, guess a $B$-encoding of $V$ with two encoded coordinates, possibly after adding at most two zero-effect dummy coordinates.
Choose a state $\Vec q\in[0,B]^2$ and an LPS $\Gamma$ from $\Vec q$ to itself in the corresponding $2B$-encoding.
Let $C$ be the transformation-area threshold supplied by~\cref{thm:generalized_transformation_area} for this encoding.
Restore the original transition effects and form the block
\[
    (\chi_{\Vec q}^C,\NN)(\Gamma,\ZZ)(\chi_{\Vec q}^C,\NN),
\]
projecting away any dummy coordinates afterward.

By~\cref{lem:mixed_lps_transformation_area}, every such block defines a subrelation of $\Reach_V$ and satisfies the distance inequality.
Thus every candidate mixed LPS is sound, while the preceding construction guarantees that these candidates include all factors needed for the factorization.
The encoded paths and cycles, checking gadgets, and polynomial parameter bounds can all be verified in polynomial time.

%% file: appendix/Appendix-D-Basic.tex
\section{Proofs for~\texorpdfstring{\cref{sec:basic_reachability_bounds}}{Section~\ref*{sec:basic_reachability_bounds}}}
\label{app:basic}

\subsection{Proof of~\texorpdfstring{\cref{thm:polynomial_witness}}{Theorem~\ref*{thm:polynomial_witness}}}
\label{app:polynomial_witness}

\lenFiveVAS*

\begin{proof}
	Let $(V,\Vec{s}, \Vec{t})$ be a reachable $5$-VAS and $M:=\Size(V,\Vec s,\Vec t)$.
	We obtain $B\le M^{\bigO(1)}$ and the two case distinctions by applying~\cref{thm:cfz} to $V$.

	For \cref{item:cfz_two_components}, assume that $V'' = (Q, T)$.
	Then $V'$ has $2|Q| \cdot (B + 1)^2$ states, at most $2|T| \cdot (B + 1)^2 + 1$ transitions, and norm at most $\norm{T}$, which implies that
	\begin{equation*}
		\Size(V') \le (2(B + 1)^2 + 1) \cdot \Size(V'').
	\end{equation*}
	\cref{thm:cfz} gives that $\MinLen_V(\Vec{s}, \Vec{t}) \le \MinLen_{V'}(s',t')$ and $B, \Size(V'') \le \Size(V, \Vec{s}, \Vec{t})^{\bigO(1)}$.
	Since
	\begin{equation*}
		B + \Size(V',s',t') \le \bigO((B+1)^2) \cdot \Size(V, \Vec{s}, \Vec{t}) ^{\bigO(1)} \le \Size(V, \Vec{s}, \Vec{t}) ^{\bigO(1)},
	\end{equation*}
	by~\cref{prop:two_component_reachability}, $\MinLen_V(\Vec{s}, \Vec{t}) \le \MinLen_{V'}(s',t') \le \Size(V, \Vec{s}, \Vec{t}) ^{\bigO(1)}$.

	In~\cref{item:cfz_encoding}, we fix the SCCs and bridges visited by a witnessing run $\pi$ and obtain $U$.
	By~\cref{prop:encoding_reachability},
	\begin{equation*}
		\MinLen_V(\Vec s,\Vec t)\le(B+M)^{\bigO(1)} \le M^{\bigO(1)}.
	\end{equation*}
	The exponents are uniform over all instances, so one constant suffices for both cases.
\end{proof}

\subsection{Proof of~\texorpdfstring{\cref{lem:proper_cone_trim}}{Lemma~\ref*{lem:proper_cone_trim}}}
\label{app:proper_cone_trim}

Geometrically, a proper cycle cone is contained in a half-space with a nonzero normal.
The existence and the description size of the separating plane are guaranteed by the following lemma, which is a part of~\cite[Lemma 4.5]{DBLP:conf/concur/ChenFZ26}.

\begin{lemma}
	\label{lem:separating_normal}
	Let $X \subseteq \mathbb{Z}^d$ be a finite set of vectors such that $\Cone(X) \ne \Span(X)$.
	There exists a nonzero vector $\Vec{n} \in \mathbb{Z}^d$ such that
	\begin{enumerate}
		\item \label{item:normal_pos} $\Inner{\Vec{n}, \Vec{x}} \ge 0$ for all $\Vec{x} \in X$;
		\item \label{item:normal_bound} $\norm{\Vec{n}} \le (d + 1)\cdot (r\norm{X})^r$, where $r := \dim(\Span(X))$.
	\end{enumerate}
\end{lemma}

\begin{proof}[Proof of~\cref{lem:proper_cone_trim}]
	Assume that $U=(Q,T)$, $s=p(\Vec x)$, and $t=q(\Vec y)$.
	If $\dimcyc(U)\le2$, take $\widetilde U=U$ with unchanged endpoints.
	Otherwise $\CycleSpace(U)=\QQ^3$.
	Let $X := \{\Eff{\theta} : \theta \text{ is a simple cycle in } U\}$.
	Thus $X$ is finite and $\Cone(X)\ne\Span(X)$.
	By~\cref{lem:separating_normal}, there exists a nonzero integer normal $\Vec h$ satisfying
	\begin{equation*}
		\Inner{\Vec h,\Eff{\theta}}\ge0\text{ for every simple cycle }\theta,
		\qquad \norm{\Vec h}\le 4(3|Q|\cdot \norm{T})^3.
	\end{equation*}
	Let $r(\Vec z)$ occur on a run from $p(\Vec x)$ to $q(\Vec y)$.
	We delete simple cycles separately from its prefix and suffix.
	Then each remaining path has length at most $|Q|-1$, while every deleted cycle has nonnegative inner product with $\Vec h$.
	Consequently, with $D:=(|Q|-1)\norm{T} \cdot \norm{\Vec h}_1$,
	\begin{equation*}
		\Inner{\Vec h,\Vec x}-D
		\le\Inner{\Vec h,\Vec z}
		\le\Inner{\Vec h,\Vec y}+D.
	\end{equation*}
	In particular, all these inner products lie in $[-H,H]$, where
	\begin{equation*}
		H:=(\norm{\Vec{x}} + \norm{\Vec{y}} + (|Q|-1)\norm{T})\cdot\norm{\Vec h}_1
		\le \Size(U,s,t)^{\bigO(1)}.
	\end{equation*}

	Define $\widetilde U$ with states $(r,b)\in Q \times[-H,H]$ and transitions
	\begin{equation*}
		\bigl((r,b),\Vec a,(r',b+\Inner{\Vec h,\Vec a})\bigr)
	\end{equation*}
	whenever $(r,\Vec a,r')\in T$ and both states belong to $Q \times [-H,H]$.
	Take $\widetilde s=(p,\Inner{\Vec h,\Vec x})(\Vec x)$ and $\widetilde t=(q,\Inner{\Vec h,\Vec y})(\Vec y)$.
	The bound above lifts every original run from $s$ to $t$ to $\widetilde U$, and projection gives the converse, preserving lengths.
	Every cycle of $\widetilde U$ returns its recorded inner product to the same value, so its effect lies in $\Vec h^\perp := \{\Vec x \in \QQ^{3} : \Inner{\Vec h, \Vec x} = 0\}$.
	Thus $\dimcyc(\widetilde U)\le2$.
	There are at most $|Q|\cdot (2H+1)$ states and $|T|\cdot (2H+1)$ transitions, with unchanged transition vectors and endpoint norms, proving the polynomial size bound.
\end{proof}

\subsection{Proof of~\texorpdfstring{\cref{lem:short_integer_vass}}{Lemma~\ref*{lem:short_integer_vass}}}
\label{app:short_integer_vass}

\shortIntVASS*

\begin{proof}
	We fix a witnessing path $\pi$ and mark the first occurrence of every visited state.
	The at most $|Q|$ marks divide $\pi$ into at most $|Q|$ segments, including the suffix after the last mark.
	Delete simple cycles from each segment until a simple path remains.
	The concatenation $\eta$ of the remaining paths has length at most $|Q|^2$ and still visits every state visited by $\pi$, since every marked occurrence remains as an endpoint of one segment.
	Each deleted simple cycle has length at most $|Q|$ and effect of norm at most $K:=|Q|\cdot \norm{T}$.

	There are at most $(2K+1)^d$ distinct effects among the deleted cycles.
	Choose one representative cycle for each such effect, and let $\Vec A$ be the matrix of these effects.
	The original multiplicities give a nonnegative integer solution of $\Vec A\Vec z=\Vec b$, where $\Vec b:=\Vec y-\Vec x-\Eff(\eta)$.
	By~\cref{thm:pottier}, a minimal solution of this system satisfies
	\begin{equation*}
		\norm{\Vec z}_1\le
		\bigl(1+(2K+1)^dK+\norm{\Vec{x}}+\norm{\Vec{y}}+|Q|^2\norm{T}\bigr)^d \le (\Size(U) + \norm{\Vec{x}} + \norm{\Vec{y}})^{\bigO(d^2)}.
	\end{equation*}
	Here the augmented matrix $\begin{pmatrix} \Vec A & \Vec b \end{pmatrix}$ has rank at most $d$.
	If $\Vec b=\Vec0$, the zero solution suffices.

	Since every state visited by $\pi$ still appears in $\eta$, each representative cycle can be placed at a state visited by $\eta$.
	Using the multiplicities given by $\Vec z$, we obtain a $\ZZ$-run of total effect $\Vec y-\Vec x$.
	Its length is at most $|\eta|+|Q|\cdot \norm{\Vec z}_1\le (\Size(U) + \norm{\Vec{x}} + \norm{\Vec{y}})^{\PolyF{d}}$.
\end{proof}

\subsection{Proof of~\texorpdfstring{\cref{lem:joining_coverable_hybrids}}{Lemma~\ref*{lem:joining_coverable_hybrids}}}
\label{app:joining_coverable_hybrids}

\joinCoverableHybrid*

\begin{proof}
	Write $H_i=\Vec b_i+P_i^{\langle\Vec c_i,\cdot\rangle\le K_i}$ for $i=1,2$, write $\Vec c_i(\Vec a)$ for the cost of $\Vec a\in P_i$, and fix representations $\Vec x=\Vec b_1+\sum_{\Vec a\in P_1}\lambda^*_{\Vec a}\cdot\Vec a$ and $\Vec y=\Vec b_2+\sum_{\Vec a\in P_2}\mu^*_{\Vec a}\cdot\Vec a$ satisfying $K_1$ and $K_2$, respectively.
	Consider the augmented VASS $U[P_1,P_2]$.
	The fixed representations and the hypothesis give $\Vec p(\Vec b_1)\xrightarrow{*}\Vec q(\Vec b_2)$ in $U[P_1,P_2]$.
	Moreover, the two bases are respectively forward and backward $C$-coverable in $U[P_1,P_2]$.
	By \cref{cor:threshold_closure}, $C$ remains a transformation-area threshold after adjoining these periods.
	Hence \cref{cor:coverable_case} gives a run $\rho$ in the augmented VAS from $(\Vec p,\Vec b_1)$ to $(\Vec q,\Vec b_2)$ of length $L\le M^{\bigO(1)}$, where $M:=B+C+\Size(V)+\hnorm(H_1)+\hnorm(H_2)$.

	For $\Vec a\in P_1$, let $\lambda_{\Vec a}$ be the number of occurrences of $+\Vec a$ in $\rho$, and for $\Vec a\in P_2$, let $\mu_{\Vec a}$ be the number of occurrences of $-\Vec a$.
	Let
	\begin{equation*}
		\Vec x':=\Vec b_1 + \sum_{\Vec a\in P_1}\lambda_{\Vec a}\cdot\Vec a,
		\qquad
		\Vec y':=\Vec b_2 + \sum_{\Vec a\in P_2}\mu_{\Vec a}\cdot\Vec a.
	\end{equation*}
	Delete all added transitions from $\rho$ and start the remaining path at $\Vec x'$.
	At every position of the resulting path, its vector is the corresponding vector of $\rho$ plus the sum of the unused positive periods and the absolute effects of the negative periods already removed.
	Hence the resulting path is a run in $V$ from $(\Vec p,\Vec x')$ to $(\Vec q,\Vec y')$, of length at most $L$.

	It remains to verify the cost bounds.
	Since $\rho$ has length at most $L$, the total cost of the added period occurrences on side $i$ is at most $\norm{\Vec c_i}L$.
	Hence, if $K_i\ge\norm{\Vec c_i}L$, the corresponding vector already belongs to $H_i$.

	If $K_i<\norm{\Vec c_i}L$, apply~\cref{lem:restricted_period_absorption} to the original member $\Vec x$ or $\Vec y$ on side $i$.
	This gives a subset $\widehat {H}_i$ of $H_i$ containing that member, with only free periods and base norm at most $\norm{\Vec b_i}+K_i\norm{P_i}\le M^{O(1)}$.
	Thus the original connected pair $(\Vec p(\Vec x), \Vec q(\Vec y))$ remains available.
	Since the base increases and the free periods are unchanged, forward or backward $C$-coverability is preserved by monotonicity.
	We repeat the construction with the new hybrid sets.
	Each side can be absorbed at most once, so at most two repetitions are needed.
	The final vectors belong to the original hybrid sets $H_1,H_2$, have norm at most $M^{\bigO(1)}$, and are connected in $V$ by a run of length at most $M^{\bigO(1)}$.
\end{proof}

%% file: appendix/Appendix-E-Induction.tex
\section{Proofs for~\texorpdfstring{\cref{sec:semi_coverable}}{Section~\ref*{sec:semi_coverable}}}
\label{app:semi_coverable_proofs}

\subsection{Proof of~\texorpdfstring{\cref{lem:bounded_prefix}}{Lemma~\ref*{lem:bounded_prefix}}}
\label{app:bounded_prefix}

\boundedPrefix*

\begin{proof}
	Assume that $U=(Q,T)$ with $n:=|Q|$.
	Let $D:=D_0+\norm{T}$, where
	\begin{equation*}
        L_c:=\max\{2,n,\norm{T},C\}^{c_{\mathrm{cov}}},\qquad
        D_0:=\max\{(C+n\norm T+\Size(U))^{c_{\mathrm{ext}}},C+L_c\norm{T}+1\}.
	\end{equation*}
	Here $c_{\mathrm{cov}},c_{\mathrm{ext}}\ge1$ are fixed integer constants from~\cref{lem:rackoff_cover,lem:coordinate_extraction} for $d=3$.
	Write the witness as $\pi=\xi_1u_1\xi_2\cdots u_{k-1}\xi_k$, where $\xi_j$ stays in $U_j$.
	Some coordinate remains below $D_0$ along $\xi_1$.
	Otherwise,~\cref{lem:coordinate_extraction} gives a run in $U_1$ from $p(\Vec x)$ to a configuration $q'(\Vec w)$ with $\Vec w\ge (C+n\norm{T})\cdot\Vec1$.
	A simple path from $q'$ back to $p$ has length at most $n-1$ and decreases each coordinate by at most $(n-1)\norm{T}$.
	It is fireable at $\Vec w$ and returns with at least $C\cdot\Vec 1$, contradicting the hypothesis.

	Choose $i\in\BDCounters(\xi_1,D_0)$ and put
	\begin{equation*}
		b:=\min\{j\in\{2,\ldots,k\}:i\notin\BDCounters(\xi_j,D_0)\},
	\end{equation*}
	with $b:=k+1$ if the set is empty.
	If $b\le k$, set $\pi_0:=\xi_1u_1\cdots\xi_{b-1}u_{b-1}$.
	Otherwise, set $\pi_0:=\pi$.
	Every included component segment keeps coordinate $i$ below $D_0$, including $\xi_1$.
	Across each included bridge, its value satisfies
	\begin{equation*}
		\Vec w'(i)=\Vec w(i)+\Eff(u_j)(i)
		<D_0+\norm T=D.
	\end{equation*}
	This proves that $i\in\BDCounters(\pi_0,D)$.

	Suppose $b\le k$, and let $r(\Vec z)$ be the entry of $U_b$.
	There is a prefix $\eta$ of $\xi_b$ reaching $q'(\Vec w)$ with $\Vec w(i)\ge D_0$.
	Choose a simple path $\zeta$ from $q'$ to $r$ in $U_b$.
	Coordinate $i$ is nonnegative along $\eta$, and every prefix $\zeta'$ of $\zeta$ satisfies
	\begin{equation*}
		\Vec w(i)+\Eff(\zeta')(i)
		\ge D_0-(n-1)\norm{T}
		\ge C.
	\end{equation*}
	Projecting $\rho:=\eta\zeta$ onto coordinate $i$ gives a run from $r(\Vec z(i))$ to $r(v)$ with $v\ge C$.
	By~\cref{lem:rackoff_cover}, there is such a covering run of length less than $L_c$.
	Now fix any $\Vec w$ satisfying the conditions in the second outcome.
	Lift the shortened run to a $\ZZ$-run $\rho'$ in $U_b$ from $r(\Vec w)$ to $r(\Vec v)$.
	Since $\Vec w(i)\ge\Vec z(i)$, coordinate $i$ is nonnegative along $\rho'$ and $\Vec v(i)\ge C$.
	For every $j\ne i$ and every prefix $\eta'$ of $\rho'$, we have $\Vec w(j)+\Eff(\eta')(j)\ge D-L_c\norm{T}>C$.
	Therefore, $\rho'$ is a run and $r(\Vec w)$ is forward $C$-coverable in $U_b$.

	Finally, encoding coordinate $i$ in $\{0,\ldots,D-1\}$ gives the $2$-VASS $\enc{U}_{\{i\},D-1}$ with
	\begin{equation*}
		\Size(\enc{U}_{\{i\},D-1})\le D\Size(U).
	\end{equation*}
	Projection and lifting preserve run lengths and represent exactly the runs keeping coordinate $i$ below $D$.
	The construction remains valid with any larger $D$ if $D_0$ is replaced by $D-\norm T$.
\end{proof}

\subsection{Proof of~\texorpdfstring{\cref{lem:alternating_boundedness}}{Lemma~\ref*{lem:alternating_boundedness}}}
\label{app:alternating_boundedness}

\alternatingBoundedness*

\begin{proof}
	Write $U=(Q,T)$ and put $I:=[3]\setminus\Supp(\Vec a_0)$ and $n:=|Q|$.
	Choose $D$ by the formula in the proof of~\cref{lem:bounded_prefix}, using the parameters of $U[\{\Vec a_0\}]$ and the threshold $C$.
	These parameters bound those of every suffix augmented with the same loop at its source, so the same $D$ works throughout and satisfies~\eqref{eq:alternating_threshold_bound}.
	Let $r_*(\Vec z_*)$ be the first component entry on $\pi$ that is forward $C$-coverable after adjoining this loop there.
	If no such entry exists, take $r_*(\Vec z_*):=t$.
	Let $\pi_*$ be the original prefix ending at this position.
	If $\pi_*=\varepsilon$, its empty partition proves the claim.
	In particular, this handles $I=\varnothing$, since repetitions of the added loop already give forward $C$-coverability at the source.

	We divide $\pi_*$ into bounded prefixes by repeated applications of~\cref{lem:bounded_prefix}.
	Suppose the unprocessed suffix starts at an original component entry $p_j(\Vec x_j)$ before $r_*(\Vec z_*)$.
	By its choice, this entry is not forward $C$-coverable in its component augmented with a self-loop of effect $\Vec a_0$.
	A run returning to this entry cannot cross a bridge, so the entry is also not forward $C$-coverable in the augmented suffix.
	Execute the loop $D$ times and replay the unprocessed part of $\pi_*$ translated by $D\cdot\Vec a_0$.
	By monotonicity this is a run from $p_j(\Vec x_j)$, including the initial loops, and every coordinate in $\Supp(\Vec a_0)$ reaches at least $D$ within its first component.
	By the construction in the proof of~\cref{lem:bounded_prefix}, the extracted prefix contains all the initial loop occurrences.
	Its bounded coordinate therefore lies in $I$.
	Removing these loops and undoing the translation by $D\cdot\Vec a_0$ gives a prefix $\xi$ of the original run with
	\[
		\BDCounters(\xi,D)\cap I\ne\emptyset.
	\]
	If $\xi$ stops at a later entry before $r_*(\Vec z_*)$, apply the same construction to the remaining suffix.
	Continuing in this way gives a partition
	\begin{equation*}
		\pi_*=\xi_1\xi_2\cdots\xi_m
	\end{equation*}
	such that
	\begin{equation}\label{eq:pump_unsupported_bounded}
		\BDCounters(\xi_j,D)\cap I\ne\varnothing
		\qquad\text{for every }j\in[m].
	\end{equation}
	The same threshold $D$ is used throughout, and each nonterminal application passes at least one bridge, so $m\le k$.
	Repeatedly merge adjacent segments that share a coordinate in $I$ remaining below $D$.
	Since
	\[
		\BDCounters(\alpha\beta,D)\cap I
		=
		\bigl(\BDCounters(\alpha,D)\cap I\bigr)
		\cap
		\bigl(\BDCounters(\beta,D)\cap I\bigr),
	\]
	each merged segment still satisfies~\eqref{eq:pump_unsupported_bounded}.
	When no further merge is possible, the resulting partition satisfies~\eqref{eq:d_alternating}.
\end{proof}

\subsection{Proof of~\texorpdfstring{\cref{lem:hill_cutting}}{Lemma~\ref*{lem:hill_cutting}}}
\label{app:hill_cutting}

\hillCutting*

\begin{proof}
	By taking a prefix, we assume w.l.o.g. that $q(\Vec y)$ is the first configuration where coordinate $i$ has value at least $D'$.
	Let $r_0(\Vec u_0),\ldots,r_s(\Vec u_s)$ be the successive strict records of coordinate $i$, including $r_0(\Vec u_0):=p(\Vec x)$ and $r_s(\Vec u_s):=q(\Vec y)$.
	Each record exceeds the preceding one by at most $\norm T$, so
	\[
		s\norm T\ge\Vec u_s(i)-\Vec x(i)\ge D|Q|\norm T+1,
		\qquad s\ge D|Q|+1.
	\]
	Here $\norm T\ge1$, since coordinate $i$ increases along $\pi$.
	The first $s$ records precede the crossing, and their state and $j$-coordinate pairs take at most $D|Q|$ values.
	Hence there are $0\le a<b<s$ with
	\[
		r_a=r_b,
		\qquad \Vec u_a(j)=\Vec u_b(j),
		\qquad \Vec u_a(i)<\Vec u_b(i).
	\]
	The intervening sub-run keeps coordinate $i$ below $D'$ and coordinate $j$ below $D$.
	Encode these two coordinates into states, obtaining a $1$-VASS with at most $|Q|DD'$ states and transition norm at most $\norm T$.
	The sub-run witnesses coverability of $\widehat r_b(0)$ from $\widehat r_a(\Vec u_a(k))$, where $\{k\}=[3]\setminus\{i,j\}$ and $\widehat r_c:=(r_c,\Vec u_c(i),\Vec u_c(j))$ for $c\in\{a,b\}$.
	By~\cref{lem:rackoff_cover}, there is such a covering run of length
	\[
		\max\{2,|Q|DD',\norm T\}^{\bigO(1)}
		\le(D+\Size(U)+\Vec x(i))^{\bigO(1)}.
	\]
	Lifting this run to $U$ gives an enabled cycle $\theta$ at $r_a(\Vec u_a)$, since $r_a=r_b$.
	Its encoded endpoints guarantee $\Eff(\theta)(i)=\Vec u_b(i)-\Vec u_a(i)>0$ and $\Eff(\theta)(j)=0$, as required.
\end{proof}

\subsection{Proof of~\texorpdfstring{\cref{lem:pump_extraction}}{Lemma~\ref*{lem:pump_extraction}}}
\label{app:pump_extraction}

\pumpExtraction*

\begin{proof}
	Let $M:=C+\Size(U,p(\Vec x))$.
    Pumpability gives a coordinate $i$ for which $p(\Vec x+\Vec e_i)$ is coverable from $p(\Vec x)$.
	By~\cref{lem:rackoff_cover}, there is a cycle $\theta_0$ fireable from $p(\Vec x)$ such that
    \[
        \Vec a_0:=\Eff(\theta_0)\in\NN^3\setminus\{\Vec0\},
        \qquad
        |\theta_0|,\norm{\Vec a_0}\le M^{O(1)}.
    \]
    Permute coordinates so that $\Vec a_0(3)>0$, and put $I:=[3]\setminus\Supp(\Vec a_0)\subseteq\{1,2\}$.

	Apply~\cref{lem:alternating_boundedness} to obtain a partition $\pi=\pi_*\sigma_*$, where $\pi_*$ ends at $r_*(\Vec z_*)$ and is $D$-alternating on $I$ for some $D\le M^{O(1)}$.
    If $\sigma_*$ is nonempty, its starting configuration is forward $C$-coverable in its component after adjoining
    $\Vec a_0$ as a self-loop.
    Since $|I|\le2$, the prefix $\pi_*$ is also $D$-alternating on $\{1,2\}$.
    Apply~\cref{lem:hill_cuts} to $\pi_*$ and enlarge $D_1$ if necessary so that $D_1>\max\{D,\norm{\Vec x}\}$.
	Then $D_1,L\le M^{\bigO(1)}$.
	We distinguish the two outcomes of~\cref{lem:hill_cuts}.

	Suppose first that $j\in\BDCounters(\pi_*,D_1)$ for some $j\in[3]$.
	Encode this coordinate into states, obtaining the $2$-VASS $\widetilde U:=\enc U_{\{j\},D_1-1}$ of size at most $D_1\Size(U)\le M^{\bigO(1)}$.
	Apply~\cref{thm:hybrid_description} with the encoded singleton source, and choose a hybrid set at state $(r_*,\Vec z_*(j))$ containing the projection of $\Vec z_*$.
	Inserting $\Vec z_*(j)$ into its base and zeros into its periods gives a hybrid set $G$ containing $\Vec z_*$, with $\hnorm(G)\le M^{\bigO(1)}$.
	For each $\Vec h\in G$, apply~\cref{lem:hybrid_source} in $\widetilde U$ with the same singleton source and the encoded $r_*(\Vec h)$ as target.
	The resulting short run lifts to $U$.
	Thus every $\Vec h\in G$ has a run from $p(\Vec x)$ to $r_*(\Vec h)$ in $U$ of length at most $(M+\norm{\Vec h})^{\bigO(1)}$.
	If $r_*(\Vec z_*)=t$, take $\tau:=\pi$, $\sigma:=\varepsilon$, and $H:=G$, which proves the lemma in this case.
	Otherwise, $r_*(\Vec z_*)$ is a component entry that is forward $C$-coverable after adjoining the $\Vec a_0$-loop there.
	Adjoin $\Vec a_0$ as a free period and write
	\begin{equation*}
		G^+:=G+\Vec a_0^\ast=\Vec b+P^{\langle\Vec c,\cdot\rangle\le K},
		\qquad \Vec z_*\in G^+.
	\end{equation*}
	For every $\Vec h\in G$ and $k\in\NN$, reach $\Vec h+k\cdot\Vec a_0\in G^+$ by pumping $\theta_0^k$ at the source and translating a witness reaching $\Vec h$.
	Since $k\le\norm{\Vec h+k\cdot\Vec a_0}$, this again yields polynomial-length witnesses.

	\begin{claim}\label{clm:pump_coverable_refinement}
		There is a hybrid set $H$ with $\Vec z_*\in H\subseteq G^+$ and $\hnorm(H)\le M^{\bigO(1)}$ that is forward $C$-coverable in the component containing $r_*$.
	\end{claim}

	\begin{claimproof}
	By~\cref{lem:rackoff_cover}, choose a run $\rho_{\mathrm{cov}}$ of length $M^{\bigO(1)}$ from $r_*(\Vec z_*)$ covering $r_*(C\cdot\Vec1)$ in this component augmented with the $\Vec a_0$-loop.
	Let $\Vec v:=\max\{\Drop(\rho_{\mathrm{cov}}),C\cdot\Vec1-\Eff(\rho_{\mathrm{cov}}),\Vec0\}$, so $\Vec z_*\ge\Vec v$ and $\norm{\Vec v}\le M^{\bigO(1)}$.
	Choose a representation $\Vec z_*=\Vec b+P\Vec\eta$ in $G^+$ with $\langle\Vec c,\Vec\eta\rangle\le K$.
	Starting from $\Vec b$, add unused period occurrences from this representation that increase a deficient coordinate until the resulting base dominates $\Vec v$.
	Each addition decreases the total integer deficit, giving $\Vec0\le\Vec\mu\le\Vec\eta$ with $\norm{\Vec\mu}_1\le3\norm{\Vec v}$ and $\Vec b':=\Vec b+P\Vec\mu\ge\Vec v$.
	Set
	\begin{equation*}
		H:=\Vec b'+P^{\langle\Vec c,\cdot\rangle\le K-\langle\Vec c,\Vec\mu\rangle}.
	\end{equation*}
	Then $\Vec z_*\in H\subseteq G^+$ and $\hnorm(H)\le M^{\bigO(1)}$.
	Since $\Vec a_0$ is a free period of $H$ and $\Vec b'\ge\Vec v$, the same run $\rho_{\mathrm{cov}}$ witnesses forward $C$-coverability of $H$.
	\end{claimproof}

	Take $H$ from~\cref{clm:pump_coverable_refinement}, $\tau:=\pi_*$, and let $\sigma$ be the original remaining suffix.
	The witnesses for $G^+$ give the required length bound for every member of $H$.

	Otherwise, $\BDCounters(\pi_*,D_1)=\emptyset$.
	Since $D$-alternation on $I$ with $|I|\le1$ implies that some coordinate stays below $D_1$ throughout $\pi_*$, we must have $I=\{1,2\}$ and $\Vec a_0=(0,0,a)$ for some $a\ge1$.
	\cref{lem:hill_cuts} supplies two ordered cycles $\theta_1,\theta_2$ and a partition $\pi_*=\rho_1\rho_2\rho$.
	Apply~\cref{lem:two_cycle_capture} with common cycle-length bound $\max\{|\theta_0|,L\}\le M^{\bigO(1)}$.
	The resulting $U^{(q_1,q_2)}$ and period $\Vec a_+\ge\Vec1$ satisfy $\Size(U^{(q_1,q_2)}),\norm{\Vec a_+}\le M^{\bigO(1)}$.
	Its first property gives a run to $f(\Vec z_2)$ on which coordinate $2$ remains below $D_1$.
	Encode this coordinate into states, obtaining $\widetilde U^{(q_1,q_2)}:=\enc{U^{(q_1,q_2)}}_{\{2\},D_1-1}$.
	Its size satisfies
	\[
		\Size(\widetilde U^{(q_1,q_2)})\le D_1\Size(U^{(q_1,q_2)})\le M^{\bigO(1)}.
	\]
	As above, apply~\cref{thm:hybrid_description} with the encoded singleton source and choose a hybrid set at state $(f,\Vec z_2(2))$ containing the projection of $\Vec z_2$.
	Recovering the fixed second coordinate gives a hybrid set $G$ with
	\begin{equation*}
		\Vec z_2\in G\subseteq\Reach^f_{U^{(q_1,q_2)}}(p^{(0)}(\Vec x)),
		\qquad \hnorm(G)\le M^{\bigO(1)}.
	\end{equation*}
	For each $\Vec h\in G$, \cref{lem:hybrid_source} in $\widetilde U^{(q_1,q_2)}$, with the same singleton source and the encoded $f(\Vec h)$ as target, gives a run to $f(\Vec h)$ in $U^{(q_1,q_2)}$ of length at most $(M+\norm{\Vec h})^{\bigO(1)}$.

	Set $H:=G+\Vec a_+^\ast$, assigning the added period cost zero.
	Then $\Vec z_2\in H$ and $\hnorm(H)\le M^{\bigO(1)}$.
	For $\Vec v=\Vec h+k\cdot\Vec a_+\in H$, we have $\Vec h\le\Vec v$ and $k\le\norm{\Vec v}$, so~\cref{lem:two_cycle_capture} gives
	\begin{equation*}
		\MinLen_U\bigl(p(\Vec x),q_2(\Vec v)\bigr)\le(M+\norm{\Vec h})^{\bigO(1)}+kM^{\bigO(1)}\le(M+\norm{\Vec v})^{\bigO(1)}.
	\end{equation*}
	
	Take $\tau:=\rho_1\rho_2$ and $\sigma:=\rho\sigma_*$, so that $\pi=\tau\sigma$ with the boundary configuration $q_2(\Vec z_2)$.
	Since $\Vec a_+\ge\Vec1$ is a free period of $H$, its base is forward $C$-coverable in the component containing $q_2$ after adjoining the corresponding self-loop.
	This completes the proof.
\end{proof}

\subsection{Proof of~\texorpdfstring{\cref{lem:free_covered}}{Lemma~\ref*{lem:free_covered}}}
\label{app:free_covered}

\freeCovered*

\begin{proof}
	Write $H=\Vec b+P^{\langle\Vec c,\cdot\rangle\le K}$ and let $M:=C+B+\Size(V,\Vec t)+\hnorm(H)$.
	Consider $U[P]$ with specified source state $\Vec r$.
	Let $V^+$ be the original $5$-VAS with the lifted periods $(0,0,\Vec a)$, $\Vec a\in P$, adjoined as transitions.
	Note that $U[P]$ is a sequential sub-VASS of $\enc{V^+}_B$.
	Clearly, $\Vec r(\Vec b)\xrightarrow{*}\Vec q(\Vec y)$ in $U[P]$, since some member of $H$ reaches this target in $U$.
	A nonzero free period of $H$ makes $\Vec r(\Vec b)$ pumpable in $U[P]$.
	The original backward $C$-covering run of the target is still available.
	By~\cref{cor:threshold_closure}, $C$ remains a transformation-area threshold after adjoining the periods to $V$.
	Since $\Size(U)\le(B+1)^2\Size(V)$ and $|P|\le(\norm P+1)^3$, both $U[P]$ and $V^+$ have size polynomial in $M$.
	By~\cref{lem:pumpable_coverable_target}, there is a run in $V^+$
	\begin{equation*}
		(\Vec r,\Vec b)\xrightarrow{\rho}\Vec t,
		\qquad |\rho|\le L:=F_0(M)
	\end{equation*}
	for a fixed polynomial $F_0$.

	Remove the added periods as in the proof of~\cref{lem:joining_coverable_hybrids}, with singleton target $\{\Vec y\}$.
	This gives a run $\rho'$ in $V$ with
	\begin{equation}\label{eq:free_deleted_run_bound}
		(\Vec r,\Vec h)\xrightarrow{\rho'}\Vec t,
		\qquad |\rho'|\le L,
		\qquad \norm{\Vec h}\le\norm{\Vec b}+L\norm P\le M^{\bigO(1)}.
	\end{equation}
	If $K\ge\norm{\Vec c}L$, the same cost estimate gives $\Vec h\in H$.
	If $K<\norm{\Vec c}L$, choose $\Vec z\in H$ such that $\Vec r(\Vec z)\xrightarrow{*}\Vec q(\Vec y)$ in $U$, and let $P_0:=\{\Vec a\in P:\Vec c(\Vec a)=0\}$.
	Applying~\cref{lem:restricted_period_absorption} to a representation of $\Vec z$ in $H$ gives
	\begin{equation*}
		\Vec z\in\widehat H:=\widehat{\Vec b}+P_0^*\subseteq H,
		\qquad \norm{\widehat{\Vec b}}\le M^{\bigO(1)}.
	\end{equation*}
	Since $\Vec z\in\widehat H$ and $P_0$ contains a nonzero period, $\Vec r(\widehat{\Vec b})$ reaches $\Vec q(\Vec y)$ and is pumpable in $U[P_0]$.
	The target remains backward $C$-coverable, so reapply~\cref{lem:pumpable_coverable_target} and remove the added periods as in~\eqref{eq:free_deleted_run_bound}.
	All remaining periods are free, so the selected member belongs to $\widehat H\subseteq H$, with its norm and distance to $\Vec t$ bounded by $M^{\bigO(1)}$.
\end{proof}

\subsection{Proof of~\texorpdfstring{\cref{clm:translated_period_removal}}{Claim~\ref*{clm:translated_period_removal}}}
\label{app:translated_period_removal}

\translatedPeriodRemoval*

\begin{claimproof}
	Let $\lambda_{\Vec d}$ and $\mu_{\Vec d'}$ count the added positive and negative period occurrences in $\rho$, and define
	\begin{equation*}
		\Vec\xi:=\sum_{\Vec d\in P}\lambda_{\Vec d}\cdot\Vec d,
		\qquad \Vec\xi':=\sum_{\Vec d'\in P'}\mu_{\Vec d'}\cdot\Vec d'.
	\end{equation*}
	Set $\Vec h:=\Vec b_r+\Vec\xi$ and $\Vec h':=\Vec b'_\ell+\Vec\xi'$.
	Delete the added occurrences and start the remaining transitions at $(\Vec p,\Vec h)$.
	At every position of the resulting path, the vector is increased by the unused positive periods and the already removed negative periods.
	Thus the remaining transitions form a run in $V$ from $(\Vec p,\Vec h)$ to $(\Vec q,\Vec h')$ of length at most $L$.
	Moreover,
	\begin{equation*}
		\norm{\Vec h},\norm{\Vec h'}
		\le\max\{\norm{\Vec b},\norm{\Vec b'}\}
		+\kappa_0\max\{\norm{\Vec\delta},\norm{\Vec\delta'}\}
		+L\max\{\norm P,\norm{P'}\}\le M^{\bigO(1)}.
	\end{equation*}
	Both selected vectors meet their cost bounds if
	\begin{equation*}
		K\ge\norm{\Vec c}L,
		\qquad K'\ge\norm{\Vec c'}L.
	\end{equation*}
	If either cost bound fails, we absorb the restricted periods on that side.
	We describe the source side, where $K<\norm{\Vec c}L$; the target side is symmetric.
	Apply~\cref{lem:restricted_period_absorption} to the fixed representation of $\Vec z$ to obtain
	\begin{equation*}
		\widehat H_s:=\widehat{\Vec b}+s\cdot\Vec\delta+P_0^*\subseteq H_s,
		\qquad \Vec z\in\widehat H_0,
		\qquad \norm{\widehat{\Vec b}}\le\norm{\Vec b}+K\norm P,
	\end{equation*}
	where $0\le s\le\kappa$.
	The original connected members remain available, while monotonicity preserves the covering runs and $\kappa$.
	We repeat the augmented-run construction with the reduced period sets.
	The bound in~\eqref{eq:bilateral_joining_range} still applies.
	Each failure eliminates all restricted periods on at least one side, so at most two repetitions suffice to obtain members of the original families with norms and reachability distance at most $M^{\bigO(1)}$.
\end{claimproof}

\subsection{Proof of~\texorpdfstring{\cref{lem:two_counter_tools}}{Lemma~\ref*{lem:two_counter_tools}}}
\label{app:two_counter_tools}

\twoCounterTools*

\begin{proof}
	The LPS representation is~\cite[Theorem 3.1]{DBLP:journals/jacm/BlondinEFGHLMT21}, and the length bound is~\cite[Theorem 3.2]{DBLP:journals/jacm/BlondinEFGHLMT21}.
	The hybrid-source description is~\cite[Lemma C.11]{DBLP:journals/corr/abs-2607-14983}.
	The bound $\hnorm(G)\le(S+\hnorm(A))^{\bigO(1)}$ is already established in the cited proof.
	It remains to show that $\norm{P_G}$, $\norm{\Vec c_G}$, and $|P_G|$ have polynomial bounds independent of the source base $\Vec b_A$ and the cost bound $K_A$.
	Put $R:=S+\norm{P_A}+\norm{\Vec c_A}+1$.

	By~\cite[Lemma C.1 and Theorem C.2]{DBLP:journals/corr/abs-2607-14983}, the original construction consists of a short prefix, a zigzag middle, and a short suffix.
	The first and last stages involve LPSs with constantly many cycles.
	Their output periods and costs are determined by homogeneous solutions of linear Diophantine systems, whose coefficients are independent of $\norm{\Vec{b}_A}$ and $K_A$.
	Since these systems have a constant number of equations, \cref{thm:pottier} gives
	\[
		\norm{P_{\mathrm{out}}},\norm{\Vec c_{\mathrm{out}}}
		\le (S+\norm{P_{\mathrm{in}}}+\norm{\Vec c_{\mathrm{in}}}+1)^{\bigO(1)}.
	\]

	By~\cite[Lemma C.4]{DBLP:journals/corr/abs-2607-14983}, each one-turn return produces a finite family of arithmetic sets, represented by hybrid sets with at most one period of cost zero or one and with period norms satisfying
	\[
		D_{i+1}\le\max\{D_i,Q(S)\},
	\]
	where $D_i$ bounds the period norms and $Q$ is a fixed polynomial.
	Consequently, the period bound does not accumulate over the polynomially many returns.

	Applying the constant-cycle bound to the final short suffix yields $\norm{P_G},\norm{\Vec c_G}\le R^{\bigO(1)}.$
	Finally, since $P_G\subseteq\NN^2\setminus\{\Vec0\}$, $|P_G|\le(\norm{P_G}+1)^2\le R^{\bigO(1)}$.
	This proves the required bounds.
\end{proof}

\subsection{Proofs of the Claims in~\texorpdfstring{\cref{lem:large_prefix_dichotomy}}{Lemma~\ref*{lem:large_prefix_dichotomy}}}
\label{app:mixed_dichotomy_claims}

Whenever we apply~\cref{lem:two_counter_tools} to an LPS $\Lambda$, we use the $2$-VASS induced by $\Lambda$, realizing its paths and cycles with fresh states.
This VASS has polynomial size in $|\Lambda|+\norm\Lambda$, and it has the same reachability relation as $\Lambda$ between the designated source and target.
Thus the resulting hybrid sets represent only target vectors of runs captured by the corresponding LPS.

\mixedFreePeriod*

\begin{claimproof}
	If a remaining cycle $\beta$ has nonzero nonnegative effect $\Vec e$, require one copy by replacing its starred expression with $\beta\beta^*$.
	This restricted scheme still contains the actual witness and has polynomial size.
	Choose a hybrid set containing its actual target by~\cref{lem:two_counter_tools} and adjoin $\Vec e$ as a free period.
	Every represented run executes an enabled copy of $\beta$.
	Further copies and the original suffix are enabled by monotonicity.
	The extra copies increase the exponent of the same starred block, so every member lies in $\mathcal S_\Lambda$, including when the original cycle has a nonzero internal drop.
	If $G$ is the selected hybrid set, this construction is
	\begin{equation*}
		H:=G+\Vec e^{\ast},
		\qquad \Vec z\in G\subseteq H,
		\qquad \Vec e\ge\Vec0,\quad \Vec e\ne\Vec0.
	\end{equation*}
	The added period has norm at most $E$ and cost zero.
	Thus the first outcome has period norm and base norm polynomial in $D+\norm{\Vec x}$.
\end{claimproof}

\mixedLongCycle*

\begin{claimproof}
	After deleting zero-effect cycles, every cycle with positive effect on coordinate $k$ has negative effect on coordinate $j$.
	The $R$-large hypothesis gives
	\begin{equation*}
		\Vec z(k)>\norm{\Vec x}+R
		\ge\Vec x(k)+\kappa E.
	\end{equation*}

	If every cycle increasing $k$ occurred fewer than $\kappa$ times, the total $k$-increase would be at most $\kappa E$, counting all fixed paths.
	Choose the last block with at least $\kappa$ copies of a cycle $\beta$ increasing $k$, and write
	\begin{equation*}
		\beta^\ell=\beta^{\ell-\kappa}\beta^\kappa,
		\qquad \Eff(\beta)=-u\cdot\Vec e_j+v\cdot\Vec e_k,
		\quad \ell\ge\kappa,\quad u,v\ge1.
	\end{equation*}
	From any position after this block, the remaining $k$-effect is at most $\kappa(E-|\beta|\norm{T})$: later positive-effect cycles occur fewer than $\kappa$ times, complete nonpositive-effect cycles contribute at most zero, and fixed paths and a possible partial cycle use only the scheme outside the selected block.
	Consequently every vector $\Vec w$ on that suffix satisfies
	\begin{equation}\label{eq:phase_suffix_reserve}
		\Vec w(k)\ge\Vec z(k)-\kappa(E-|\beta|\norm{T})>\kappa|\beta|\norm{T}\ge\kappa v.
	\end{equation}

	Let $\Vec a$ be the vector just before the last $\kappa$ copies of $\beta$.
	The prefix scheme $\Lambda_-$ ending there includes the preceding starred copies of $\beta$, and $\Lambda_+$ is the unchanged suffix scheme after this block.
	Both are obtained from $\Lambda$ without unrolling the remaining repetitions, so their sizes are polynomial in $D$ and independent of $\kappa$.
	The actual prefix and suffix give $\eta_-$ and $\eta_+$, and $u,v\le|\beta|\norm{T}\le E$.
\end{claimproof}

\mixedTranslatedFamily*

\begin{claimproof}
	Let $\Vec d:=\Drop(\beta)$.
	Apply~\cref{lem:two_counter_tools} to choose a hybrid set $A=\Vec b_A+P_A^{\langle\Vec c_A,\cdot\rangle\le K_A}$ containing $\Vec a$ and consisting of vectors reachable through $\Lambda_-$.
	The validity of the actual $\beta^\kappa$ gives the guards
	\begin{equation}\label{eq:phase_turn_bound}
		\Vec a(j)\ge g_j:=\max\{\kappa u,\Vec d(j)+(\kappa-1)u\},
		\qquad \Vec a(k)\ge g_k:=\Vec d(k).
	\end{equation}
	Fix its actual coefficients $\Vec\lambda_*$.
	Freeze occurrences from the fixed representation of $\Vec a$ into the base until $\overline{\Vec b}_A(j)\ge g_j$ and $\overline{\Vec b}_A(k)\ge g_k$.
	Choose $\Vec0\le\Vec\theta\le\Vec\lambda_*$ and set
	\begin{equation*}
		\overline{\Vec b}_A:=\Vec b_A+P_A\Vec\theta,
		\qquad \overline K_A:=K_A-\langle\Vec c_A,\Vec\theta\rangle.
	\end{equation*}
	Let $\overline A:=\overline{\Vec b}_A+P_A^{\langle\Vec c_A,\cdot\rangle\le\overline K_A}$.
	Since the full representation of $\Vec a$ satisfies both inequalities and all periods are nonnegative, a greedy choice needs at most $g_j+g_k$ occurrences.
	Hence
	\begin{equation*}
		\norm{\Vec\theta}_1\le g_j+g_k,
		\qquad
		\overline K_A\ge0,
		\qquad
		\norm{\overline{\Vec b}_A}\le\norm{\Vec b_A}+(g_j+g_k)\norm{P_A}.
	\end{equation*}
	Thus $\Vec a\in\overline A\subseteq A$.
	The base of $\overline A$ has norm polynomial in $D+\norm{\Vec x}+\kappa$, while its periods and costs satisfy the same bounds as those of $A$.
	For $\Vec w\in\overline A$ and $0\le t<\kappa$, these guards give
	\begin{equation*}
		\bigl(\Vec w+t\cdot\Eff(\beta)\bigr)(j)
		=\Vec w(j)-tu\ge\Vec d(j)+(\kappa-1-t)u\ge\Vec d(j).
	\end{equation*}
	The other coordinate satisfies $\bigl(\Vec w+t\cdot\Eff(\beta)\bigr)(k)=\Vec w(k)+tv\ge\Vec d(k)$.
	Hence every member enables $\beta^q$ for each $0\le q\le\kappa$, including all internal transitions of the cycle.

	\emph{Constructing the Translated Family.}
	Define the shifted source hybrid set
	\begin{equation*}
		A^-:=\overline A-\kappa u\cdot\Vec e_j,
		\qquad \Vec a^-:=\Vec a-\kappa u\cdot\Vec e_j\in A^-.
	\end{equation*}
	Its base is nonnegative by~\eqref{eq:phase_turn_bound}.
	By~\eqref{eq:phase_suffix_reserve}, subtracting $\kappa v\cdot\Vec e_k$ from $\eta_+$ gives a valid run from $\Vec a^-$ to $\Vec z-\kappa v\cdot\Vec e_k$.
	Apply~\cref{lem:two_counter_tools} with source $A^-$ to select an output hybrid set $G=\Vec b_G+P^{\langle\Vec c,\cdot\rangle\le K_G}$ containing that target.
	Its base has norm polynomial in $D+\norm{\Vec x}+\kappa$, while~\cref{lem:two_counter_tools} bounds its periods and costs without using the shifted source base.
	Define
	\begin{equation}
		H_s:=G+su\cdot\Vec e_j+(\kappa-s)v\cdot\Vec e_k,
	\end{equation}
	where $0\le s\le\kappa$.
	Then $H_s=\Vec b+s\cdot\Vec\delta+P^{\langle\Vec c,\cdot\rangle\le K_G}$, where $\Vec b:=\Vec b_G+\kappa v\cdot\Vec e_k$ and $\Vec\delta:=u\cdot\Vec e_j-v\cdot\Vec e_k$.
	Every base is nonnegative, $\Vec z\in H_0$, and the periods, costs, and cost bound are those of $G$.

	To verify the prefix witnesses, take $\Vec g\in G$ reached by a run $\rho_+$ through $\Lambda_+$ from some $\Vec w^-\in A^-$.
	The vector $\Vec w:=\Vec w^-+\kappa u\cdot\Vec e_j\in\overline A$ has a prefix $\rho_-$ through $\Lambda_-$ and satisfies the guards in~\eqref{eq:phase_turn_bound}.
	Execute $\kappa-s$ copies of $\beta$ from $\Vec w$ and translate the suffix witness by the nonnegative vector $su\cdot\Vec e_j+(\kappa-s)v\cdot\Vec e_k$.
	This gives
	\begin{equation*}
		\begin{gathered}
			\Vec x\xrightarrow{\rho_-}\Vec w
			\xrightarrow{\beta^{\kappa-s}}
			\Vec w^-+su\cdot\Vec e_j+(\kappa-s)v\cdot\Vec e_k
			\xrightarrow{\rho_+}
			\Vec g+su\cdot\Vec e_j+(\kappa-s)v\cdot\Vec e_k\in H_s.
		\end{gathered}
	\end{equation*}
	Thus every member of the family is reachable in the bounded-coordinate encoding.
\end{claimproof}

\subsection{Proof of~\texorpdfstring{\cref{lem:large_prefix_terminal}}{Lemma~\ref*{lem:large_prefix_terminal}}}
\label{app:large_prefix_terminal}

\phaseLargeTerminal*

\begin{proof}
	By~\cref{lem:large_prefix_dichotomy}, a fixed polynomial $f(D)$ bounds the periods of every extracted translated family for all $\kappa$.
	Use the constant target family $J_\ell=\{\Vec y\}$ for $0\le\ell\le\kappa$, with $P'=\varnothing$, $K'=0$, and $\Vec\delta'=\Vec0$.
	Every $J_\ell$ is backward $C$-coverable.
	As $\norm{P'}=0$, let $g$ be a fixed nondecreasing polynomial such that the threshold in~\cref{lem:translated_join} satisfies $\kappa_0\le g(D+\norm P)$.
	Choose
	\begin{equation*}
		\max\{2D,g(D+f(D))\}\le\kappa\le D^{\bigO(1)},
	\end{equation*}
	and then choose the sufficient $R$ from~\cref{lem:large_prefix_dichotomy}:
	\begin{equation*}
		D\le R\le(D+\kappa)^{\bigO(1)}\le D^{\bigO(1)}.
	\end{equation*}
	These choices are uniform over all suffixes.

	Apply~\cref{lem:large_prefix_dichotomy}.
	If $\Vec z\in H$ and $H$ has a nonzero free period, use~\cref{lem:free_covered} to choose $\Vec h\in H$ with polynomial norm and distance to $\Vec t$.
	Together with the prefix bound in~\eqref{eq:large_prefix_distance}, this gives
	\begin{align*}
		\MinLen_V(\Vec s,\Vec t)\le\MinLen_V\bigl(\Vec s,(\Vec r,\Vec h)\bigr)+\MinLen_V\bigl((\Vec r,\Vec h),\Vec t\bigr)\le(D+\norm{\Vec x}+\norm{\Vec t})^{\bigO(1)}.
	\end{align*}
	Otherwise $\Vec z\in H_0$ for a translated family from~\cref{lem:large_prefix_dichotomy}.
	Our choice gives $\kappa\ge\kappa_0$ and $\kappa\ge2D$.
	For the extracted direction $u\cdot\Vec e_j-v\cdot\Vec e_k$ with $u,v\ge1$, choose $a:=\lceil\max\{D-\Vec b(j),0\}/u\rceil\le D$.
	Then $\Vec b_a(j)\ge D$, while $\Vec b_\kappa\ge\Vec0$ gives $\Vec b_a(k)\ge(\kappa-a)v\ge D$.
	Since $\Vec\delta$ is nonzero precisely on coordinates $j,k$, the definition of a translated family implies that every period has zero effect on coordinate $i$.
	Together with $\Vec z\in H_0$, this gives $\Vec b_a(i)=\Vec z(i)$.
	The covering hypothesis makes $H_a$ forward $C$-coverable.
	Together with the given suffix and $J_0$, this proves $C$-joinability.
	By~\cref{lem:translated_join}, choose $\Vec h\in H_r$ with polynomial norm and distance to $\Vec t$, since the selected member of the constant target family is $\Vec y$.
	Append its run to the prefix supplied by~\eqref{eq:large_prefix_distance}, obtaining
	\begin{equation*}
		\MinLen_V(\Vec s,\Vec t)\le(D+\norm{\Vec x}+\norm{\Vec t})^{\bigO(1)}.
	\end{equation*}
\end{proof}

%% file: appendix/Appendix-F-General.tex
\section{Proofs for~\texorpdfstring{\cref{sec:removing_coverability}}{Section~\ref*{sec:removing_coverability}}}
\label{app:arbitrary_target_proofs}

\subsection{Proof of~\texorpdfstring{\cref{lem:pumpable_arbitrary}}{Lemma~\ref*{lem:pumpable_arbitrary}}}
\label{app:pumpable_arbitrary}

\pumpableArbitrary*

\begin{proof}
	Let $M:=D+\norm{\Vec s}+\norm{\Vec t}$ and apply \cref{lem:pump_extraction}.
	If the extracted prefix reaches $\Vec q(\Vec y)$, its pointwise length bound at $\Vec y$ proves the claim.
	Otherwise, it gives a forward $C$-coverable hybrid set $G=\Vec a+P^{\langle\Vec c,\cdot\rangle\le K}$ at state $\Vec r$, with $\Vec z\in G$ and $\Vec r(\Vec z)\xrightarrow{*}\Vec q(\Vec y)$ in the suffix of $U$, such that
	\begin{equation}\label{eq:terminal_second_hybrid}
		\hnorm(G)\le M^{\bigO(1)},\qquad
		\MinLen_U\bigl(\Vec p(\Vec x),\Vec r(\Vec g)\bigr)
		\le(M+\norm{\Vec g})^{\bigO(1)}
		\quad\text{for all }\Vec g\in G.
	\end{equation}

	Let $\widehat P:=\{(0,0,\Vec d):\Vec d\in P\}\subseteq\NN^5$.
	Consider the augmented VAS $V[\widehat P]$ and the suffix of $U$ augmented with a self-loop of effect $\Vec d$ at $\Vec r$ for every $\Vec d\in P$.
	The base $\Vec r(\Vec a)$ is forward $C$-coverable there and reaches $\Vec q(\Vec y)$ by first generating $\Vec z$.
	By \cref{cor:threshold_closure,thm:semi_coverable}, there is a run $\rho$ from $(\Vec r,\Vec a)$ to $\Vec t$ in $V[\widehat P]$ of length $L\le M^{\bigO(1)}$.
	Remove the added periods and handle the cost bounds as in the proof of~\cref{lem:joining_coverable_hybrids}, with singleton target $\{\Vec y\}$.
	If absorption is needed, it keeps $\Vec z$ in the selected hybrid set, increases the base, and leaves the free periods unchanged, so the base remains forward $C$-coverable by monotonicity.
	After at most one repetition of the preceding construction, this yields $\Vec g\in G$ with norm and distance to $\Vec t$ polynomial in $M$.
	Prepend its prefix witness from~\eqref{eq:terminal_second_hybrid}.
\end{proof}

\subsection{Proof of~\texorpdfstring{\cref{lem:hybrid_terminal}}{Lemma~\ref*{lem:hybrid_terminal}}}
\label{app:hybrid_terminal}

\hybridTerminal*

\begin{proof}
	Let $M:=D+\hnorm(H)+\hnorm(J)$ and write $J=\Vec b+P'^{\langle\Vec c',\cdot\rangle\le K'}$.
	Fix representations of the connected members $\Vec x\in H$ and $\Vec y\in J$, together with a witnessing run in $U$.
	In $W:=U[P,P']$, the positive period loops at $\Vec p$ generate $\Vec x$ from $\Vec a$, and the negative period loops at $\Vec q$ reduce $\Vec y$ to $\Vec b$.
	Thus $\Vec p(\Vec a)\xrightarrow{*}\Vec q(\Vec b)$ in $W$.
	Let $V^{\pm}$ be $V$ with the lifted transitions $(0,0,\Vec d)$ for $\Vec d\in P$ and $(0,0,-\Vec d')$ for $\Vec d'\in P'$ adjoined.
	The source condition remains valid in $W$, and~\cref{cor:threshold_closure} preserves $C$ for $V^{\pm}$.
	Since the period sets are subsets of $\NN^3$, their cardinalities and $\Size(V^{\pm})$ are polynomial in $M$.
	The structural threshold for the augmented instance can therefore be chosen at most $(B+C+\Size(V^{\pm}))^{\bigO(1)}\le M^{\bigO(1)}$.
	Apply~\cref{thm:semi_coverable} in the forward coverable case and~\cref{lem:pumpable_arbitrary} in the pumpable case to obtain
	\begin{equation*}
		(\Vec p,\Vec a)\xrightarrow{\rho}(\Vec q,\Vec b)
		\quad\text{in }V^{\pm},\qquad |\rho|\le L\le M^{\bigO(1)}.
	\end{equation*}

	Remove the added periods and handle the cost bounds as in the proof of~\cref{lem:joining_coverable_hybrids}.
	Each absorption preserves the fixed connected member on that side, increases its base, and leaves its free periods unchanged.
	Hence monotonicity preserves the source's forward $C$-coverability or pumpability in the system augmented with its free periods, so the corresponding short-run bound remains applicable.
	Each side is absorbed at most once, and the same polynomial estimates give members $\Vec h\in H$ and $\Vec h'\in J$ satisfying~\eqref{eq:hybrid_terminal_member}.
\end{proof}

\subsection{Proof of~\texorpdfstring{\cref{lem:hybrid_intersection}}{Lemma~\ref*{lem:hybrid_intersection}}}
\label{app:hybrid_intersection}

\phaseIntersection*

\begin{proof}
	Let $M:=D+\hnorm(H)+\hnorm(J)$ and write
	\begin{equation*}
		H=\Vec b+P^{\langle\Vec c,\cdot\rangle\le K},\qquad
		J=\Vec b'+P'^{\langle\Vec c',\cdot\rangle\le K'},\qquad
		\Vec z\in H\cap J.
	\end{equation*}
	Fix coefficients $\Vec\lambda_*,\Vec\mu_*$ representing $\Vec z$ in both hybrid sets and choose a componentwise-minimal solution below $(\Vec\lambda_*,\Vec\mu_*)$ of
	\begin{equation*}
		\begin{pmatrix}P&-P'\end{pmatrix}
		\binom{\Vec\lambda}{\Vec\mu}=\Vec b'-\Vec b,
		\qquad(\Vec\lambda,\Vec\mu)\in\NN^{|P|+|P'|}.
	\end{equation*}
	Any smaller solution would also lie below the fixed pair, so this solution is globally minimal.
	Nonnegative costs give
	\begin{equation*}
		\langle\Vec c,\Vec\lambda\rangle\le\langle\Vec c,\Vec\lambda_*\rangle\le K,
		\qquad\langle\Vec c',\Vec\mu\rangle\le\langle\Vec c',\Vec\mu_*\rangle\le K'.
	\end{equation*}
	Let $m:=|P|+|P'|$ and $A_0:=\max\{\norm P,\norm{P'}\}$.
	By \cref{thm:pottier}, the three-row system satisfies
	\begin{equation*}
		\norm{\Vec\lambda}_1+\norm{\Vec\mu}_1
		\le\bigl(1+mA_0+\norm{\Vec b'-\Vec b}\bigr)^3
		\le M^{\bigO(1)}.
	\end{equation*}
	Thus $\Vec h:=\Vec b+P\Vec\lambda=\Vec b'+P'\Vec\mu$ has the required properties.
\end{proof}

\subsection{Proof of~\texorpdfstring{\cref{lem:translated_arbitrary}}{Lemma~\ref*{lem:translated_arbitrary}}}
\label{app:translated_arbitrary}

\translatedArbitrary*

\begin{proof}
	We follow~\cref{thm:semi_coverable}, keeping $\mathcal J$ fixed.
	By~\cref{lem:large_prefix_dichotomy}, the periods of every translated family extracted from a source-side prefix have norm at most $D^{\bigO(1)}$, independently of the prescribed translation range.
	Thus~\cref{lem:translated_join} allows us to choose a uniform sufficient joining threshold $2D\le\kappa_0\le(D+\norm{P'})^{\bigO(1)}$ before choosing $R$.
	Assume $\kappa\ge\kappa_0$, restrict $\mathcal J$ to indices at most $\kappa_0$, and choose $D\le R\le(D+\kappa_0)^{\bigO(1)}$ by~\cref{lem:large_prefix_dichotomy}.
	The member $J_a$ remains available because $a\le D\le\kappa_0$, and its backward covering run stays in the terminal component of every suffix considered below.

	Let $M:=D+\hnorm(J_0)+\norm{\Vec\delta'}$.
	For $k\ge1$, let $L_k(N)$ bound both the norm of the selected member of $\mathcal J$ and its reachability distance in $V$ from a source of norm at most $N$, uniformly over the relevant sequential subinstances with at most $k$ components.
	Choose a fixed nondecreasing polynomial $F$ large enough for the bounds below.

	\emph{Base Case: $k=1$.}
	If $\Vec p(\Vec x)$ is forward $C$-coverable, apply~\cref{lem:hybrid_terminal} with $H=\{\Vec x\}$ and $J=J_0$.
	Otherwise,~\cref{lem:bounded_prefix} captures the entire witnessing run to $\Vec r'(\Vec z')$ in a $2$-VASS of size at most $D\Size(U)$, since there is no later component.
	By~\cref{lem:two_counter_tools}, its reachable vectors at the encoded target state include a hybrid set $H$ containing $\Vec z'$ with $\hnorm(H)\le(D+\norm{\Vec x})^{\bigO(1)}$, after restoring the bounded coordinate.
	Since $\Vec z'\in H\cap J_0$, \cref{lem:hybrid_intersection} selects a small $\Vec h'\in H\cap J_0$.
	Apply~\cref{lem:two_counter_tools} to reach this member in the same encoding and lift the run to $V$.
	Both cases give $L_1(N)\le F(M+N)$.

	For $k\ge1$, assume that $L_k$ has been constructed, and consider a witnessing run in a subinstance with at most $k+1$ components and $\norm{\Vec s}\le N$.
	When the source is not forward $C$-coverable, apply~\cref{lem:bounded_prefix} and write the witness as $\pi_0\pi_1$, where $\pi_0$ ends at $\Vec r(\Vec z)$ and some coordinate $i$ stays below $D$ throughout $\pi_0$.
	We distinguish four terminal cases and one inductive case.

	\emph{Terminal Cases.}
	\begin{enumerate}
		\item \emph{The source is forward $C$-coverable.}
		Apply~\cref{lem:hybrid_terminal} with $H=\{\Vec x\}$ and $J=J_0$.

		\item \emph{The bounded-coordinate prefix reaches $\Vec r'(\Vec z')$.}
		Here $\pi_1=\varepsilon$.
		Use~\cref{lem:two_counter_tools,lem:hybrid_intersection} to select and reach a small member of $H\cap J_0$, exactly as in the base case.

		\item \emph{An $R$-large prefix yields a nonzero free period.}
		Apply~\cref{lem:large_prefix_dichotomy} with range $\kappa_0$ to obtain a hybrid set $H$ with $\Vec z\in H$ and a nonzero free period.
		Its base is pumpable after adjoining the free periods.
		The unchanged suffix connects $H$ to $J_0$, so~\cref{lem:hybrid_terminal} selects small members $\Vec h\in H$ and $\Vec h'\in J_0$ connected by a short run in $V$.
		Prepend the run to $\Vec h$ supplied by~\eqref{eq:large_prefix_distance}.

		\item \emph{An $R$-large prefix yields a translated family.}
		The other outcome of~\cref{lem:large_prefix_dichotomy} gives $\mathcal H=\{H_r\}_{r=0}^{\kappa_0}$ with $\Vec z\in H_0$.
		Since $\kappa_0\ge2D$, the nonnegative bases and the stopping condition of~\cref{lem:bounded_prefix} give a forward $C$-coverable $H_b$ for some $b\le D$, by the same argument as in~\cref{app:large_prefix_terminal}.
		Together with the unchanged suffix from $\Vec r(\Vec z)$ to $\Vec r'(\Vec z')$ and the backward $C$-coverability of $J_a$, this proves $C$-joinability of $\mathcal H$ and $\mathcal J$ in that suffix subinstance.
		Apply~\cref{lem:translated_join} to select small members $\Vec h\in H_r$ and $\Vec h'\in J_\ell$ connected by a short run in $V$, and prepend the prefix from~\eqref{eq:large_prefix_distance}.
	\end{enumerate}
	In every terminal case, both $\norm{\Vec h'}$ and the concatenated length are at most $F(M+N)$.

	\emph{Inductive Case.}
	The remaining case has an $R$-small prefix $\pi_0$ ending at a later component entry $\Vec r(\Vec z)$.
	For $\Vec u:=(\Vec r,\Vec z)$, the same estimate as in~\eqref{eq:source_restart_invariant} gives $\norm{\Vec u}\le N+R$.
	By~\cref{lem:two_counter_tools}, replace $\pi_0$ by a run in $V$ to this same configuration of length at most $F(M+N+R)$.
	The unchanged suffix has at most $k$ components, reaches the same $\Vec z'\in J_0$, and preserves the backward covering run for $J_a$.
	Apply the induction hypothesis with the same target family and prepend the shortened prefix.
	The selected-member norm and total length therefore satisfy the recurrence in~\eqref{eq:one_covered_component_recurrence}, with $F$ in place of $P$.
	Since $k\le(B+1)^2$ and $R\le M^{\bigO(1)}$, the same additive-growth estimate as in~\cref{thm:semi_coverable} proves~\eqref{eq:translated_arbitrary_length}, with degree independent of $k$.
\end{proof}

\subsection{Proof of~\texorpdfstring{\cref{prop:encoding_reachability}}{Proposition~\ref*{prop:encoding_reachability}}}
\label{app:encoding_reachability}

\generalEncodingBound*

\begin{proof}
	Choose $C$ by~\cref{thm:generalized_transformation_area} and $D$ uniformly for sequential subinstances in both orientations by~\cref{lem:bounded_prefix}.
	The period bound in~\cref{lem:large_prefix_dichotomy} is independent of the prescribed range, so~\cref{lem:translated_arbitrary} gives a uniform sufficient cutoff $\kappa_0\le D^{\bigO(1)}$ for all extracted families in the reversed system.
	Choose $\max\{2D,\kappa_0\}\le\kappa\le D^{\bigO(1)}$, then choose $D\le R\le(D+\kappa)^{\bigO(1)}$ by the dichotomy.
	By~\cref{cor:threshold_closure}, the same $C$ is valid under reversal.
	We follow the additive-growth induction in the proof of~\cref{lem:translated_arbitrary}, now with the fixed target $\Vec q(\Vec y)$ and a witnessing run $\pi$.
	If the source is not forward $C$-coverable, use~\cref{lem:bounded_prefix} to write $\pi=\pi_0\pi_1$, with $\pi_0$ ending at $\Vec r(\Vec z)$.

	\emph{Terminal Cases.}
	\begin{enumerate}
		\setlength{\itemsep}{0pt}
		\item \emph{The source is forward $C$-coverable.}
		Apply~\cref{thm:semi_coverable}.

		\item \emph{The bounded-coordinate prefix reaches the target.}
		If $\pi_1=\varepsilon$, shorten $\pi_0$ in its $2$-VASS encoding by~\cref{lem:two_counter_tools} and lift the run to $V$.

		\item \emph{An $R$-large prefix yields a nonzero free period.}
		The dichotomy gives a hybrid set $H$ containing $\Vec z$, whose base is pumpable after adjoining the free periods.
		Apply~\cref{lem:hybrid_terminal} to $H$ and $\{\Vec y\}$ in the suffix subinstance, and prepend the run to its selected member supplied by~\eqref{eq:large_prefix_distance}.

		\item \emph{An $R$-large prefix yields a translated family.}
		Write $\mathcal H=\{H_r\}_{r=0}^{\kappa}$ for the extracted family, with $\Vec z\in H_0$.
		Since $\kappa\ge2D$, the stopping condition gives a forward $C$-coverable $H_b$ for some $b\le D$, as in~\cref{app:large_prefix_terminal}.
		In the reversed suffix, $H_b$ is backward $C$-coverable and $\Vec q(\Vec y)$ reaches $\Vec r(\Vec z)$ with $\Vec z\in H_0$.
		Thus~\cref{lem:translated_arbitrary} applies to this fixed target family and selects a polynomial-norm $\Vec h\in H_r$ with a short run from $\Vec t$ to $(\Vec r,\Vec h)$ in $-V$.
		Reverse that run and prepend the prefix from~\eqref{eq:large_prefix_distance}.
	\end{enumerate}
	The first two cases cover a single-component instance.
	The extracted hybrid norms and translation directions are polynomially bounded in $D+\norm{\Vec s}$, so every terminal case gives a run of length at most $(D+\norm{\Vec s}+\norm{\Vec t})^{\bigO(1)}$.

	The remaining $R$-small case is identical to the inductive step of~\cref{lem:translated_arbitrary}: the new source norm is at most $\norm{\Vec s}+R$, and the unchanged suffix has fewer components.
	Since there are at most $(B+1)^2$ components and $D,R\le(B+\Size(V))^{\bigO(1)}$, the same recurrence~\eqref{eq:one_covered_component_recurrence} proves~\eqref{eq:encoding_reachability_length}.
\end{proof}